\documentclass{lmcs} %
\pdfoutput=1
\usepackage[utf8]{inputenc}

\usepackage{lastpage}
\lmcsdoi{21}{2}{24}
\lmcsheading{}{\pageref{LastPage}}{}{}%
{Jul.~13,~2023}{Jun.~17,~2025}{}


\usepackage{hyperref}
\usepackage[textwidth=2.2cm,textsize=small]{todonotes}
\usepackage{tikz}
\usepackage{amssymb}
\usepackage{bbm}

\usepackage{algorithmicx}
\usepackage{algorithm}
\usepackage[noend]{algpseudocode}

\algdef{SE}[DOWHILE]{Do}{DoWhile}{\algorithmicdo}[1]{\algorithmicwhile\ #1}
\algrenewcommand{\algorithmiccomment}[1]{\hfill $\rhd$ \emph{#1}}%
\algrenewcommand{\algorithmicrequire}{\textbf{Input:}}
\algrenewcommand{\algorithmicensure}{\textbf{Output:}}
\algnewcommand{\OR}{\textbf{or }}
\algnewcommand{\AND}{\textbf{and }}
\algnewcommand{\Not}{\textbf{not}\,}
\algnewcommand{\True}{\textbf{true}}
\algnewcommand{\False}{\textbf{false}}

\usepackage[skins]{tcolorbox}
\newtcolorbox{myframe}[2][]{%
  enhanced,colback=white,colframe=black,coltitle=black,
  sharp corners,boxrule=0.4pt,
  fonttitle=\scshape,
  attach boxed title to top left={yshift=-0.3\baselineskip-0.4pt,xshift=2mm},
  boxed title style={tile,size=minimal,left=0.5mm,right=0.5mm,
    colback=white,before upper=\strut},
  title=#2,#1
}

\usetikzlibrary{arrows.meta}
\usetikzlibrary{tikzmark,decorations.pathreplacing}
\theoremstyle{plain} %

\newcommand\Lcal{\mathcal{L}}

\newcommand{\T}{\mathbb T}
\newcommand{\R}{\mathbb R}
\newcommand{\Z}{\mathbb Z}
\newcommand{\N}{\mathbb N}

\newcommand{\1}{\mathbbm 1}

\newcommand{\size}[1]{\ensuremath{\left|#1\right|}}

\newcommand{\ceil}[1]{\ensuremath{\left\lceil#1\right\rceil}}
\newcommand{\floor}[1]{\ensuremath{\left\lfloor#1\right\rfloor}}

\newcommand{\set}[1]{\ensuremath{\left\{#1\right\}}}
\newcommand{\pr}[1]{\ensuremath{\left(#1\right)}}

\DeclareMathOperator{\drop}{drop}

\begin{document}

\title[Strategy Iteration with Universal Trees]{Beyond Value Iteration for Parity Games:\texorpdfstring{\\}{}Strategy Iteration with Universal Trees\rsuper*}
\titlecomment{{\lsuper*}An extended abstract of this paper has appeared in Proceedings of the 47th International Symposium on Mathematical Foundations of Computer Science, MFCS 2022.}
\thanks{This project has received funding from the European Research Council (ERC) under the European Union's Horizon 2020 research and innovation programme (grant agreement nos.~757481--ScaleOpt and 805241--QIP). \\ Z.~K.~Koh---This work was done while the author was at the London School of Economics and Centrum Wiskunde \& Informatica.}	

\author[Z.~K.~Koh]{Zhuan Khye Koh\lmcsorcid{0000-0002-4450-8506}}[a]
\author[G.~Loho]{Georg Loho\lmcsorcid{0000-0001-6500-385X}}[b]

\address{Boston University, 665 Commonwealth Avenue, Boston, MA 02215, USA.}	%
\email{zkkoh@bu.edu}  %

\address{University of Twente, Drienerlolaan 5, 7522 NB Enschede, The Netherlands.}	%
\email{g.loho@utwente.nl}  %

\begin{abstract}
  \noindent Parity games have witnessed several new quasi-polynomial algorithms since the breakthrough result of Calude et al.~(STOC 2017).
  The combinatorial object underlying these approaches is a \emph{universal tree}, as identified by Czerwi{\'n}ski et al.~(SODA 2019).
  By proving a quasi-polynomial lower bound on the size of a universal tree, they have highlighted a barrier that must be overcome by all existing approaches to attain polynomial running time.
  This is due to the existence of worst case instances which force these algorithms to explore a large portion of the tree.

  As an attempt to overcome this barrier, we propose a strategy iteration framework which can be applied on any universal tree.
  It is at least as fast as its value iteration counterparts, while allowing one to take bigger leaps in the universal tree. 
  Our main technical contribution is an efficient method for computing the least fixed point of 1-player games. 
  This is achieved via a careful adaptation of shortest path algorithms to the setting of ordered trees.
  By plugging in the universal tree of Jurdzi{\'n}ski and Lazi{\'c}~(LICS 2017), or the Strahler universal tree of Daviaud et al.~(ICALP 2020), we obtain instantiations of the general framework that take time $O(mn^2\log n\log d)$ and $O(mn^2\log^3 n \log d)$ respectively per iteration.
\end{abstract}

\maketitle

\section{Introduction} \label{sec:intro}
A \emph{parity game} is an infinite duration game between two players Even and Odd.
It takes place on a sinkless directed graph $G=(V,E)$ equipped with a \emph{priority} function $\pi:V\rightarrow\{1,2,\dots,d\}$.
Let $n = |V|$ and $m=|E|$.
The node set $V$ is partitioned into $V_0\sqcup V_1$ such that nodes in $V_0$ and $V_1$ are owned by Even and Odd respectively.
The game starts when a token is placed on a node.
In each turn, the owner of the current node moves the token along an outgoing arc to the next node, resulting in an infinite walk.
If the highest priority occurring infinitely often in this walk is even, then Even wins.
Otherwise, Odd wins.

By the positional determinacy of parity games~\cite{Martin75,conf/focs/EmersonJ91}, there exists a partition of $V$ into two subsets from which Even and Odd can force a win respectively.
The main algorithmic problem of parity games is to determine this partition, or equivalently, to decide the winner given a starting node. 
This is a notorious problem that lies in $\text{NP}\cap\text{co-NP}$ \cite{conf/cav/EmersonJS93}, and also in $\text{UP}\cap\text{co-UP}$ \cite{journals/ipl/Jurdzinski98}, with no known polynomial algorithm to date.

Due to its intriguing complexity status, as well as its fundamental role in automata theory and logic \cite{conf/cav/EmersonJS93,conf/stoc/KupfermanV98}, parity games have been intensely studied over the past three decades. 
Prior to 2017, algorithms for solving parity games, e.g.~\cite{journals/tcs/Zielonka98,conf/stacs/Jurdzinski00,conf/cav/VogeJ00,conf/stacs/BjorklundSV03,conf/csl/Schewe08,journals/siamcomp/JurdzinskiPZ08,journals/jcss/Schewe17,journals/disopt/MnichRR18,journals/fmsd/BenerecettiDM18}, are either exponential or mildly subexponential.
In a breakthrough result, Calude et al.~\cite{conf/stoc/CaludeJKL017} gave the first quasi-polynomial algorithm.
Since then, many other quasi-polynomial algorithms \cite{journals/sttt/FearnleyJKSSW19,conf/lics/JurdzinskiL17,conf/lics/Lehtinen18,conf/mfcs/Parys19,journals/jcss/BenerecettiDMSW25} have been developed.
Most of them have been unified by Czerwi{\'n}ski et al.~\cite{conf/soda/CzerwinskiDFJLP19} via the concept of a \emph{universal tree}.
A universal tree is an ordered tree into which every ordered tree of a certain size can be isomorphically embedded.
They proved a quasi-polynomial lower bound on the size of a universal tree.

\paragraph{Value iteration}
The starting point of this paper is the classic \emph{progress measure} algorithm \cite{conf/stacs/Jurdzinski00,conf/lics/JurdzinskiL17} for solving parity games.
It belongs to a broad class of algorithms called \emph{value iteration} -- a well-known method for solving more general games on graphs such as mean payoff games and stochastic games.
In value iteration, every node $v$ in $G$ is assigned a value $\mu(v)\in \mathcal{V}$ from some totally ordered set $\mathcal{V}$, and the values are locally improved until we reach the \emph{least fixed point} of a set of operators associated with the game. 
The set $\mathcal{V}$ is called the \emph{value domain}, which is usually a bounded set of real numbers or integers.
For the progress measure algorithm, its value domain is the set of leaves $L(T)$ in a universal tree~$T$.
As the values are monotonically improved, the running time is proportional to $|L(T)|$.
The first progress measure algorithm of Jurdzi{\'n}ski \cite{conf/stacs/Jurdzinski00} uses a perfect $n$-ary tree, which runs in exponential time.
Its subsequent improvement by Jurdzi{\'n}ski and Lazi{\'c}~\cite{conf/lics/JurdzinskiL17} uses a quasi-polynomial-sized tree, which runs in $n^{\log (d/\log n) + O(1)}$ time.

Despite having good theoretical efficiency, the progress measure algorithm is not robust against its worst-case behaviour.
In fact, it is known to realize its worst-case running time on very simple instances.
As an example, let $(G,\pi)$ be an arbitrary instance with maximum priority $d$, with $d$ being even.
For a small odd constant $k$, if we add two nodes of priority $k$ as shown in Figure \ref{fig:vi_bad_example}, then the progress measure algorithm realizes its worst-case running time. 
This is because the values of these nodes are updated superpolynomially many times.

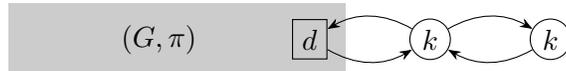
\begin{figure}[h!]
\centering
\begin{tikzpicture}[scale=0.8]
\fill[fill=gray!40] (-7,0.6) rectangle (-1.4,-0.6);
\node (G) at (-4.5,0) {$(G,\pi)$};
\begin{scope}[every node/.style={circle,draw,inner sep = 2pt}]
    \node (A) at (0,0) {$k$};
    \node (B) at (2,0) {$k$};   
\end{scope}
\begin{scope}[every node/.style={rectangle,draw}]
    \node (C) at (-2,0) {$d$};
\end{scope}

\begin{scope}[>={Stealth[black]},
              every edge/.style={draw}]
    \path [->] (A) edge [bend left] (B);
    \path [->] (A) edge [bend right] (C);
    \path [->] (B) edge [bend left] (A);   
    \path [->] (C) edge [bend right] (A);   
\end{scope}
\end{tikzpicture}
\caption{
  A worst-case construction for the progress measure algorithm.
  Nodes in $V_0$ and $V_1$ are drawn as squares and circles, respectively.
  }
\label{fig:vi_bad_example}
\end{figure}

\paragraph{Strategy iteration}
A different but related method for solving games on graphs is \emph{strategy iteration}.
For a parity game $(G,\pi)$, a \emph{(positional) strategy} $\tau$ for a player (say Odd) is a choice of an outgoing arc from every node in $V_1$. 
Removing the unchosen outgoing arcs from every node in $V_1$ results in a \emph{strategy subgraph} $G_\tau\subseteq G$.
A general framework for strategy iteration is given, e.g., in \cite{Friedmann:2011}. 
Following that exposition, one fixes a suitable value domain $\mathcal{V}$ and associates a valuation $\mu:V\rightarrow \mathcal{V}$ to each strategy for Odd.
The valuation of a strategy $\tau$ can be thought of as an evaluation of the best counterstrategy for Even. More formally, it is a fixed point of a set of operators associated with the 1-player game $(G_{\tau},\pi)$ for Even.
Given two strategies $\tau, \tau'$ and their respective valuations $\mu, \mu'$, we say that $\tau'$ is \emph{better} than $\tau$ if $\mu'(v)\geq \mu(v)$ for all $v\in V$ and $\mu'\neq \mu$.
In every iteration, strategy iteration maintains a strategy $\tau$ for Odd and its corresponding valuation $\mu$.
Based on $\mu$ and a \emph{pivot rule}, it switches to a better strategy $\tau'$, and then computes its valuation $\mu'$.
This process is repeated until we reach an optimal strategy for Odd.

Originally introduced by Hoffman and Karp for stochastic games~\cite{HoffmanKarp:1966}, variants of strategy iteration for parity games have been developed \cite{thesis/Puri95,conf/cav/VogeJ00,conf/stacs/BjorklundSV03,journals/dam/BjorklundV07,conf/csl/Schewe08,Luttenberger:2008}. %
They usually perform well in practice, but tedious constructions of their worst case (sub)exponential complexity are known~\cite{conf/lics/Friedmann09}.   
Motivated by the construction of quasi-polynomial universal trees \cite{conf/lics/JurdzinskiL17,conf/icalp/DaviaudJT20}, a natural question is whether there exists a strategy iteration algorithm with value domain $L(T)$ for a universal tree $T$.
It is not hard to see that with value domain $L(T)$, unfortunately, the fixed point of a 1-player game $(G_\tau,\pi)$ may not be unique.
Moreover, in a recent thesis \cite{Ohlmann:2021}, Ohlmann showed that a valuation that is fit for strategy iteration cannot be defined using $L(T)$ for the aforementioned trees \cite{conf/stacs/Jurdzinski00,conf/lics/JurdzinskiL17,conf/icalp/DaviaudJT20}.

\paragraph{Our contribution}
We show that an adaptation of strategy iteration with value domain $L(T)$ is still possible.
To circumvent the impossibility result of Ohlmann \cite{Ohlmann:2021}, we slightly alter the strategy iteration framework as follows.
After pivoting to a strategy $\tau'$ in an iteration, we update the current node labeling $\mu$ to the least fixed point of $(G_{\tau'},\pi)$ that is \emph{pointwise at least} $\mu$.
In other words, we force $\mu$ to increase (whereas this happens automatically in the previous framework).
Since the fixed point of a 1-player game may not be unique, this means that we may encounter a strategy more than once during the course of the algorithm.
The motivation of our approach comes from tropical geometry, as discussed in Appendix \ref{sec:mpg}.

To carry out each iteration efficiently, we give a combinatorial method for computing the least fixed point of 1-player games with value domain $L(T)$.
In the literature, valuations are usually computed using shortest path techniques, such as label-setting (e.g.~Dijkstra's algorithm) and label-correcting (e.g.~Bellman--Ford algorithm), with running times $O(m+n \log n)$ and $O(mn)$ respectively. %
We adapt these techniques to the setting of ordered trees.
When $T$ is instantiated as a specific universal tree constructed in the literature, we obtain the following running times: 
\begin{itemize}
  \item The perfect $n$-ary tree of height $d/2$ takes $O(d(m+n\log n))$.
  \item The universal tree of Jurdzi{\'n}ski and Lazi{\'c} \cite{conf/lics/JurdzinskiL17} takes $O(mn^2\log n\log d)$.
  \item The Strahler universal tree of Daviaud et al.~\cite{conf/icalp/DaviaudJT20} takes $O(mn^2\log^3 n \log d)$. 
\end{itemize}

The total number of strategy iterations is trivially bounded by $n|L(T)|$, the same bound for the progress measure algorithm.  
Whereas we do not obtain a strict improvement over previous running time bounds, it is conceivable that our algorithm would terminate in fewer iterations than the progress measure algorithm on most examples.
Moreover, our framework provides large flexibility in the choice of pivot rules.
Identifying a pivot rule that may provide strictly improved (and possibly even polynomial) running time is left for future research.

\subsection{Computing the Least Fixed Point of 1-Player Games}
\label{sec:techniques-Cramer-computation}

Let $(G_\tau,\pi)$ be a 1-player game for Even, and $\mu^*$ be its least fixed point with value domain $L(T)$ for some universal tree $T$.
Starting from $\mu(v)=\min L(T)$ for all $v\in V$, the progress measure algorithm successively lifts the label of a node based on the labels of its out-neighbours until $\mu^*$ is reached.
This procedure is not polynomial in general, even on 1-player games.
So, instead of approaching $\mu^*$ from below, we approach it from above.
This is reminiscent of shortest path algorithms, where node labels form upper bounds on the shortest path distances throughout the algorithm.
To compute shortest paths to a target node $t$, the label at $t$ is initialized to $0$, while the labels at all other nodes are initialized to $+\infty$. 
In a label-correcting method like the Bellman--Ford algorithm, the labels are monotonically decreased until they converge to shortest path distances.
Differently, in a label-setting method such as Dijkstra's algorithm, the label of a node is directly fixed to its shortest path distance in every iteration.
We refer to Ahuja et al.~\cite{AhujaMagnantiOrlin:1993} for an overview on label-correcting and label-setting techniques for the shortest path problem. 

There are two main challenges when applying these techniques to the setting of ordered trees.
First, the target node $t$ is not explicitly given.
Let us call a cycle \emph{even} if its maximum priority is even.
By inspecting the operators associated with $(G_\tau,\pi)$, one can deduce that for every node $v$, one gets $\mu^*(v)<\top$ only if $v$ can reach an even cycle in $G_\tau$ (Lemma~\ref{lem:tight_cycle}).
However, it is unclear which even cycle determines the value of $\mu^*(v)$ a priori.
To overcome this issue, we consider maximum priority nodes in even cycles of $G_\tau$, which we call \emph{base nodes}.
They will play the role of target nodes.

The second challenge is how to initialize the labels on the base nodes. When $T$ is a perfect $n$-ary tree, one can show that for every base node $w$, the label $\mu^*(w)$ is equal to the smallest leaf in $T$ (Theorem~\ref{thm:label-setting} and Remark~\ref{rem:label-setting}). This property of perfect $n$-ary trees allows us to compute $\mu^*$ using a label-setting method, with the initial node labeling $\nu$ given by $\nu(v) = \min L(T)$ if $v$ is a base node, and $\nu(v) = \top$ otherwise. Here, $\top$ is defined to be bigger than every element in $L(T)$ (so $\top$ is analogous to $+\infty$ for real numbers).

For the shortest path problem, Dijkstra's algorithm selects a node with the smallest label to be fixed in every iteration. Unfortunately, this criterion does not work for us because the representation of a parity game as a mean payoff game can have negative arc weights, and so the correctness of Dijkstra's algorithm is not guaranteed. To fix this issue, we define a potential function to guide the algorithm in selecting the correct node.
Let $H$ be the subgraph of $G_\tau$ obtained by deleting all the base nodes.
For $p\in \mathbb{N}$, let $H_p$ be the subgraph of $H$ induced by the nodes with priority at most $p$.
The potential function is constructed by interlacing the label on a node with a tuple that encodes its topological orders in $H_2,H_4,\dots$.
In every iteration, a node with the smallest potential is selected, and its label is fixed.

When $T$ is a general universal tree, we lose the property that $\mu^*(w) = \min L(T)$ for all base nodes $w$. Since we do not know $\mu^*(w)$, we resort to label-correcting techniques. It turns out that the Bellman--Ford algorithm has a less stringent requirement on the initial node labeling $\nu$. For every base node $w$, we identify a range $[\mu^*(w),\widehat\mu(w)]$ for $\nu(w)$ such that the Bellman--Ford algorithm is guaranteed to return $\mu^*$.

To understand the upper bound $\widehat\mu(w)$, consider the set of cycles in $G_\tau$ which contain $w$ as a maximum priority node.
Clearly, every such cycle $C$ induces a subgame $(C,\pi)$ on which Even wins because $\pi(w)$ is even.
The least fixed point of $(C,\pi)$ consists of leaves in $T$ which belong to a subtree $T_C$ of height $j:=\pi(w)/2$.
The upper bound $\widehat\mu(w)$ is given by a cycle $C$ with the `narrowest' $T_C$. Since the subtrees of height $j$ in $T$ form a poset $(\mathcal{T}_j,\sqsubseteq)$ with respect to the partial order of embeddability, the notion of `width' here is defined using a chain cover $\mathcal{C}_j$ of $(\mathcal{T}_j,\sqsubseteq)$.

To compute $\nu(w)\in [\mu^*(w),\widehat\mu(w)]$ for a base node $w$, we construct an arc-weighted auxiliary digraph $D$ on the set of base nodes.
Every arc $uv$ in $D$ represents a path from base node $u$ to base node $v$ in $G_\tau$.
We show that a minimum bottleneck cycle containing $w$ in $D$ gives rise to a cycle $C$ in $G_\tau$, such that the smallest leaf in $T_C$ lies in $[\mu^*(w),\widehat\mu(w)]$.
After computing $\nu(w)$ for all base nodes $w$, we set $\nu(v) = \top$ for the other nodes.

With this initial node labeling $\nu$, the Bellman--Ford algorithm returns $\mu^*$ in $O(mn)$ steps.
The overall running time is dominated by the computation of $\nu$, whose running time is proportional to the size of the chain cover $\mathcal{C}_j$.
We prove that the quasi-polynomial universal trees constructed in the literature \cite{conf/lics/JurdzinskiL17,conf/icalp/DaviaudJT20} admit small chain covers.
Using this result, we then give efficient implementations of our method for these trees.

\paragraph{Paper organization}
In Section \ref{sec:prelim}, we introduce notation and provide the necessary preliminaries on parity games and universal trees.
Section \ref{sec:strategy_iteration} contains our strategy iteration framework based on universal trees. 
Section \ref{sec:label-setting} gives a label-setting method for computing the least fixed point of 1-player games, and demonstrates its applicability to perfect $n$-ary trees.
In Section \ref{sec:label_correcting}, we develop a label-correcting method for the same task, and apply it to known constructions of quasi-polynomial universal trees.

\section{Preliminaries on Parity Games and Universal Trees}
\label{sec:prelim}

For $d\in \N$, let $[d] = \{1,2,\dots,d\}$.
For a graph $G$, we use $V(G)$ as its vertex set, and $E(G)$ as its edge set.
A parity game instance is given by $(G,\pi)$, where $G=(V,E)$ is a sinkless directed graph with $V=V_0\sqcup V_1$, and $\pi:V\rightarrow [d]$ is a priority function.
Recall that the players are Even and Odd, who control vertices in $V_0$ and $V_1$ respectively.
Without loss of generality, we may assume that $d$ is even.
In this paper, we are only concerned with positional strategies.
A \emph{strategy} for Odd is a function $\tau:V_1\rightarrow V$ such that $v\tau(v)\in E$ for all $v\in V_1$. 
Its \emph{strategy subgraph} is $G_\tau = (V,E_\tau)$, where $E_\tau := \set{vw\in E:v\in V_0} \cup \set{v\tau(v):v\in V_1}$.
A strategy for Even and its strategy subgraph are defined analogously.
We always denote a strategy for Even as $\sigma$, and a strategy for Odd as $\tau$.
If we fix a strategy $\tau$ for Odd, the resulting instance $(G_\tau,\pi)$ is a \emph{1-player game} for Even.

For the sake of brevity, we overload the priority function $\pi$ as follows.
Given a subgraph $H\subseteq G$, let $\pi(H)$ be the highest priority in $H$.
The subgraph $H$ is said to be \emph{even} if $\pi(H)$ is even, and \emph{odd} otherwise.
For a fixed $\pi$, we denote by $\Pi(H)$ the set of nodes with the highest priority in~$H$.
If $v\in \Pi(H)$, we say that $v$ \emph{dominates} $H$.
For $p\in [d]$, $H_p$ refers to the subgraph of $H$ induced by nodes with priority at most $p$.
For a node $v$, let $\delta^-_H(v)$ and $\delta^+_H(v)$ be the incoming and outgoing arcs of $v$ in $H$ respectively.
Similarly, let $N^-_H(v)$ and $N^+_H(v)$ be the in-neighbors and out-neighbors of $v$ in $H$ respectively.
When $H$ is clear from context, we will omit it from the subscripts.

The win of a player can be certified by \emph{node labels} from a \emph{universal tree}, as stated in Theorem~\ref{thm:feasible_pg}. 
We give the necessary background for this now.

\subsection{Ordered Trees and Universal Trees}

An \emph{ordered tree} $T$ is a prefix-closed set of tuples, whose elements are drawn from a linearly ordered set $M$.
The linear order of $M$ lexicographically extends to $T$.  
We note that any proper prefix of a tuple is smaller than the tuple itself.
Equivalently, $T$ can be thought of as a rooted tree, whose root we denote by $r$.
Under this interpretation, elements in $M$ correspond to the branching directions at each vertex of $T$ (see Figures~\ref{fig:perfect_tree} and~\ref{fig:succinct_tree} for examples). 
Every tuple then corresponds to a vertex $v\in V(T)$.
This is because the tuple can be read by traversing the unique $r$-$v$ path in $T$.
Observe that $v$ is an $h$-tuple if and only if $v$ is at depth $h$ in $T$.
In particular, $r$ is the empty tuple.

In this paper, we always use the terms `vertex' and `edge' when referring to an ordered tree $T$.
The terms `node' and `arc' are reserved for the game graph $G$.

Given an ordered tree $T$ of height $h$, let $L(T)$ be the set of leaves in $T$.
For convenience, we assume that every leaf in $T$ is at depth $h$ throughout.
The tuple representing a leaf $\xi\in L(T)$ is denoted as $\xi = (\xi_{2h-1}, \xi_{2h-3}, \dots, \xi_1)$, where $\xi_i\in M$ for all $i$.
We refer to $\xi_{2h-1}$ as the \emph{first} component of $\xi$, even though it has index $2h-1$. 
For a fixed $p\in [2h]$, the \emph{$p$-truncation} of $\xi$ is
\[\xi|_p := \begin{cases}
    (\xi_{2h-1}, \xi_{2h-3}, \dots, \xi_{p+1}), &\text{if $p$ is even}\\
    (\xi_{2h-1}, \xi_{2h-3}, \dots, \xi_p), &\text{if $p$ is odd.}
\end{cases}\]
In other words, the $p$-truncation of a tuple is obtained by deleting the components with index less than $p$.
Note that a truncated tuple is an ancestor of the untruncated tuple in $T$.

\begin{defi} \label{def:embedding-order}
  Given ordered trees $T$ and $T'$, we say that $T$ \emph{embeds into} $T'$ (denoted $T\sqsubseteq T'$) if there exists an \emph{injective} and \emph{order-preserving} homomorphism from $T$ to $T'$ such that leaves in $T$ are mapped to leaves in $T'$. 
  Formally, this is an injective function $f:V(T)\rightarrow V(T')$ which satisfies the following properties:
\begin{enumerate}
    \item For all $u,v\in V(T)$, $uv\in E(T)$ implies $f(u)f(v)\in E(T')$;
    \item For all $u,v\in V(T)$, $u\leq v$ implies $f(u)\leq f(v)$.
    \item $f(u) \in L(T')$ for all $u \in L(T)$. 
\end{enumerate} 
We write $T \equiv T'$ if $T\sqsubseteq T'$ and $T'\sqsubseteq T$.
Also, $T\sqsubset T'$ if $T\sqsubseteq T'$ and $T\not\equiv T'$.
\end{defi}

In the definition above, since $f$ is order-preserving, the children of every vertex in $T$ are mapped to the children of its image injectively such that their order is preserved.
As an example, the tree in Figure \ref{fig:succinct_tree} embeds into the tree in Figure \ref{fig:perfect_tree}.
It is easy to verify that $\sqsubseteq$ is a partial order on the set of all ordered trees.

\begin{defi}
An \emph{$(\ell,h)$-universal tree} is an ordered tree $T'$ of height $h$ such that $T\sqsubseteq T'$ for every ordered tree $T$ of height $h$ and with at most $\ell$ leaves, all at depth exactly $h$.
\end{defi}

The simplest example of an $(\ell,h)$-universal tree is the perfect $\ell$-ary tree of height $h$, which we call a \emph{perfect universal tree}. 
The linearly ordered set $M$ for this tree can be chosen as $\set{0,1,\dots,\ell-1}$ (see Figure \ref{fig:perfect_tree} for an example).
It has $\ell^h$ leaves, which grows exponentially with $h$. 
Jurdzi\'{n}ski and Lazi\'{c} \cite{conf/lics/JurdzinskiL17} constructed an $(\ell,h)$-universal tree with at most $\ell^{\log h + O(1)}$ leaves, which we call a \emph{succinct universal tree}. 
In this tree, every leaf $\xi$ corresponds to an $h$-tuple of binary strings with at most $\floor{\log(\ell)}$ bits in total\footnote{A slightly looser bound of $\ceil{\log\ell}$ was derived in \cite[Lemma 1]{conf/lics/JurdzinskiL17}. It can be strengthened to $\floor{\log \ell}$ with virtually no change in the proof.}.
We use $\size{\xi}$ and $\size{\xi_i}$ to denote the total number of bits in $\xi$ and $\xi_i$ respectively.
The linearly ordered set $M$ for this tree consists of finite binary strings, where $\varepsilon\in M$ is the empty string (see Figure \ref{fig:succinct_tree} for an example).
For any pair of binary strings $s,s'\in M$ and a bit $b$, the linear order on $M$ is defined as $0s<\varepsilon<1s'$ and $bs<bs' \iff s<s'$.

\begin{figure}[h!]
\def\x{0.7}
\def\y{1.4}
\begin{minipage}{0.49\textwidth}
\centering
\begin{tikzpicture}[scale=0.9]
\begin{scope}[every node/.style={circle,draw, inner sep=2pt}]
    \node (p) at (0,0) {};
    \node (p1) at (-3*\x,-1*\y) {};   
    \node (p2) at (0,-1*\y) {};
    \node (p3) at (3*\x,-1*\y) {};
    \node (p11) at (-3*\x-\x,-2*\y) {};   
    \node (p12) at (-3*\x,-2*\y) {};
    \node (p13) at (-3*\x+\x,-2*\y) {};
    \node (p21) at (0-\x,-2*\y) {};   
    \node (p22) at (0,-2*\y) {};
    \node (p23) at (0+\x,-2*\y) {};
    \node (p31) at (3*\x-\x,-2*\y) {};   
    \node (p32) at (3*\x,-2*\y) {};
    \node (p33) at (3*\x+\x,-2*\y) {};
\end{scope}

\begin{scope}[every node/.style={fill=white,circle,font=\footnotesize,inner sep=2pt},
              every edge/.style={draw}]
    \path [-] (p) edge node {0} (p1);
    \path [-] (p) edge node {1} (p2);
    \path [-] (p) edge node {2} (p3);
    \path [-] (p1) edge node {0} (p11);
    \path [-] (p1) edge node {1} (p12);
    \path [-] (p1) edge node {2} (p13);
    \path [-] (p2) edge node {0} (p21);
    \path [-] (p2) edge node {1} (p22);
    \path [-] (p2) edge node {2} (p23);
    \path [-] (p3) edge node {0} (p31);
    \path [-] (p3) edge node {1} (p32);
    \path [-] (p3) edge node {2} (p33);
\end{scope}
\end{tikzpicture}
\caption{The perfect (3,2)-universal tree.}
\label{fig:perfect_tree}
\end{minipage}
\begin{minipage}{0.49\textwidth}
\centering
\begin{tikzpicture}[scale=0.9]
\begin{scope}[every node/.style={circle,draw, inner sep=2pt}]
    \node (p) at (0,0) {};
    \node (p1) at (-3*\x,-1*\y) {};   
    \node (p2) at (0,-1*\y) {};
    \node (p3) at (3*\x,-1*\y) {}; 
    \node (p12) at (-3*\x,-2*\y) {};
    \node (p21) at (0-\x,-2*\y) {};   
    \node (p22) at (0,-2*\y) {};
    \node (p23) at (0+\x,-2*\y) {};
    \node (p32) at (3*\x,-2*\y) {};
\end{scope}

\begin{scope}[every node/.style={fill=white,circle,font=\footnotesize,inner sep=2pt},
              every edge/.style={draw}]
    \path [-] (p) edge node {0} (p1);
    \path [-] (p) edge node {$\varepsilon$} (p2);
    \path [-] (p) edge node {1} (p3);
    \path [-] (p1) edge node {$\varepsilon$} (p12);
    \path [-] (p2) edge node {0} (p21);
    \path [-] (p2) edge node {$\varepsilon$} (p22);
    \path [-] (p2) edge node {1} (p23);
    \path [-] (p3) edge node {$\varepsilon$} (p32);
\end{scope}
\end{tikzpicture}
\caption{The succinct (3,2)-universal tree.}
\label{fig:succinct_tree}
\end{minipage}
\end{figure}

\subsection{Node Labelings from Universal Trees}

Let $(G,\pi)$ be a parity game instance and $T$ be an ordered tree of height $d/2$.
We augment the set of leaves with an extra \emph{top} element~$\top$, denoted $\bar{L}(T):=L(T)\cup\set{\top}$, such that $\top > v$ for all $v\in V(T)$.
In this way, $\top$ is analogous to $+\infty$ for real numbers.
We also set 
$\top|_p := \top$ for all $p\in [d]$.
A function $\mu:V\rightarrow \bar{L}(T)$ which maps the nodes in $G$ to $\bar{L}(T)$ is called a \emph{node labeling}. 
The node labeling is \emph{finite} if $\mu(v)\neq\top$ for all $v\in V$.

For a subgraph $H$ of $G$, we say that $\mu$ is \emph{feasible in $H$} if there exists a strategy $\sigma:V_0\rightarrow V$ for Even with $v\sigma(v)\in E(H)$ whenever $\delta^+_H(v)\neq \emptyset$, such that the following condition holds for every arc $vw$ in $H\cap G_{\sigma}$:
\begin{itemize}
    \item If $\pi(v)$ is even, then $\mu(v)|_{\pi(v)} \geq \mu(w)|_{\pi(v)}$.
    \item If $\pi(v)$ is odd, then $\mu(v)|_{\pi(v)} > \mu(w)|_{\pi(v)}$ or $\mu(v) = \mu(w) = \top$.
\end{itemize}
An arc $vw$ which does not satisfy the condition above is called \emph{violated} (with respect to $\mu$).
On the other hand, we distinguish two possibilities for non-violated arcs. 
If $\mu(v)$ is the smallest element in $\bar{L}(T)$ such that $vw$ is not violated, then $vw$ is said to be \emph{tight}.
Any arc which is neither tight nor violated is called \emph{loose}.
We remark that loose does not mean non-tight.
We say that a subgraph is \emph{tight} if it consists of tight arcs.

In the literature, a node labeling which is feasible in $G$ is also called a \emph{progress measure}.
The node labeling given by $\mu(v) = \top$ for all $v\in V$ is trivially feasible in $G$.
However, we are primarily interested in progress measures with maximal finite support, i.e.~the set of nodes with non-top labels is inclusion-wise maximal.

\begin{thmC}[{\cite{conf/stacs/Jurdzinski00}}]\label{thm:feasible_pg}
Given an $(n,d/2)$-universal tree $T$, let $\mu^*:V\rightarrow \bar{L}(T)$ be a node labeling which is feasible in $G$ and has maximal finite support.
Then, Even wins from $v\in V$ if and only if $\mu^*(v) \neq \top$.
\end{thmC}

The above theorem formalizes the following intuition: nodes with smaller labels are more advantageous for Even to play on.
Note that if $\mu$ is a minimal node labeling which is feasible in $G$, i.e.~$\mu'$ is infeasible in $G$ for all $\mu'<\mu$, then there exists a strategy $\sigma$ for Even such that $v\sigma(v)$ is tight for all $v\in V_0$.
The next observation is well-known (see, e.g.,~\cite[Lemma~2]{conf/lics/JurdzinskiL17}) and follows directly from the definition of feasibility.

\begin{lem}[Cycle Lemma]\label{lem:tight_cycle}
  If a node labeling $\mu$ is feasible in a cycle $C$, then
  \begin{align*}
    \mu(v)|_{\pi(C)} = \mu(w)|_{\pi(C)} \text{ for all } v,w\in V(C) \enspace .
  \end{align*}
Furthermore, if $\mu(v) \neq \top$ for some $v\in V(C)$, then $C$ is even.
\end{lem}

We assume to have access to the following algorithmic primitive \textsc{Tighten}, whose running time we denote by $\gamma(T)$.
Its implementation depends on the ordered tree $T$.
For instance, $\gamma(T) = O(d)$ if $T$ is a perfect $(n,d/2)$-universal tree.
If $T$ is a succinct $(n,d/2)$-universal tree, Jurdzi\'{n}ski and Lazi\'{c} \cite[Theorem 7]{conf/lics/JurdzinskiL17} showed that $\gamma(T) = O(\log n \log d)$.

\begin{myframe}{Tighten{$(\mu,vw)$}}
  Given a node labeling $\mu:V\rightarrow \bar{L}(T)$ and an arc $vw\in E$, return the unique element $\xi\in \bar{L}(T)$ such that $vw$ is tight after setting $\mu(v)$, the label of $v$, to $\xi$.
\end{myframe}

Given a node labeling $\mu:V\rightarrow \bar{L}(T)$ and an arc $vw\in E$, let $\text{lift}(\mu,vw)$ be the smallest element $\xi\in \bar{L}(T)$ such that $\xi\geq \mu(v)$ and $vw$ is not violated after setting $\mu(v)$ to $\xi$.
Observe that if $vw$ is violated, $\text{lift}(\mu,vw)$ is given by {\sc Tighten}$(\mu,vw)$.
Otherwise, it is equal to $\mu(v)$.
Hence, it can be computed in $\gamma(T)$ time.

\subsection{Fixed Points in Lattices}

We recall a fundamental result on the existence of fixed points in lattices.
Let $\Lcal$ be a non-empty finite lattice, and let $\mu, \nu \in \Lcal$.
An operator $\phi \colon \Lcal \to \Lcal$ is \emph{monotone} if $\mu \leq \nu \Rightarrow \phi(\mu) \leq \phi(\nu)$, and it is \emph{inflationary} if $\mu \leq \phi(\mu)$.
Given a family $\mathcal{G}$ of inflationary monotone operators on $\Lcal$, we denote $\mu^{\mathcal{G}}$ as the \emph{least} (simultaneous) fixed point of ${\mathcal{G}}$ which is \emph{at least} $\mu$.
The following proposition, which can be seen a variant of the Knaster--Tarski theorem, guarantees the existence of $\mu^\mathcal{G}$.
We remark that $\mu^{\mathcal{G}}$ may not exist if the inflationary assumption is removed.

\begin{prop}\label{prop:Tarski}
Let $\mathcal{G}$ be a family of inflationary monotone operators on a non-empty finite lattice $\mathcal{L}$. For any $\mu\in \mathcal{L}$ and $\mathcal{H}\subseteq \mathcal{G}$,
  \begin{enumerate}[label=(\roman*)]
  \item The least fixed point $\mu^{\mathcal{H}}$ exists. It can be obtained by applying the operators in $\mathcal{H}$ to $\mu$ in any order until convergence.
  \item The least fixed point is non-decreasing with the set of operators: $\mu^{\mathcal{H}} \leq \mu^{\mathcal{G}}$.
  \item The least fixed point is monotone: if $\mu \leq \nu$ then $\mu^{\mathcal{H}} \leq \nu^{\mathcal{H}}$. 
  \end{enumerate} 
\end{prop}

\begin{proof}
For part (i), let $\mu \eqqcolon \mu_1 < \mu_2 < \dots < \mu_k$ be a maximal sequence obtained by applying the operators in $\mathcal{H}$ to $\mu$; this sequence is finite because $\mathcal{L}$ is finite.
By maximality, $\mu_k$ is a simultaneous fixed point of $\mathcal{H}$.
To show minimality of this fixed point $\mu_k$, let $\mu'\geq \mu$ be an arbitrary simultaneous fixed point of $\mathcal{H}$.
For each $i<k$, let $\phi_i\in \mathcal{H}$ where $\phi_i(\mu_i) = \mu_{i+1}$.
Then, by monotonicity and induction, $\mu' = (\phi_i \circ \cdots \circ \phi_1)(\mu') \geq \phi_i(\mu_i) = \mu_{i+1}$ for all $i<k$.
In particular, $\mu'\geq \mu_k$, so $\mu_k = \mu^{\mathcal{H}}$.

Part (ii) holds by definition because any simultaneous fixed point of $\mathcal{G}$ is a simultaneous fixed point of $\mathcal{H}$.
Part (iii) holds again by definition because $\nu^{\mathcal{H}}$ is a simultaneous fixed point of $\mathcal{H}$ which is at least $\mu$.
\end{proof}

An operator $\psi \colon \Lcal \to \Lcal$ is \emph{deflationary} if $\mu \geq \psi(\mu)$. 
Considering the lattice through an order-reversing poset isomorphism implies that the analogous statements of Proposition~\ref{prop:Tarski} hold for deflationary monotone operators.
Given a family $\mathcal{G}$ of deflationary monotone operators on $\mathcal{L}$, we denote $\mu^{\mathcal{G}}$ as \emph{greatest} (simultaneous) fixed point of $\mathcal{G}$ which is \emph{at most} $\mu$.

In this paper, $\mathcal{L}$ is the finite lattice of node labelings mapping $V$ to $\bar{L}(T)$ for a universal tree $T$. 
The partial order of $\mathcal{L}$ is induced by the total order of $\bar{L}(T)$.
For a subgraph $H\subseteq G$, consider the following operators. %
For every node $v\in V_0$, define $\text{Lift}_v:\mathcal{L}\times V\rightarrow \bar{L}(T)$ as
\[\text{Lift}_v(\mu,u):=\begin{cases}
    \min_{vw\in E(H)}\text{lift}(\mu,vw), &\text{ if }u=v\text{ and }\delta^+_H(v)\neq \emptyset\\
    \mu(u), &\text{ otherwise.}
\end{cases}\]
For every arc $vw\in E(H)$ where $v\in V_1$, define $\text{Lift}_{vw}:\mathcal{L} \times V\rightarrow \bar{L}(T)$ as
\[\text{Lift}_{vw}(\mu,u):=\begin{cases}
    \text{lift}(\mu,vw), &\text{ if }u=v\\
    \mu(u), &\text{ otherwise.}
\end{cases}\]
We denote $\mathcal{H}^\uparrow:=\{\text{Lift}_v:v\in V_0\}\cup\{\text{Lift}_{vw}:v\in V_1\}$ as the operators in $H$.
Since they are inflationary and monotone, for any $\mu\in \mathcal{L}$, the least fixed point $\mu^{\mathcal{H}^\uparrow}$ exists by Proposition~\ref{prop:Tarski}(i).
Note that a node labeling is a fixed point of $\mathcal{H}^\uparrow$ if and only if it is feasible in $H$.
In this terminology, the \emph{progress measure algorithm} \cite{conf/stacs/Jurdzinski00,conf/lics/JurdzinskiL17} is an iterative application of the operators in $\mathcal{G}^\uparrow$ to $\mu$ to obtain $\mu^{\mathcal{G}^\uparrow}$.

\section{Strategy Iteration with Tree Labels}
\label{sec:strategy_iteration}

In this section, we present a strategy iteration algorithm (Algorithm \ref{algo:strategy+iteration}) whose pivots are guided by a universal tree.
It takes as input an instance $(G,\pi)$, a universal tree $T$, and an initial strategy $\tau_1$ for Odd.
Throughout, it maintains a node labeling $\mu:V\rightarrow \bar{L}(T)$, initialized as the least simultaneous fixed point of $\mathcal{G}^\uparrow_{\tau_1}$.
At the start of every iteration, the algorithm maintains a strategy $\tau$ for Odd such that the following two invariants hold: (1) $\mu$ is feasible in $G_\tau$; and (2) there are no loose arcs in $G_\tau$ with respect to $\mu$.
So, every arc in $G_\tau$ is either tight (usable by Even in her counterstrategy $\sigma$) or violated (not used by Even).
Note that our initial node labeling satisfies these conditions with respect to $\tau_1$.

\begin{algorithm}[htbp]
  \caption{Strategy iteration with tree labels: $(G,\pi)$ instance, $T$ universal tree, $\tau_1$ initial strategy for Odd}
  \label{algo:strategy+iteration}
  \begin{algorithmic}[1]
    \Procedure{StrategyIteration}{$(G,\pi),T,\tau_1$}
    \State $\mu(v) \gets \min L(T) \;\forall v\in V$
    \State $\tau \gets \tau_1$, $\mu\gets \mu^{\mathcal{G}^\uparrow_\tau}$
    \While{$\exists$ an admissible arc in $G$ with respect to $\mu$}
      \State Pivot to a strategy $\tau'$ by selecting admissible arc(s) \Comment{requires a pivot rule}
      \State $\tau \gets \tau'$, $\mu\gets \mu^{\mathcal{G}^\uparrow_\tau}$
    \EndWhile
    \State \Return  $\tau$, $\mu$
    \EndProcedure
  \end{algorithmic}
\end{algorithm}

For $v \in V_1$, we call a violated arc $vw\in E$ with respect to $\mu$ \emph{admissible} (as it admits Odd to perform an improvement).
If there are no admissible arcs in $G$, then the algorithm terminates.
In this case, $\mu$ is feasible in $G$. 
Otherwise, Odd pivots to a new strategy $\tau'$ by switching to admissible arc(s).
The choice of which admissible arc(s) to pick is governed by a \emph{pivot rule}.
Then, $\mu$ is updated to $\mu^{\mathcal{G}^\uparrow_{\tau'}}$.
Note that both invariants continue to hold at the start of the next iteration. Indeed, there are no loose arc in $G_\tau$ with respect to $\mu$ due to the minimality of $\mu^{\mathcal{G}^\uparrow_{\tau'}}$.

We remark that a strategy $\tau$ may occur more than once during the course of the algorithm, as mentioned in the description of strategy iteration in Section~\ref{sec:intro}.
This is because the fixed points of $\mathcal{G}^\uparrow_\tau$ are not necessarily unique.
See Figure \ref{fig:iterations} for an example run with a perfect universal tree and a succinct universal tree.

\begin{figure}[ht] 
\begin{minipage}{0.33\textwidth}
\centering
\begin{tikzpicture}[scale=0.8]
\node (left) at (-5,1) {};
\node (above) at (0,2.9) {};
\node (below) at (0,-.9) {};
\node (l1) at (-4.5,1) {\footnotesize $(0,0)$};
\node (l2) at (0,2.65) {\footnotesize $(0,0)$};

\begin{scope}[every node/.style={circle,draw,inner sep = 2pt}]
    \node (A) at (0,2) {1};    
    \node (E) at (-3.7,1) {2};
\end{scope}
\begin{scope}[every node/.style={rectangle,draw}]
    \node[label={below}:{\footnotesize $(0,0)$}] (B) at (0,0) {1};    
    \node[label={below}:{\footnotesize $(0,0)$}] (C) at (-2,0) {3};
    \node[label={above}:{\footnotesize $(0,0)$}] (D) at (-2,2) {4};
\end{scope}

\begin{scope}[>={Stealth[black]},
              every edge/.style={draw}]
    \path [->] (B) edge [bend right] (A);
    \path [->] (C) edge (D);
    \path [->] (D) edge [bend right] (E);
    \path [->] (E) edge [bend right] (C);
    \path [->] (C) edge [bend right] (E);
    \path [->] (C) edge (B);

    \path [->] (A) edge [bend right] (B);
    \path [->] (A) edge (D);
\end{scope}
\end{tikzpicture}

\begin{tikzpicture}[scale=0.8]
\node (left) at (-5,1) {};
\node (above) at (0,2.9) {};
\node (below) at (0,-.9) {};
\node (l1) at (-4.5,1) {\footnotesize $(0,\varepsilon)$};
\node (l2) at (0,2.65) {\footnotesize $(0,\varepsilon)$};

\begin{scope}[every node/.style={circle,draw,inner sep = 2pt}]
    \node (A) at (0,2) {1};    
    \node (E) at (-3.7,1) {2};
\end{scope}
\begin{scope}[every node/.style={rectangle,draw}]
    \node[label={below}:{\footnotesize $(0,\varepsilon)$}] (B) at (0,0) {1};    
    \node[label={below}:{\footnotesize $(0,\varepsilon)$}] (C) at (-2,0) {3};
    \node[label={above}:{\footnotesize $(0,\varepsilon)$}] (D) at (-2,2) {4};
\end{scope}

\begin{scope}[>={Stealth[black]},
              every edge/.style={draw}]
    \path [->] (B) edge [bend right] (A);
    \path [->] (C) edge (D);
    \path [->] (D) edge [bend right] (E);
    \path [->] (E) edge [bend right] (C);
    \path [->] (C) edge [bend right] (E);
    \path [->] (C) edge (B);

    \path [->] (A) edge [bend right] (B);
    \path [->] (A) edge (D);
\end{scope}
\end{tikzpicture}
\end{minipage}
\begin{minipage}{0.33\textwidth}
\centering
\begin{tikzpicture}[scale=0.8]
\node (left) at (-5,1) {};
\node (above) at (0,2.9) {};
\node (below) at (0,-.9) {};
\node (l1) at (-4.5,1) {\footnotesize $(1,0)$};
\node (l2) at (0,2.65) {\footnotesize $(0,1)$};

\begin{scope}[every node/.style={circle,draw,inner sep = 2pt}]
    \node (A) at (0,2) {1};    
    \node (E) at (-3.7,1) {2};
\end{scope}
\begin{scope}[every node/.style={rectangle,draw}]
    \node[label={below}:{\footnotesize $(0,2)$}] (B) at (0,0) {1};    
    \node[label={below}:{\footnotesize $(1,0)$}] (C) at (-2,0) {3};
    \node[label={above}:{\footnotesize $(0,0)$}] (D) at (-2,2) {4};
\end{scope}

\begin{scope}[>={Stealth[black]},
              every edge/.style={draw}]
    \path [->] (B) edge [bend right] (A);
    \path [->] (C) edge (D);
    \path [->] (D) edge [bend right] (E);
    \path [->] (E) edge [bend right] (C);
    \path [->] (C) edge [bend right] (E);
    \path [->] (C) edge (B);

    \path [->, opacity=.2] (A) edge [bend right] (B);
    \path [->] (A) edge (D);
\end{scope}
\end{tikzpicture}

\begin{tikzpicture}[scale=0.8]
\node (left) at (-5,1) {};
\node (above) at (0,2.9) {};
\node (below) at (0,-.9) {};
\node (l1) at (-4.5,1) {\footnotesize $(\varepsilon,0)$};
\node (l2) at (0,2.65) {\footnotesize $(\varepsilon,0)$};

\begin{scope}[every node/.style={circle,draw,inner sep = 2pt}]
    \node (A) at (0,2) {1};    
    \node (E) at (-3.7,1) {2};
\end{scope}
\begin{scope}[every node/.style={rectangle,draw}]
    \node[label={below}:{\footnotesize $(\varepsilon,\varepsilon)$}] (B) at (0,0) {1};    
    \node[label={below}:{\footnotesize $(\varepsilon,0)$}] (C) at (-2,0) {3};
    \node[label={above}:{\footnotesize $(0,\varepsilon)$}] (D) at (-2,2) {4};
\end{scope}

\begin{scope}[>={Stealth[black]},
              every edge/.style={draw}]
    \path [->] (B) edge [bend right] (A);
    \path [->] (C) edge (D);
    \path [->] (D) edge [bend right] (E);
    \path [->] (E) edge [bend right] (C);
    \path [->] (C) edge [bend right] (E);
    \path [->] (C) edge (B);

    \path [->, opacity=.2] (A) edge [bend right] (B);
    \path [->] (A) edge (D);
\end{scope}
\end{tikzpicture}
\end{minipage}
\begin{minipage}{0.32\textwidth}
\centering
\begin{tikzpicture}[scale=0.8]
\node (left) at (-5,1) {};
\node (above) at (0,2.9) {};
\node (below) at (0,-.9) {};
\node (l1) at (-4.5,1) {\footnotesize $(1,0)$};
\node (l2) at (0,2.65) {\footnotesize $\top$};

\begin{scope}[every node/.style={circle,draw,inner sep = 2pt}]
    \node (A) at (0,2) {1};    
    \node (E) at (-3.7,1) {2};
\end{scope}
\begin{scope}[every node/.style={rectangle,draw}]
    \node[label={below}:{\footnotesize $\top$}] (B) at (0,0) {1};    
    \node[label={below}:{\footnotesize $(1,0)$}] (C) at (-2,0) {3};
    \node[label={above}:{\footnotesize $(0,0)$}] (D) at (-2,2) {4};
\end{scope}

\begin{scope}[>={Stealth[black]},
              every edge/.style={draw}]
    \path [->] (B) edge [bend right] (A);
    \path [->] (C) edge node[right] {\small $e_2$} (D);
    \path [->] (D) edge [bend right] (E);
    \path [->] (E) edge [bend right] (C);
    \path [->] (C) edge [bend right] (E);
    \path [->] (C) edge node[above] {\small $e_3$} (B);

    \path [->] (A) edge [bend right] (B);
    \path [->,opacity=.2] (A) edge node[below] {\small $e_1$} (D);
\end{scope}
\end{tikzpicture}

\begin{tikzpicture}[scale=0.8]
\node (left) at (-5,1) {};
\node (above) at (0,2.9) {};
\node (below) at (0,-.9) {};
\node (l1) at (-4.5,1) {\footnotesize $(\varepsilon,0)$};
\node (l2) at (0,2.65) {\footnotesize $\top$};

\begin{scope}[every node/.style={circle,draw,inner sep = 2pt}]
    \node (A) at (0,2) {1};    
    \node (E) at (-3.7,1) {2};
\end{scope}
\begin{scope}[every node/.style={rectangle,draw}]
    \node[label={below}:{\footnotesize $\top$}] (B) at (0,0) {1};    
    \node[label={below}:{\footnotesize $(\varepsilon,0)$}] (C) at (-2,0) {3};
    \node[label={above}:{\footnotesize $(0,\varepsilon)$}] (D) at (-2,2) {4};
\end{scope}

\begin{scope}[>={Stealth[black]},
              every edge/.style={draw}]
    \path [->] (B) edge [bend right] (A);
    \path [->] (C) edge node[right] {\small $e_2$} (D);
    \path [->] (D) edge [bend right] (E);
    \path [->] (E) edge [bend right] (C);
    \path [->] (C) edge [bend right] (E);
    \path [->] (C) edge node[above] {\small $e_3$} (B);

    \path [->] (A) edge [bend right] (B);
    \path [->,opacity=.2] (A) edge node[below] {\small $e_1$} (D);
\end{scope}
\end{tikzpicture}
\end{minipage}
\caption{The top and bottom rows illustrate an example run of Algorithm \ref{algo:strategy+iteration} with the perfect (3,2)-universal tree and the succinct (3,2)-universal tree respectively.
  In each row, the left figure depicts the game instance (nodes in $V_0$ and $V_1$ are drawn as squares and circles respectively).
  The next two figures show Odd's strategy and the node labeling at the start of Iteration 1 and 2.
  The arcs not selected by Odd are greyed out.
  In the right figure, $e_1$ is loose, $e_2$ is tight, and $e_3$ is violated. }
\label{fig:iterations}
\end{figure}
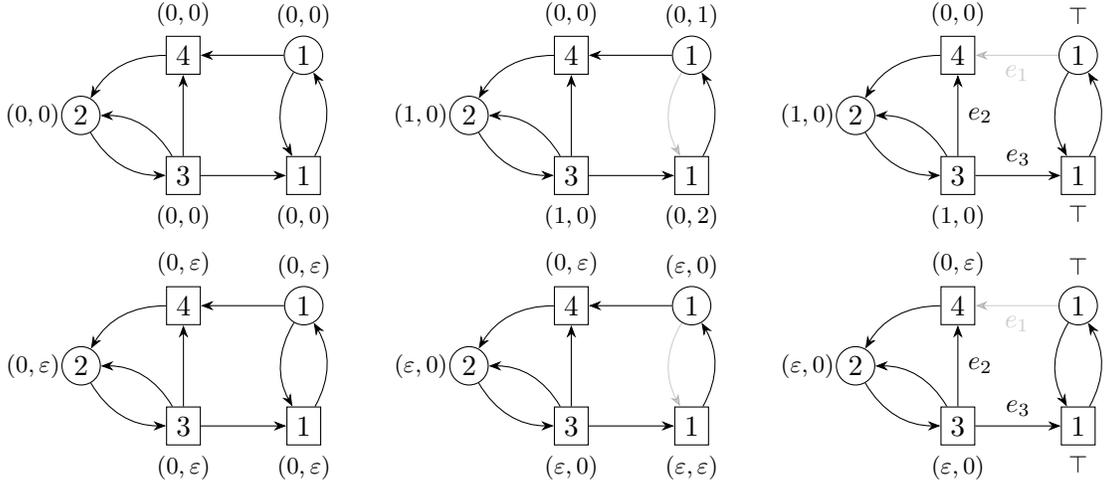

The correctness of Algorithm~\ref{algo:strategy+iteration} is an easy consequence of Proposition \ref{prop:Tarski}.

\begin{thm}\label{thm:correctness}
Algorithm~\ref{algo:strategy+iteration} returns the pointwise minimal node labeling $\mu^*:V\rightarrow \bar{L}(T)$ which is feasible in $G$. 
\end{thm}
\begin{proof}
The node labeling $\mu$ is monotone increasing in every iteration.
Since $\mathcal{L}$ is finite and the all-top node labeling is feasible in $G$, the algorithm terminates.
Let $\mu^*$ be the pointwise minimal node labeling which is feasible in $G$.
Note that $\mu^*$ is the least simultaneous fixed point of $\mathcal{G}^\uparrow$ in $\mathcal{L}$.
By induction and Proposition~\ref{prop:Tarski}\,(ii), we have $\mu^{\mathcal{G}^\uparrow_\tau} \leq \mu^*$ in every iteration.
As the algorithm terminates with a simultaneous fixed point of $\mathcal{G}^\uparrow$, it terminates with $\mu^*$.
\end{proof}
Thus, by Theorem~\ref{thm:feasible_pg}, the algorithm correctly determines the winning positions for Even.

\paragraph{Running time}
In Algorithm \ref{algo:strategy+iteration}, a naive method for computing the least fixed point $\mu^{\mathcal{G}^\uparrow_\tau}$ is to iterate the operators in $\mathcal{G}^\uparrow_\tau$ on $\mu$ until convergence.
The operators $\text{Lift}_v$ and $\text{Lift}_{vw}$ can be implemented to run in $O(|\delta^+(v)|\gamma(T))$ time.
Since $\mu$ is monotone increasing throughout, the total running time of Algorithm \ref{algo:strategy+iteration} is 
\[O\pr{\sum_{v\in V} |\delta^+(v)|\gamma(T)|L(T)|} = O(m\gamma(T)|L(T)|),\]
which matches progress measure algorithms \cite{conf/stacs/Jurdzinski00,conf/lics/JurdzinskiL17}.
However, this method of computing $\mu^{\mathcal{G}^\uparrow_\tau}$ can take $\Omega(\gamma(T)|L(T)|)$ time; recall that $|L(T)|$ is at least quasi-polynomial for a universal tree $T$ \cite{conf/soda/CzerwinskiDFJLP19}.
Our goal is to compute $\mu^{\mathcal{G}^\uparrow_\tau}$ in polynomial time.
Nevertheless, running this method in parallel with a more efficient algorithm for computing $\mu^{\mathcal{G}^\uparrow_\tau}$ is still useful in ensuring that we are not slower than the progress measure algorithm overall.

\subsection{The Least Fixed Point of 1-Player Games}
\label{sec:prelude}

Let $(G_\tau,\pi)$ be a 1-player game for Even, and let $\mu\in \mathcal{L}$ be a node labeling such that there are no loose arcs in $G_\tau$.
In the rest of the paper, we focus on developing efficient methods for computing $\mu^{\mathcal{G}^\uparrow_\tau}$.
We know that applying the operators in $\mathcal{G}^\uparrow_\tau$ to $\mu$ is not polynomial in general.
So, we will approach $\mu^{\mathcal{G}^\uparrow_\tau}$ from above instead.

Given a node labeling $\nu:V\rightarrow \bar{L}(T)$ and an arc $vw\in E$, let $\text{drop}(\nu,vw)$ be the largest element $\xi\in \bar{L}(T)$ such that $\xi\leq \nu(v)$ and $vw$ is not loose after setting $\nu(v)$ to $\xi$.
Observe that if $vw$ is loose, then $\text{drop}(\nu,vw)$ is given by {\sc Tighten}$(\nu,vw)$.
Otherwise, it is equal to $\nu(v)$.
Hence, it can be computed in $\gamma(T)$ time.

\begin{exa}
    Let $vw\in E$ be an arc where $\pi(v) = 2$. 
    If $\nu$ is a node labeling from the perfect (3,2)-universal tree such that $\nu(v) = (1,2)$ and $\nu(w) = (1,1)$, then $\text{drop}(\nu,vw) = (1,0)$.
    If $\nu'$ is a node labeling from the succinct (3,2)-universal tree such that $\nu'(v) = (0,\varepsilon)$ and $\nu'(w) = (\varepsilon,\varepsilon)$, then $\text{drop}(\nu',vw) = (0,\varepsilon)$.
\end{exa}

We are ready to define the deflationary counterpart of $\text{Lift}_{vw}$.
For every arc $vw\in E_\tau$, define the operator $\text{Drop}_{vw}:\mathcal{L} \times V\rightarrow \bar{L}(T)$ as
\[\text{Drop}_{vw}(\nu,u):= \begin{cases}
    \text{drop}(\nu,vw), &\text{ if }u=v\\
    \nu(v), &\text{ otherwise.}
\end{cases}\]
For a subgraph $H\subseteq G_\tau$, we denote $\mathcal{H}^\downarrow:=\{\text{Drop}_e:e\in E(H)\}$ as the operators in $H$.
Since they are deflationary and monotone, for any $\nu\in \mathcal{L}$, the \emph{greatest} simultaneous fixed point $\nu^{\mathcal{H}^\downarrow}$ exists by Proposition~\ref{prop:Tarski}\,(i).
Note that a node labeling is a simultaneous fixed point of $\mathcal{H}^\downarrow$ if and only if there are no loose arcs in $H$ with respect to it.

Our techniques are inspired by the methods of \emph{label-correcting} and \emph{label-setting} for the shortest path problem.
In the shortest path problem, we have a designated target node $t$ whose label is initialized to $0$.
For us, the role of $t$ is replaced by a (potentially empty) set of even cycles in $G_\tau$, which we do not know a priori.
So, we define a set of candidates nodes called \emph{base nodes}, whose labels need to be initialized properly.

\begin{defi}
Given a 1-player game $(G_\tau,\pi)$ for Even, we call $v\in V$ a \emph{base node} 
if $v\in \Pi(C)$ for some even cycle $C$ in $G_\tau$.\
Denote $B(G_\tau)$ as the set of base nodes in $G_\tau$.
\end{defi}

Next, we introduce the {\sc MinBottleneckCycles} problem, which will play an important role in this paper. The input is a strongly connected digraph $G = (V,E)$ with arc costs $c\in \R^E$. For a cycle $C$, we define its \emph{bottleneck cost} as $c(C) := \max_{e\in E(C)} c(e)$. The goal is to compute for every node $v\in V$, the minimum bottleneck cost of a cycle which contains $v$. 
  It can be solved efficiently using a slight variation of the algorithm by \cite{conf/fossacs/KingKV01} for finding even cycles in parity automata.

\begin{thm}\label{thm:min-bottleneck-cycles}
Given a strongly connected digraph $G=(V,E)$ with $n$ nodes, $m$ arcs and $\ell$ distinct arc costs, {\sc MinBottleneckCycles} can be solved in $O(m\log \ell)$ time.
\end{thm}

\begin{proof}
The main idea is to recursively decompose the graph into strongly connected components (SCCs) using binary search.
Let $c_1<c_2<\dots < c_\ell$ be the distinct arc costs.
Consider the subgraph $H\subseteq G$ induced by arcs with cost at most $c_{\ceil{\ell/2}}$, and let $K_1, K_2, \dots, K_t$ be its SCCs.
They can be computed in $O(n+m) = O(m)$ time using Tarjan's SCCs algorithm \cite{journals/siamcomp/Tarjan72}.
For each $i\in [t]$, if $\size{E(K_i)}>0$, then every node in $K_i$ has a cycle going through it with cost at most $c_{\ceil{\ell/2}}$.
Otherwise, $K_i$ is a singleton with no self-loops.
Hence, every cycle going through it has cost greater than $c_{\ceil{\ell/2}}$.
This observation allows us to split $G$ into $t+1$ subinstances $G_1,G_2,\dots,G_{t+1}$.
The first $t$ subinstances are given by $G_i = (V_i,E_i) = K_i$ for all $i\in [t]$.
The last subinstance $G_{t+1} = (V_{t+1},E_{t+1})$ is obtained by contracting each $K_i$ into a node in $G$, and destroying any self-loops on the resulting node.
Then, the procedure is repeated on these subinstances.
In particular, for each $i\in [t]$, we consider the subgraph $H_i\subseteq G_i$ induced by arcs with cost at most $c_{\ceil{\ell/4}}$.
For $G_{t+1}$, we consider the subgraph $H_{t+1}\subseteq G_{t+1}$ induced by arcs with cost at most $c_{\ceil{3\ell/4}}$.

By construction, every edge in $G$ contributes an edge to at most one of the $t+1$ subinstances.
Hence, $\sum_{i=1}^{t+1}|E_i| \leq |E|$.
Furthermore, for every $i\in [t+1]$, we have $|E_i|\geq |V_i|$ because $G_i$ is strongly connected.
It follows that the total running time of all $t+1$ calls to Tarjan's SCCs algorithm is $O(m)$ again.
Since the recursion depth is $O(\log \ell)$, the overall running time is $O(m\log \ell)$.
\end{proof}

Computing the set of base nodes in $G_\tau$ can be reduced to {\sc MinBottleneckCycles}. For every node $v\in V$, assign cost $c(e) := \pi(v)$ to all its outgoing arcs $e\in \delta^+_{G_\tau}(v)$. Then, apply Theorem~\ref{thm:min-bottleneck-cycles} to every SCC of $G_\tau$ which contains at least one arc. It is easy to see that $v\in B(G_\tau)$ if and only if $\pi(v)$ is even and the minimum bottleneck cost of a cycle passing through $v$ is $\pi(v)$. As this reduction takes $O(n+m) = O(m)$ time, computing $B(G_\tau)$ takes $O(m\log d)$ time.

\smallskip

We start by presenting a label-setting method for computing $\mu^{\mathcal{G}^\uparrow_\tau}$, and demonstrate its applicability to perfect universal trees.
Then, we develop a label-correcting method for computing $\mu^{\mathcal{G}^\uparrow_\tau}$, and apply it to the quasi-polynomial universal trees constructed in the literature \cite{conf/lics/JurdzinskiL17,conf/icalp/DaviaudJT20}.

\section{Label-Setting Method for Computing the Least Fixed Point}
\label{sec:label-setting}

The epitome of a label-setting algorithm is none other than Dijkstra's algorithm \cite{journals/nm/Dijkstra59}. 
In this section, we develop its analogue for ordered trees.
It can be used to compute $\mu^{\mathcal{G}^\uparrow_\tau}$ whenever the labels $\{\mu^{\mathcal{G}^\uparrow_\tau}(v):v\in B(G_\tau)\}$ are known beforehand. We show that these labels are easily obtained for perfect universal trees, so it can be used to compute $\mu^{\mathcal{G}^\uparrow_\tau}$ for these trees.

The algorithm takes as input a 1-player game $(G_\tau,\pi)$ for Even and 
a node labeling $\nu:V\rightarrow \bar{L}(T)$ from some ordered tree $T$.
During its execution, the node labeling $\nu$ is updated.
The algorithm also maintains a growing node set $S\subseteq V$ such that $\nu(v)$ remains fixed for all $v\in S$.
For the sake of brevity, let us denote $H:=G_\tau\setminus \cup_{v\in B(G_\tau)}\delta^+(v)$. 
Observe that there are no even cycles in $H$.

In every iteration, a new node is added to $S$, whose label is fixed.
To determine this node, we introduce a label function $\Phi$, which remains fixed throughout the algorithm.
It encodes a family of topological orders in $H$ induced by the even priorities.
The next node to be added to $S$ is then selected using a \emph{potential} function $\Phi^{\nu}$, defined based on the labels $\Phi$ and $\nu$.
We remark that the usual criterion of selecting a node $v\in V\setminus S$ based on the smallest label $\nu(v)$ does not work. This is because the representation of a parity game as a mean payoff game can have negative arc weights, and the correctness of Dijkstra's algorithm is not guaranteed in the presence of negative arc weights.

To describe $\Phi$, we define a family of functions $\Phi_p$, parametrized by the even priorities in $H$.
For an even $p\in [d]$, the function $\Phi_p$ encodes the topological order of nodes in the subgraph $H_p$. 
Recall that $H_p$ is the subgraph of $H$ induced by the nodes with priority at most $p$.
Formally, $\Phi_p:V\rightarrow \Z_+$ is any function which satisfies the following three properties:
\begin{itemize}
  \item $\Phi_p(v) = 0$ if and only if $\pi(v) > p$;
  \item $\Phi_p(v) \geq \Phi_p(w)$ if $v$ can reach $w$ in $H_p$;
  \item $\Phi_p(v) = \Phi_p(w) > 0$ if and only if $v$ and $w$ are strongly connected in $H_p$.
\end{itemize}

The label function $\Phi:V\rightarrow \Z^{d/2}_+$ is then defined as
\[\Phi(v) := (\Phi_d(v), \Phi_{d-2}(v), \dots, \Phi_2(v)).\] 
A linear order on $\Phi(v)$ is obtained by extending the linear order of its components lexicographically.

\begin{rem}
Given a pair of nodes $v$ and $w$, comparing $\Phi(v)$ and $\Phi(w)$ amounts to finding the largest $p\in [d]$ such that $\Phi_p(v)\neq \Phi_p(w)$.
Observe that if $\Phi_q(v) = \Phi_q(w) > 0$ for some $q\in [d]$, then $\Phi_r(v) = \Phi_r(w)$ for all $r\geq q$.
On the other hand, if $\Phi_q(v) = \Phi_q(w) = 0$, then $\Phi_r(v) = \Phi_r(w) = 0$ for all $r\leq q$.
Hence, such a $p$ can be computed in $O(\log d)$ time via binary search.
\end{rem}

Given a node labeling $\nu:V\rightarrow \bar{L}(T)$, the potential function $\Phi^\nu:V\rightarrow (\Z_+^{d/2}\times L(T))\cup\set{\infty}$ is obtained by interlacing the components of $\Phi$ and $\nu$ in the following way
\[\Phi^\nu(v) := \begin{cases}
  (\Phi_d(v), \nu(v)_{d-1}, \dots, \Phi_2(v), \nu(v)_1), &\text{ if }\nu(v)\neq \top.\\
  \infty &\text{ otherwise.}
\end{cases} \]
A linear order on $\Phi^\nu(v)$ is acquired by extending the linear order of its components lexicographically.
For any $p\in [d]$, the \emph{$p$-truncation} of $\Phi(v)$ and $\Phi^\nu(v)$, denoted $\Phi(v)|_p$ and $\Phi^\nu(v)|_p$ respectively, are obtained by deleting the components with index less than $p$, with the convention $\infty|_p := \infty$.

We are ready to state Dijkstra's algorithm for ordered trees (Algorithm \ref{algo:dijkstra}).
First, it initializes the node set $S$ as $B(G_\tau)$, and sets $\nu(v) := \top$ for all $v\in V\setminus S$.
Then, for each even $p\in [d]$, it computes the topological order $\Phi_p$ by running Tarjan's SCCs algorithm on $H_p$.
Next, $\nu$ is updated by dropping the tail labels of the incoming arcs $\delta^-(S)$.
At the start of every iteration, the algorithm selects a node $u$ with minimum potential $\Phi^\nu(u)$ among all the nodes in $V\setminus S$ (ties are broken arbitrarily).
Then, it adds $u$ to $S$ and updates $\nu$ by dropping the tail labels of the incoming arcs $\delta^-(u)\cap \delta^-(S)$. 
The algorithm terminates when $S = V$.

\begin{algorithm}[htbp]
  \caption{Dijkstra: $(G_\tau,\pi)$ 1-player game for Even, $\nu:V\rightarrow \bar{L}(T)$ node labeling from an ordered tree $T$}
  \label{algo:dijkstra}
  \begin{algorithmic}[1]
    \Procedure{Dijkstra}{$(G_\tau, \pi), \nu$}
    \State $S \gets B(G_\tau)$
    \State $\nu(v) \gets \top$ for all $v\in V\setminus S$
    \State Compute topological order $\Phi_p$ for all even $p\in [d]$ \Comment{Using Tarjan's SCCs algorithm}
    \ForAll{$vw\in \delta^-(S)$}
    \State $\nu(v) \gets \text{drop}(\nu,vw)$
    \EndFor
    \While{$S\subsetneq V$}
    \State $u\in \arg\min_{v\in V\setminus S}\Phi^\nu(v)$ \Comment{Break ties arbitrarily}
    \State $S\gets S\cup\set{u}$
    \ForAll{$vu\in \delta^-(u)$ where $v\notin S$}
    \State $\nu(v) \gets \text{drop}(\nu,vu)$
    \EndFor
    \EndWhile
    \State \Return $\nu$
    \EndProcedure
  \end{algorithmic}
\end{algorithm}

An efficient implementation of Dijkstra's algorithm using Fibonacci heaps was given by Fredman and Tarjan \cite{journals/jacm/FredmanT87}.
Its running time is $O(m+n\log n)$ when the keys in the heap are real numbers, assuming that each elementary operation on the reals takes constant time.
In our setting, the keys are the node potentials $\Phi^\nu$.
Since computing drop$(\nu,e)$ takes $\gamma(T)$ time while comparing the potential of two nodes takes $\gamma(T) + \log d$ time, their result translates to $O(\gamma(T) m+ (\gamma(T)+\log d)n\log n)$ time here.
Since identifying the base nodes $B(G_\tau)$ takes $O(m\log d)$ time while computing the topological orders $\Phi_p$ takes $O(dm)$ time, the total running time of Algorithm \ref{algo:dijkstra} is $O((\gamma(T)+d)m+ (\gamma(T)+\log d)n\log n)$.

The following lemma is a consequence of no even cycles in $H$.

\begin{lem}\label{lem:tight_arc}
Let $\nu\in \mathcal{L}$ be a node labeling. For every non-violated arc $vw\in E(H)$, we have $\Phi^\nu(v)\geq\Phi^\nu(w)$.
\end{lem}

\begin{proof}
We may assume that $\nu(v)\neq \top$, as otherwise $\Phi^\nu(v) =\infty$.
Since $v$ can reach $w$ in $H$, we have $\Phi(v)|_{\pi(v)}\geq \Phi(w)|_{\pi(v)}$.
We also have $\nu(v)|_{\pi(v)}\geq \nu(w)|_{\pi(v)}$ because $vw$ is non-violated. 
Combining these two inequalities yields $\Phi^\nu(v)|_{\pi(v)} \geq \Phi^\nu(w)|_{\pi(v)}$.

For the purpose of contradiction, suppose that $\Phi^\nu(v)<\Phi^\nu(w)$.
Then, $\Phi^\nu(v)|_{\pi(v)} = \Phi^\nu(w)|_{\pi(v)}$.
Since this implies $\nu(v)|_{\pi(v)} = \nu(w)|_{\pi(v)}$, we conclude that $\pi(v)$ is even because $vw$ is non-violated and $\nu(v),\nu(w)\neq \top$.
It follows that $0 < \Phi_{\pi(v)}(v) = \Phi_{\pi(v)}(w)$; recall that $\Phi$ is indexed by the even priorities.
Hence, $v$ and $w$ are strongly connected in $H_{\pi(v)}$, so there exists a cycle $C$ in $H_{\pi(v)}$ such that $vw\in E(C)$.
As $v$ has the largest priority in $H_{\pi(v)}$, we have $\pi(C) = \pi(v)$.
Thus, $C$ is even, which is a contradiction.
\end{proof}

The next theorem shows that Algorithm \ref{algo:dijkstra} returns the pointwise minimal node labeling which is feasible in $H$.
The key observation is that the sequence of node potentials admitted to $S$ during the algorithm is monotonically nondecreasing. 

\begin{thm}\label{thm:dijkstra}
Given an initial node labeling $\nu\in \mathcal{L}$, Algorithm \ref{algo:dijkstra} returns the pointwise minimal node labeling $\nu^*\in \mathcal{L}$ which is feasible in $H$ and satisfies $\nu^*(v) = \nu(v)$ for all $v\in B(G_\tau)$.
\end{thm}

\begin{proof}
Let $\nu_0:=\nu$ be the input node labeling, and suppose that the algorithm ran for $k-1$ iterations.
For every $i\geq 1$, let $\nu_i$ be the node labeling at the start of iteration $i$. 
Then, the algorithm returns $\nu_k$.
Clearly, $\nu_0$ is feasible in $H$ because $\nu_0(v) = \top$ for all $v\notin B(G_\tau)$.
Since $\nu_i$ is obtained from $\nu_{i-1}$ by a sequence of $\text{drop}(\cdot,\cdot)$ operations, and this operation preserves feasibility, it follows that $\nu_k$ is feasible in $H$.
Furthermore, $\nu_k(v) = \nu_{k-1}(v) = \dots = \nu_0(v)$ for all $v\in B(G_\tau)$ due to our initialization $S\gets B(G_\tau)$.

\smallskip

For every $i\geq 1$, let $u_i$ be the node added to $S$ in iteration $i$.
We start by showing that $\Phi^{\nu_k}(u_i)\geq \Phi^{\nu_k}(u_{i-1})$ for all $i>1$. 
This is equivalent to $\Phi^{\nu_i}(u_i) \geq \Phi^{\nu_i}(u_{i-1})$ because $\nu_k(u_i) = \nu_j(u_i)$ for all $j\geq i$. If $\nu_i(u_i) = \nu_{i-1}(u_i)$, then we are done by our choice of $u_{i-1}$. So, we may assume that $\nu_i(u_i)<\nu_{i-1}(u_i)$, which implies that drop$(\cdot,\cdot)$ was applied to $u_i u_{i-1}\in E(H)$ in iteration $i-1$. Hence, the arc $u_iu_{i-1}$ is tight with respect to $\nu_i$. By Lemma~\ref{lem:tight_arc}, we obtain $\Phi^{\nu_i}(u_i) \geq \Phi^{\nu_i}(u_{i-1})$.

\smallskip

Next, we show that there are no loose arcs in $H$ with respect to $\nu_k$.
Since $B(G_\tau)$ is an independent set in $H$, it suffices to prove that $u_iu_j$ is not loose with respect to $\nu_k$ for all $1 \leq i < j$ where $u_iu_j\in E(H)$.
For the purpose of contradiction, let $u_iu_j\in E(H)$ be a loose arc with respect to $\nu_k$.
Note that $\nu_k(u_i)\neq \top$, as otherwise it would imply $\nu_k(u_j) = \top$ because $\Phi^{\nu_k}(u_i)\leq \Phi^{\nu_k}(u_j)$.
Since the label $\nu_k(u_i)$ was given by {\sc Tighten}, it is the smallest leaf in the subtree of $T$ rooted at $\nu_k(u_i)|_{\pi(u_i)}$. 
Therefore, $\nu_k(u_i)|_{\pi(u_i)}>\nu_k(u_j)|_{\pi(u_i)}$.
Let $p$ be the smallest even integer such that $p\geq \pi(u_i)$.
Then, $\Phi(u_i)|_p\geq \Phi(u_j)|_p$ because either $\pi(u_j)>p$ or $u_iu_j\in E(H_p)$.
However, these two inequalities yield $\Phi^{\nu_k}(u_i)>\Phi^{\nu_k}(u_j)$, which is a contradiction.

\smallskip

It is left to show the pointwise minimality of $\nu_k$.
Let $\nu':V\rightarrow \bar{L}(T)$ be a node labeling feasible in $H$ such that $\nu'(v) = \nu(v)$ for all $v\in B(G_\tau)$.
We prove that $\nu'\geq \nu_k$. Let $u\in V$; we may assume that $\nu'(u)\neq \top$.
From the definition of feasibility, there exists a strategy $\sigma$ for Even such that $G_{\sigma\tau}\setminus \cup_{v\in B(G_\tau)}\delta^+(v)$ does not contain violated arcs with respect to $\nu'$.
Observe that $u$ can reach $B(G_\tau)$ in this subgraph.
Indeed, if $u$ cannot reach $B(G_\tau)$, then it reaches a cycle $C$. As there are no even cycles in $H$, $C$ is odd and so $\nu'(v) = \top$ for all $v\in V(C)$.
Now, let $P$ be a $u$-$w$ path in this subgraph for some $w\in B(G_\tau)$.
Since $P$ has no loose arcs with respect to $\nu_k$ and $\nu_k(w)=\nu(w)=\nu'(w)$, an inductive argument from $w$ along $P$ shows that $\nu_k(v)\leq \nu'(v)$ for all $v\in V(P)$.
\end{proof}

Consequently, if we can determine $\mu^{\mathcal{G}^{\uparrow}_\tau}(v)$ for all $v\in B(G_\tau)$ beforehand, then we can use Algorithm \ref{algo:dijkstra} to compute $\mu^{\mathcal{G}^{\uparrow}_\tau}$.

\begin{cor}\label{cor:dijkstra}
Let $\mu\in \mathcal{L}$ be a node labeling with no loose arcs in $G_\tau$.
Given an initial node labeling $\nu\in \mathcal{L}$ where $\nu(v) = \mu^{\mathcal{G}^{\uparrow}_\tau}(v)$ for all $v\in B(G_\tau)$, Algorithm \ref{algo:dijkstra} returns $\mu^{\mathcal{G}^{\uparrow}_\tau}$.
\end{cor}

\begin{proof}
Let $\nu^*$ be the node labeling returned by Algorithm \ref{algo:dijkstra}.
By Theorem \ref{thm:dijkstra}, we have $\nu^* \leq \mu^{\mathcal{G}^{\uparrow}_\tau}$ because $H$ is a subgraph of $G_\tau$.
For every $v\in B(G_\tau)$, observe that $\nu^*(v) = \mu^{\mathcal{G}^{\uparrow}_\tau}(v)$ and $v$ has a non-violated outgoing arc with respect to $\mu^{\mathcal{G}^{\uparrow}_\tau}$. 
Hence, $\nu^*$ is feasible in $G_\tau$.
To show that $\nu^*\geq \mu^{\mathcal{G}^{\uparrow}_\tau}$, it suffices to prove that $\nu^*\geq \mu$ due to the pointwise minimality of $\mu^{\mathcal{G}^{\uparrow}_\tau}$.
Let $u\in V$; we may assume that $\nu^*(u)\neq \top$.
Let $\sigma$ be a strategy for Even such that $G_{\sigma\tau}$ does not contain violated arcs with respect to $\nu^*$.
As $\nu^*(u)\neq \top$, $u$ can reach an even cycle in $G_{\sigma\tau}$.
So, there exists a $u$-$w$ path $P$ for some $w\in B(G_\tau)$ in $G_{\sigma\tau}$.
Since $P$ has no loose arcs with respect to $\mu$ and $\mu(w)\leq \mu^{\mathcal{G}^{\uparrow}_\tau}(w) = \nu^*(w)$, an inductive argument from $w$ along $P$ shows that $\mu(v)\leq \nu^*(v)$ for all $v\in V(P)$.
\end{proof}

\subsection{Application to Perfect Universal Trees}
In this subsection, we show how to compute $\mu^{\mathcal{G}^{\uparrow}_\tau}$ using Algorithm~\ref{algo:dijkstra} for perfect universal trees.
Given a node labeling $\mu\in \mathcal{L}$ such that $G_\tau$ has no loose arcs, Algorithm~\ref{algo:label-setting} first applies the operator $\text{Lift}_v\in \mathcal{G}^{\uparrow}_\tau$ once for each base node $v\in B(G_\tau)$.
Then, it calls Algorithm~\ref{algo:dijkstra} with this updated node labeling.

\begin{algorithm}[htbp]
  \caption{Label-Setting for Perfect Universal Trees: $(G_\tau,\pi)$ 1-player game for Even, $\mu:V\rightarrow \bar{L}(T)$ node labeling from a perfect $(n,d/2)$-universal tree $T$ such that $G_\tau$ does not contain loose arcs.}
  \label{algo:label-setting}
  \begin{algorithmic}[1]
    \Procedure{LabelSetting}{($G_\tau,\pi),\mu$}
    \State $\nu\gets \mu$
    \ForAll{$v\in B(G_\tau)$} 
      \State $\nu(v) \gets \min_{vw\in E_\tau}\text{lift}(\nu,vw)$ \label{line:lift}
    \EndFor
    \State $\nu \gets $ \Call{Dijkstra}{$(G_\tau,\pi),\nu$}
    \State \Return $\nu$
    \EndProcedure
  \end{algorithmic}
\end{algorithm}

\begin{thm}\label{thm:label-setting}
Algorithm~\ref{algo:label-setting} returns $\mu^{\mathcal{G}^{\uparrow}_\tau}$ in $O(d(m+n\log n))$ time.
\end{thm}

\begin{proof}
Since there are no loose arcs in $G_\tau$, we have $\mu(v)|_{\pi(v)}\leq \mu(w)|_{\pi(v)}$ for all $vw\in E_\tau$ where $v\in B(G_\tau$).
Furthermore, if $\mu(v)|_{\pi(v)}= \mu(w)|_{\pi(v)}$, then $\mu(v)$ is either $\top$ or the smallest leaf in the subtree of $T$ rooted at $\mu(v)|_{\pi(v)}$.
Let $\nu\in \mathcal{L}$ be the node labeling right before calling {\sc Dijkstra}.
It is obtained by applying the operators $\{\text{Lift}_v:v\in B(G_\tau)\}\subseteq \mathcal{G}^{\uparrow}_\tau$ to $\mu$.
Hence, for every base node $v\in B(G_\tau)$, $\nu(v)$ is either $\top$ or the smallest leaf in the subtree of $T$ rooted at $\mu(v)|_{\pi(v)}$.
Moreover, we maintain $\nu\leq \mu^{\mathcal{G}^\uparrow_\tau}$ and the absence of loose arcs in $G_\tau$ with respect to $\nu$.

First, we prove correctness.
By Corollary \ref{cor:dijkstra}, it suffices to prove that $\nu(v) = \mu^{\mathcal{G}^{\uparrow}_\tau}(v)$ for all $v\in B(G_\tau)$.
Fix a base node $w\in B(G_\tau)$.
We may assume that $\nu(w)\neq \top$.
Let $C$ be a cycle dominated by $w$ in $G_\tau$, and consider the path $P:=C\setminus wu$ where $wu\in E(C)$.
Let $\bar{\nu}:V\rightarrow \bar{L}(T)$ be a node labeling such that $\bar{\nu}(v) = \top$ for all $v\notin V(P)$, $\bar{\nu}(w) = \nu(w)$ and $P$ is tight with respect to $\bar{\nu}$.
Then, $\bar{\nu}$ is feasible in $G_\tau\setminus \delta^+(w)$.
Moreover, $\bar{\nu}\geq \nu$ because there are no loose arcs in $P$ with respect to $\nu$.
Now, recall that $\nu(w)$ is the smallest leaf in the subtree of $T$ rooted at $\nu(w)|_{\pi(w)}$.
Since $\pi(v) \leq \pi(w)$ for all $v\in V(P)$, we have $\bar{\nu}(u)|_{\pi(w)} = \bar{\nu}(w)|_{\pi(w)}$ because $\size{V(P)}\leq n$.
As $\pi(w)$ is even, the arc $wu$ is tight with respect to $\bar{\nu}$.
It follows that $\bar{\nu}$ is feasible in $G_\tau$.
Thus, $\nu\leq \mu^{\mathcal{G}^{\uparrow}_\tau}\leq \nu^{\mathcal{G}^{\uparrow}_\tau} \leq \bar{\nu}$, where the second inequality is due to $\mu\leq \nu$ and the third inequality is due to the pointwise minimality of $\nu^{\mathcal{G}^\uparrow_\tau}$.
In particular, $\nu(w) = \mu^{\mathcal{G}^{\uparrow}_\tau}(w)$.

Next, we prove the running time. Applying the operator $\text{Lift}_v\in \mathcal{G}^{\uparrow}_\tau$ once for each base node $v\in B(G_\tau)$ involves calling the {\sc Tighten} subroutine at most $m$ times. On the perfect $(n,d/2)$-perfect universal tree $T$, {\sc Tighten} takes $\gamma(T) = O(d)$ time. Finally, {\sc Dijkstra} with $T$ takes $O(d(m+n\log n))$ time.
\end{proof}

\begin{rem}\label{rem:label-setting}
The purpose of Line~\ref{line:lift} in Algorithm~\ref{algo:label-setting} is to ensure that before calling {\sc Dijkstra}, $\nu(v)$ is either $\top$ or the smallest leaf in the subtree of $T$ rooted at $\nu(v)|_{\pi(v)}$ for every base node $v\in B(G_\tau)$. If Algorithm~\ref{algo:label-setting} is used to compute $\mu^{\mathcal{G}^\uparrow_\tau}$ in every iteration of Algorithm~\ref{algo:strategy+iteration}, then $\mu$ already satisfies this property due to the pointwise minimality of $\mu^{\mathcal{G}^\uparrow_\tau}$. So, Line~\ref{line:lift} can be skipped.
\end{rem}

\section{Label-Correcting Method for Computing the Least Fixed Point}
\label{sec:label_correcting}

In the previous section, we have seen that knowing the labels $\{\mu^{\mathcal{G}^\uparrow_\tau}(v):v\in B(G_\tau)\}$ allows one to compute the least fixed point $\mu^{\mathcal{G}^\uparrow_\tau}$ using {\sc Dijkstra}. We have also shown how to compute these labels for perfect universal trees. For more complicated universal trees such as the quasi-polynomial ones, it is unclear how to compute these labels efficiently. In this section, we develop a label-correcting method for computing $\mu^{\mathcal{G}^\uparrow_\tau}$. It has the advantage of only requiring an over-approximation of $\mu^{\mathcal{G}^\uparrow_\tau}(v)$ for all $v\in B(G_\tau)$.

The Bellman--Ford algorithm for the shortest path problem is a well-known implementation of the generic label-correcting method~\cite{AhujaMagnantiOrlin:1993}.
We start by giving its analogue for ordered trees.
Algorithm~\ref{algo:Bellman+Ford} takes as input a 1-player game $(G_\tau,\pi)$ for Even and a
node labeling $\nu:V\rightarrow \bar{L}(T)$ from some ordered tree $T$.
Like its classical version for shortest paths, the algorithm runs for $n-1$ iterations.
In each iteration, it replaces the tail label of every arc $e\in E_\tau$ by $\text{drop}(\nu,e)$. 
Clearly, the running time is $O(mn\gamma(T))$.
Moreover, if $\nu'$ is the returned node labeling, then $\nu'\geq \nu^{\mathcal{G}^\downarrow_\tau}$. 

\begin{algorithm}[htbp]
  \caption{Bellman--Ford: $(G_\tau,\pi)$ 1-player game for Even, $\nu:V\rightarrow \bar{L}(T)$ node labeling from an ordered tree $T$}
  \label{algo:Bellman+Ford}
  \begin{algorithmic}[1]
    \Procedure{BellmanFord}{($G,\pi), \nu$}
    \For{$i=1$ \textbf{to} $n-1$}
      \ForAll{$vw\in E$} \Comment{In any order}
        \State $\nu(v) \gets \text{drop}(\nu,vw)$
      \EndFor 
    \EndFor
    \State \Return $\nu$
    \EndProcedure
  \end{algorithmic}
\end{algorithm}

Recall that we have a node labeling $\mu\in \mathcal{L}$ such that $G_\tau$ does not have loose arcs, and our goal is to compute $\mu^{\mathcal{G}^\uparrow_\tau}$.
We now specify an upper bound of $\nu(v)$ for all $v\in B(G_\tau)$ which ensures that Algorithm~\ref{algo:Bellman+Ford} does not return a node labeling larger than $\mu^{\mathcal{G}^\uparrow_\tau}$.

\begin{defi} \label{def:threshold}
Given a node labeling $\mu\in \mathcal{L}$, define $\widehat{\mu}:B(G_\tau)\to \bar{L}(T)$ as
\[
\widehat{\mu}(v):= \min_{\tilde{\mu}\in \mathcal{L}}\set{\tilde{\mu}(v):\tilde{\mu}(v)\geq \mu(v) \text{ and } \tilde{\mu} \text{ is feasible in a cycle dominated by }v \text{ in }G_\tau} \enspace .
\]
We call $\widehat\mu(v)$ \emph{the threshold label} of base node $v$.
\end{defi}

The next lemma follows directly from the pointwise minimality of $\mu^{\mathcal{G}^\uparrow_\tau}$.

\begin{lem} \label{lem:bounds-labels-base-nodes}
Let $\mu\in \mathcal{L}$ be a node labeling such that $G_\tau$ does not have loose arcs. 
For every base node $v\in B(G_\tau)$, the threshold label of $v$ satisfies $\widehat{\mu}(v)\geq \mu^{\mathcal{G}^\uparrow_\tau}(v)$.
\end{lem}

\begin{proof}
Fix a base node $v$.
Let $C$ be a cycle dominated by $v$ in $G_\tau$.
Let $\tilde{\mu}\in \mathcal{L}$ be a node labeling which is feasible in $C$ and satisfies $\tilde{\mu}(v)\geq \mu(v)$.
It suffices to prove that $\tilde{\mu}(v)\geq \mu^{\mathcal{G}^\uparrow_\tau}(v)$, as $C$ and $\tilde{\mu}$ were chosen arbitrarily.
Without loss of generality, we may assume that $\tilde{\mu}(w) = \top$ for all $w\in V\setminus V(C)$.
Then, $\tilde{\mu}$ is feasible in $G_\tau$.
Let $u$ be the unique out-neighbour of $v$ in $C$, and consider the $u$-$v$ subpath $P$ of $C$.
Since $\tilde{\mu}(v)\geq \mu(v)$ and there are no loose arcs in $P$ with respect to $\mu$, an inductive argument from $v$ along $P$ shows that $\tilde{\mu}(w)\geq \mu(w)$ for all $w\in V(P)$.
It follows that $\tilde{\mu}\geq \mu$.
Hence, $\tilde{\mu}$ is a fixed point of $\mathcal{G}^\uparrow_\tau$ that is pointwise at least $\mu$.
From the pointwise minimality of $\mu^{\mathcal{G}^\uparrow_\tau}$, we get $\tilde{\mu}\geq \mu^{\mathcal{G}^\uparrow_\tau}$.
\end{proof}

Even though $\mu^{\mathcal{G}^\uparrow_\tau}$ is feasible in $G_\tau$, there may exist a base node $v$ such that $\mu^{\mathcal{G}^\uparrow_\tau}$ is infeasible in every cycle dominated by $v$ in $G_\tau$ (see Figure~\ref{fig:threshold}). Hence, we do not have the reverse inequality.

\begin{figure}[ht] 
\begin{minipage}{0.4\textwidth}
\centering
\begin{tikzpicture}[scale=0.9]
\begin{scope}[every node/.style={circle,draw,inner sep = 2pt}]
    \node[label={left}:{\small $(0,\varepsilon)$}] (A) at (0,0) {1};    
    \node[label={left}:{\small $(0,\varepsilon)$}] (B) at (0,2) {1};
\end{scope}
\begin{scope}[every node/.style={rectangle,draw}]
    \node[label={right}:{\small $(0,\varepsilon)$}] (C) at (2,2) {2};    
    \node[label={right}:{\small $(0,\varepsilon)$}] (D) at (2,0) {4};
\end{scope}

\begin{scope}[>={Stealth[black]},
              every edge/.style={draw}]
    \path [->] (A) edge (B);
    \path [->] (B) edge (C);
    \path [->] (C) edge (D);
    \path [->] (D) edge (A);
    \path [->] (C) edge (A);
\end{scope}
\end{tikzpicture}
\end{minipage}
\begin{minipage}{0.4\textwidth}
\centering
\begin{tikzpicture}[scale=0.9]

\begin{scope}[every node/.style={circle,draw,inner sep = 2pt}]
    \node[label={left}:{\small $(\varepsilon,\varepsilon)$}] (A) at (0,0) {1};    
    \node[label={left}:{\small $(\varepsilon,0)$}] (B) at (0,2) {1};
\end{scope}
\begin{scope}[every node/.style={rectangle,draw}]
    \node[label={right}:{\small $(0,\varepsilon)$}] (C) at (2,2) {2};    
    \node[label={right}:{\small $(0,\varepsilon)$}] (D) at (2,0) {4};
\end{scope}

\begin{scope}[>={Stealth[black]},
              every edge/.style={draw}]
    \path [->] (A) edge (B);
    \path [->] (B) edge (C);
    \path [->] (C) edge (D);
    \path [->] (D) edge (A);
    \path [->] (C) edge (A);
\end{scope}
\end{tikzpicture}
\end{minipage}
\caption{An example of a strategy subgraph $G_\tau$ with node labelings from the succinct (3,2)-universal tree ($\mu$ on the left and $\mu^{\mathcal{G}_\tau^\uparrow}$ on the right). Nodes in $V_0$ and $V_1$ are drawn as squares and circles respectively. The base nodes are $v,w$ where $\pi(v) = 2$ and $\pi(w) = 4$. Note that $\widehat\mu(v) = (\varepsilon,0) > \mu^{\mathcal{G}^\uparrow_\tau}(v)$.}
\label{fig:threshold}
\end{figure}

The next theorem shows that if we initialize the base nodes with their corresponding threshold labels, then Algorithm \ref{algo:Bellman+Ford} returns $\mu^{\mathcal{G}^\uparrow_\tau}$. 
Even more, it suffices to have an initial node labeling $\nu\in \mathcal{L}$ such that $\mu^{\mathcal{G}^\uparrow_\tau}(v)\leq \nu(v)\leq \widehat{\mu}(v)$ for all $v\in B(G_\tau)$.
For the other nodes $v\notin B(G_\tau)$, we can simply set $\nu(v) \gets \top$.

\begin{thm}\label{thm:bellmanFord}
Let $\mu\in \mathcal{L}$ be a node labeling such that $G_\tau$ does not have loose arcs.
For every base node $v\in B(G_\tau)$, let $\widehat\mu(v)$ be the threshold label of $v$.
Given an input node labeling $\nu\in \mathcal{L}$ which satisfies $\nu\geq \mu^{\mathcal{G}^\uparrow_\tau}$ and $\nu(v)\leq \widehat{\mu}(v)$ for all $v\in B(G_\tau)$, Algorithm \ref{algo:Bellman+Ford} returns $\mu^{\mathcal{G}^\uparrow_\tau}$. 
\end{thm}

\begin{proof}
Given input $\nu\in \mathcal{L}$, let $\nu'$ be the node labeling returned by Algorithm~\ref{algo:Bellman+Ford}. Then, we have
\begin{equation}\label{eq:bellmanFord}
  \mu\leq \mu^{\mathcal{G}^\uparrow_\tau} \leq \nu^{\mathcal{G}^\downarrow_\tau} \leq \nu' \leq \nu.
\end{equation}
The second inequality is justified as follows.
Since $G_\tau$ does not have loose arcs with respect to $\mu$, it also does not have loose arcs with respect to $\mu^{\mathcal{G}^\uparrow_\tau}$ due to the pointwise minimality of $\mu^{\mathcal{G}^\uparrow_\tau}$.
Hence, $\mu^{\mathcal{G}^\uparrow_\tau}$ is a fixed point of ${\mathcal{G}^\downarrow_\tau}$.
As $\mu^{\mathcal{G}^\uparrow_\tau}\leq \nu$, we get $\mu^{\mathcal{G}^\uparrow_\tau}\leq \nu^{\mathcal{G}^\downarrow_\tau}$ by the pointwise maximality of $\nu^{\mathcal{G}^\downarrow_\tau}$.
The third inequality, on the other hand, follows from the fact that $\nu^{\mathcal{G}^\downarrow_\tau}$ and $\nu'$ are obtained by iterating the operators in $\mathcal{G}^\downarrow_\tau$ on $\nu$. 

First, we show that $\mu^{\mathcal{G}^\uparrow_\tau} = \nu^{\mathcal{G}^\downarrow_\tau}$.
Since $\mu^{\mathcal{G}^\uparrow_\tau}$ is feasible in $G_\tau$, there exists a strategy $\sigma$ for Even such that $G_{\sigma\tau}$ has no violated arcs with respect to $\mu^{\mathcal{G}^\uparrow_\tau}$. In the subgraph $G_{\sigma\tau}$, every node has out-degree 1, so every node can reach a cycle. Pick any cycle $C$ in $G_{\sigma\tau}$ and let $v\in \Pi(C)$. Without loss of generality, we may assume that $\mu^{\mathcal{G}^\uparrow_\tau}(v)<\top$, as otherwise $\mu^{\mathcal{G}^\uparrow_\tau}(w) = \top$ for all the nodes $w$ which can reach $v$ in $G_{\sigma\tau}$. Then, $C$ is even by the Cycle Lemma, which implies that $v\in B(G_\tau)$. From the definition of threshold label, we have $\widehat{\mu}(v)\leq \mu^{\mathcal{G}^\uparrow_\tau}(v)$. As $\nu(v) \leq \widehat{\mu}(v)$, it follows from \eqref{eq:bellmanFord} that $\mu^{\mathcal{G}^\uparrow_\tau}(v) = \nu^{\mathcal{G}^\downarrow_\tau}(v)$. Now, let $u$ be a node which can reach $v$ in $G_{\sigma\tau}$, and consider the corresponding $u$-$v$ path $P$. Since there are no loose arcs in $G_{\sigma\tau}$ with respect to $\nu^{\mathcal{G}^\downarrow_\tau}$, an inductive argument from $v$ along $P$ shows that $\nu^{\mathcal{G}^\downarrow_\tau}(u)\leq \mu^{\mathcal{G}^\uparrow_\tau}(u)$. As $C$ and $u$ were chosen arbitrarily, $\nu^{\mathcal{G}^\downarrow_\tau}\leq \mu^{\mathcal{G}^\uparrow_\tau}$ as desired.

Next, we show that $\nu^{\mathcal{G}^\downarrow_\tau}=\nu'$, or equivalently, $\nu'$ is a fixed point of $\mathcal{G}^\downarrow_\tau$.
Recall that $\nu'$ is the node labeling returned by Algorithm~\ref{algo:Bellman+Ford}.
Let $T:=(n-1)m_\tau$ be the total number of steps carried out by the algorithm, where $m_\tau := |E_\tau|$.
For each $0\leq t\leq T$, let $\nu_t\in \mathcal{L}$ be the node labeling at the end of step $t$.
Note that $\nu_0 = \nu$ and $\nu_T := \nu'$.

\begin{prop}\label{lem:bellmanFord}
For each $0\leq t\leq T$, if an arc $uv\in E_\tau$ is loose with respect to $\nu_t$, then there exists a path $P$ from $u$ to some node $w$ such that $uv\in E(P)$, $\nu_t(w) = \nu(w)$ and $\nu_t$ is feasible in $P$.
\end{prop}

\begin{proof}
We proceed by induction on $t$.
The base case $t=0$ is trivially true with the path $P = uv$.
Now, suppose that the lemma is true for some $t\geq 0$.
Let $pq$ be the arc processed in step $t+1$.
We may assume that $pq$ is loose with respect to $\nu_t$, as otherwise $\nu_{t+1} = \nu_t$ and we are done.
By the inductive hypothesis, there exists a path $P$ from $p$ to some node $r$ such that $pq\in E(P)$, $\nu_t(r) = \nu(r)$ and $\nu_t$ is feasible in $P$.
Pick a shortest such $P$.
Then, $\nu_t(s)<\nu(s)$ for all $s\in V(P)\setminus \{p,r\}$.

Let $uv$ be a loose arc with respect to $\nu_{t+1}$.
There are 2 cases:

\smallskip
\emph{Case 1:} $uv$ was loose with respect to $\nu_t$.
By the inductive hypothesis, there exists a path $Q$ from $u$ to some node $w$ such that $uv\in E(Q)$, $\nu_t(w) = \nu(w)$ and $\nu_t$ is feasible in $Q$.
Pick a shortest such $Q$.
Then, $\nu_t(s)<\nu(s)$ for all $s\in V(Q)\setminus\{u,w\}$.
We may assume that $p\in V(Q)\setminus\{u\}$, as otherwise we are done.
Consider the walk $W = Q_{up}P$ obtained by concatenating the $u$-$p$ subpath of $Q$ and $P$.
Every arc in $W$ is non-violated with respect to $\nu_{t+1}$.

\smallskip
\emph{Case 2:} $uv$ was not loose with respect to $\nu_t$. 
In this case, $v = p$.
Consider the walk $W = uvP$ obtained by concatenating $uv$ and $P$.
Every arc in $W$ is non-violated with respect to $\nu_{t+1}$.

\smallskip
In both cases, we obtain a $u$-$r$ walk $W$ such that $uv\in E(W)$ and $\nu_{t+1}(r) = \nu(r)$ with no violated arcs with respect to $\nu_{t+1}$.
It remains to show that $W$ is a path.
For the purpose of contradiction, let $C$ be the first cycle encountered while traversing $W$ from $u$.
Note that $p\in V(C)$.
Since $\nu_{t+1}$ is feasible in $C$ and $\nu_{t+1}(p)<\nu_t(p)\leq \top$, it follows from the Cycle Lemma that $\nu_{t+1}(s)<\top$ for all $s\in V(C)$.
Furthermore, $C$ is even.
Let $x\in \Pi(C)$, which is a base node.
If $x = u$, then $\widehat{\mu}(x) < \nu_{t+1}(x) \leq \nu(x) \leq \widehat{\mu}(x)$, where the strict inequality is due to $uv\in E(C)$ being loose.
If $x\neq u$, then $\widehat{\mu}(x) \leq \nu_{t+1}(x) < \nu(x) \leq \widehat{\mu}(x)$, where the strict inequality is due to our choice of $P$ and $Q$.
Both scenarios result in a contradiction.
\end{proof}

Let $P$ be a $u$-$w$ path with $i$ arcs.
Let $\bar \nu\in \mathcal{L}$ be a node labeling for which $P$ is tight and $\bar\nu(w) = \nu(w)$.
It is left to prove that $\nu_{im_\tau}(u)\leq \bar\nu(u)$.
This is because Proposition~\ref{lem:bellmanFord} would imply the absence of loose arcs with respect to $\nu_T = \nu'$.
We proceed by induction on $i\geq 0$.
The base case $i=0$ is trivial.
Suppose that the statement is true for all paths with at most $i$ arcs, and consider $|E(P)| = i+1$.
Let $uv$ be the first arc of $P$.
By the inductive hypothesis, $\nu_{im_\tau}(v)\leq \bar\nu(v)$.
Since $uv$ is processed in iteration $i+1$ and $\nu_{t+1}\leq \nu_t$ for all $t$, we obtain $\nu_{(i+1)m_\tau}(u)\leq \bar\nu(u)$.
\end{proof}

Our strategy for computing the node labeling $\nu$ as specified in the statement of Theorem~\ref{thm:bellmanFord} is to find the cycles in Definition \ref{def:threshold}.
In particular, for every base node $v\in B(G_\tau)$, we aim to find a cycle $C$ dominated by $v$ in $G_\tau$ such that $\widehat{\mu}(v)$ can be extended to a node labeling that is feasible in $C$.
To accomplish this goal, we first introduce the notion of \emph{width} in Section \ref{sec:width}, which allows us to evaluate how `good' a cycle is.
It is defined using chains in the poset of subtrees of $T$, where the partial order is given by $\sqsubseteq$.
Then, in Section \ref{sec:raise}, we show how to obtain the desired cycles by computing minimum bottleneck cycles on a suitably defined auxiliary digraph.

\subsection{Width from a Chain of Subtrees in \texorpdfstring{$T$}{T}}
\label{sec:width}

Two ordered trees $T'$ and $T''$ are said to be \emph{distinct} if $T'\not\equiv T''$ (not isomorphic in the sense of Definition~\ref{def:embedding-order}).
Let $h$ be the height of our universal tree $T$.
For $0\leq j \leq h$, denote $\mathcal{T}_j$ as the set of distinct (whole) subtrees rooted at the vertices of depth $h-j$ in $T$. 
For example, $\mathcal{T}_h = \{T\}$, while $\mathcal{T}_0$ contains the trivial tree with a single vertex. 
Since we assumed that all the leaves in $T$ are at the same depth, every tree in $\mathcal{T}_j$ has height $j$.
We denote $\mathcal{T}=\cup_{j=0}^h \mathcal{T}_j$ as the union of all these subtrees.
The sets $\mathcal{T}$ and $\mathcal{T}_j$ form posets with respect to the partial order $\sqsubseteq$. 
The next definition is the usual chain cover of a poset, where we additionally require that the chains form an indexed tuple instead of a set.

\begin{defi}
For $0\leq j \leq h$, let $\mathcal{C}_j=(\mathcal{C}^0_j, \mathcal{C}^1_j, \dots, \mathcal{C}^{\ell}_j)$ be a tuple of chains in the poset $(\mathcal{T}_j,\sqsubseteq)$.
We call $\mathcal{C}_j$ a \emph{cover of $\mathcal{T}_j$} if $\cup_{k=0}^\ell \mathcal{C}^k_j = \mathcal{T}_j$.
A \emph{cover of $\mathcal{T}$} is a tuple $\mathcal{C}=(\mathcal{C}_0,\mathcal{C}_1,\dots, \mathcal{C}_h)$ where $\mathcal{C}_j$ is a cover of $\mathcal{T}_j$ for all $0\leq j\leq h$.
We refer to $\mathcal{C}_j$ as the \emph{$j$th-subcover of $\mathcal{C}$}.
Given a cover $\mathcal{C}$ of $\mathcal{T}$, we denote the trees in the chain $\mathcal{C}^k_j$ as $T^k_{0,j} \sqsubset T^k_{1,j} \sqsubset \cdots  \sqsubset T^k_{|\mathcal{C}^k_j|-1,j}$.
\end{defi}

An example of an ordered tree with its cover is given in Figure \ref{fig:cover}. 
We are ready to introduce the key concept of this subsection.

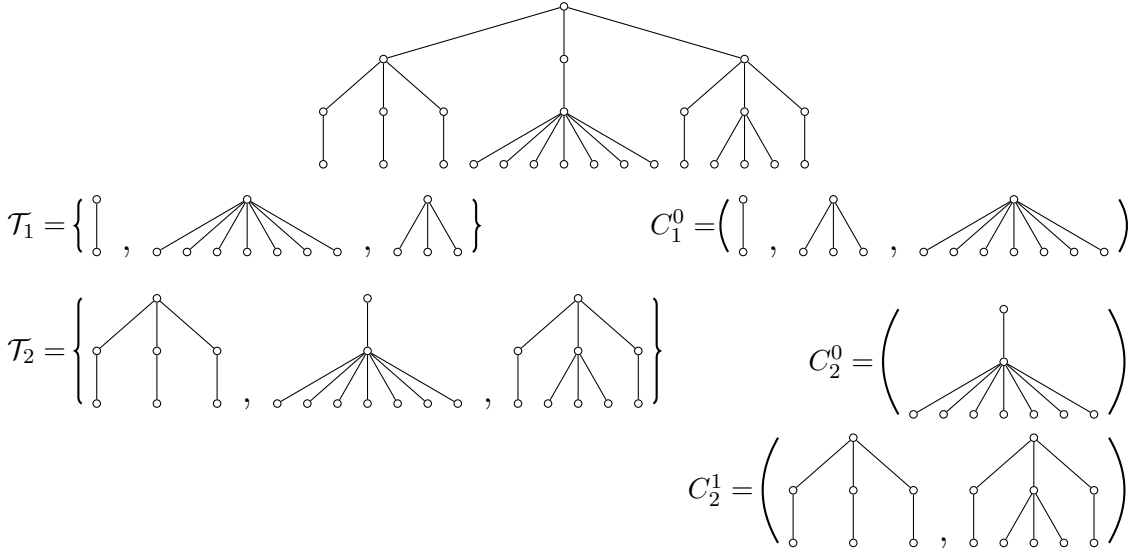
\begin{figure}[h!]
\def\x{0.4}
\def\y{0.7}
\centering
\begin{tikzpicture}
\centering
\begin{scope}[every node/.style={circle,draw,,inner sep=1pt}]
    \node (p) at (0,0) {};
    \node (q1) at (-6*\x,-1*\y) {};  
    \node (q2) at (0*\x,-1*\y) {};  
    \node (q3) at (6*\x,-1*\y) {};   

    \node (p1) at (-8*\x,-2*\y) {};  
    \node (p2) at (-6*\x,-2*\y) {};  
    \node (p3) at (-4*\x,-2*\y) {};   
    \node (p11) at (-8*\x,-3*\y) {};
    \node (p22) at (-6*\x,-3*\y) {};
    \node (p31) at (-4*\x,-3*\y) {};   
    \node (p4) at (0,-2*\y) {};
    \node (p5) at (4*\x,-2*\y) {}; 
    \node (p6) at (6*\x,-2*\y) {}; 
    \node (p7) at (8*\x,-2*\y) {}; 
    \node (p41) at (0-3*\x,-3*\y) {};   
    \node (p42) at (0-2*\x,-3*\y) {};
    \node (p43) at (0-1*\x,-3*\y) {};
    \node (p44) at (0,-3*\y) {};
    \node (p45) at (0+1*\x,-3*\y) {};   
    \node (p46) at (0+2*\x,-3*\y) {};
    \node (p47) at (0+3*\x,-3*\y) {};
    \node (p51) at (4*\x,-3*\y) {};
    \node (p61) at (5*\x,-3*\y) {};
    \node (p62) at (6*\x,-3*\y) {};
    \node (p63) at (7*\x,-3*\y) {};
    \node (p71) at (8*\x,-3*\y) {};
\end{scope}

\begin{scope}[every node/.style={fill=white,circle,font=\scriptsize,inner sep=1pt},
              every edge/.style={draw}]
    \path [-] (p) edge (q1);
    \path [-] (p) edge (q2);
    \path [-] (p) edge (q3);
    \path [-] (q1) edge (p1);
    \path [-] (q1) edge (p2);
    \path [-] (q1) edge (p3);
    \path [-] (q2) edge (p4);
    \path [-] (q3) edge (p5);
    \path [-] (q3) edge (p6);
    \path [-] (q3) edge (p7);
    \path [-] (p1) edge (p11);
    \path [-] (p2) edge (p22);
    \path [-] (p3) edge (p31);
    \path [-] (p4) edge (p41);
    \path [-] (p4) edge (p42);
    \path [-] (p4) edge (p43);
    \path [-] (p4) edge (p44);
    \path [-] (p4) edge (p45);
    \path [-] (p4) edge (p46);
    \path [-] (p4) edge (p47);
    \path [-] (p5) edge (p51);
    \path [-] (p6) edge (p61);
    \path [-] (p6) edge (p62);
    \path [-] (p6) edge (p63);
    \path [-] (p7) edge (p71);
\end{scope}
\end{tikzpicture}

\bigskip
\begin{minipage}{0.49\textwidth}
\raggedright
\begin{tikzpicture}
\raggedright
\node (T) at (-7*\x,-1.5*\y) {$\mathcal{T}_1 = $};
\draw [thick,decorate,decoration = {brace,mirror}] (-5.5*\x,-1*\y) --  (-5.5*\x,-2*\y);
\draw [thick,decorate,decoration = {brace}] (7.5*\x,-1*\y) --  (7.5*\x,-2*\y);

\node (C1) at (-4*\x,-2*\y) {\Large ,};
\node (C1) at (4*\x,-2*\y) {\Large ,};

\begin{scope}[every node/.style={circle,draw,,inner sep=1pt}]
    \node (p1) at (-5*\x,-1*\y) {};  
    \node (p2) at (0*\x,-1*\y) {};  
    \node (p3) at (6*\x,-1*\y) {};   

    \node (p11) at (-5*\x,-2*\y) {};  
    \node (p21) at (-3*\x,-2*\y) {};  
    \node (p22) at (-2*\x,-2*\y) {};  
    \node (p23) at (-1*\x,-2*\y) {}; 
    \node (p24) at (0*\x,-2*\y) {}; 
    \node (p25) at (1*\x,-2*\y) {};  
    \node (p26) at (2*\x,-2*\y) {};  
    \node (p27) at (3*\x,-2*\y) {};     
    \node (p31) at (5*\x,-2*\y) {};  
    \node (p32) at (6*\x,-2*\y) {};  
    \node (p33) at (7*\x,-2*\y) {};
\end{scope}

\begin{scope}[every node/.style={fill=white,circle,font=\scriptsize,inner sep=1pt},
              every edge/.style={draw}]
    \path [-] (p1) edge (p11);
    \path [-] (p2) edge (p21);
    \path [-] (p2) edge (p22);
    \path [-] (p2) edge (p23);
    \path [-] (p2) edge (p24);
    \path [-] (p2) edge (p25);
    \path [-] (p2) edge (p26);
    \path [-] (p2) edge (p27);
    \path [-] (p3) edge (p31);
    \path [-] (p3) edge (p32);
    \path [-] (p3) edge (p33);
\end{scope}
\end{tikzpicture}
\end{minipage}
\begin{minipage}{0.49\textwidth}
\raggedleft
\begin{tikzpicture}
\raggedleft
\node (C) at (-11*\x,-1.5*\y) {$C^0_1 = $};
\draw [thick,decorate,] (-9.5*\x,-1*\y) to[out=240,in=120] (-9.5*\x,-2*\y);
\draw [thick,decorate] (3.5*\x,-1*\y) to[out=300,in=60]  (3.5*\x,-2*\y);

\node (c1) at (-8*\x,-2*\y) {\Large ,};
\node (c2) at (-4*\x,-2*\y) {\Large ,};

\begin{scope}[every node/.style={circle,draw,,inner sep=1pt}]
    \node (p1) at (-9*\x,-1*\y) {};  
    \node (p2) at (-6*\x,-1*\y) {};  
    \node (p3) at (0*\x,-1*\y) {};   

    \node (p11) at (-9*\x,-2*\y) {};  
    \node (p21) at (-7*\x,-2*\y) {};  
    \node (p22) at (-6*\x,-2*\y) {};  
    \node (p23) at (-5*\x,-2*\y) {};   
    \node (p31) at (-3*\x,-2*\y) {};  
    \node (p32) at (-2*\x,-2*\y) {};  
    \node (p33) at (-1*\x,-2*\y) {};
    \node (p34) at (0*\x,-2*\y) {}; 
    \node (p35) at (1*\x,-2*\y) {};  
    \node (p36) at (2*\x,-2*\y) {};  
    \node (p37) at (3*\x,-2*\y) {};     
\end{scope}

\begin{scope}[every node/.style={fill=white,circle,font=\scriptsize,inner sep=1pt},
              every edge/.style={draw}]
    \path [-] (p1) edge (p11);
    \path [-] (p2) edge (p21);
    \path [-] (p2) edge (p22);
    \path [-] (p2) edge (p23);
    \path [-] (p3) edge (p31);
    \path [-] (p3) edge (p32);
    \path [-] (p3) edge (p33);
    \path [-] (p3) edge (p34);
    \path [-] (p3) edge (p35);
    \path [-] (p3) edge (p36);
    \path [-] (p3) edge (p37);
\end{scope}
\end{tikzpicture}
\end{minipage}

\bigskip
\begin{minipage}[t]{0.55\textwidth}
\raggedright
\begin{tikzpicture}
\node (T) at (-11*\x,-2*\y) {$\mathcal{T}_2 = $};
\draw [thick,decorate,decoration = {brace,mirror}] (-9.5*\x,-1*\y) --  (-9.5*\x,-3*\y);
\draw [thick,decorate,decoration = {brace}] (9.5*\x,-1*\y) --  (9.5*\x,-3*\y);

\node (C1) at (-4*\x,-3*\y) {\Large ,};
\node (C1) at (4*\x,-3*\y) {\Large ,};

\begin{scope}[every node/.style={circle,draw,,inner sep=1pt}]
    \node (q1) at (-7*\x,-1*\y) {};  
    \node (q2) at (0*\x,-1*\y) {};  
    \node (q3) at (7*\x,-1*\y) {};   

    \node (p1) at (-9*\x,-2*\y) {};  
    \node (p2) at (-7*\x,-2*\y) {};  
    \node (p3) at (-5*\x,-2*\y) {};   
    \node (p11) at (-9*\x,-3*\y) {};
    \node (p22) at (-7*\x,-3*\y) {};
    \node (p31) at (-5*\x,-3*\y) {};   
    \node (p4) at (0,-2*\y) {};
    \node (p5) at (5*\x,-2*\y) {}; 
    \node (p6) at (7*\x,-2*\y) {}; 
    \node (p7) at (9*\x,-2*\y) {}; 
    \node (p41) at (0-3*\x,-3*\y) {};   
    \node (p42) at (0-2*\x,-3*\y) {};
    \node (p43) at (0-1*\x,-3*\y) {};
    \node (p44) at (0,-3*\y) {};
    \node (p45) at (0+1*\x,-3*\y) {};   
    \node (p46) at (0+2*\x,-3*\y) {};
    \node (p47) at (0+3*\x,-3*\y) {};
    \node (p51) at (5*\x,-3*\y) {};
    \node (p61) at (6*\x,-3*\y) {};
    \node (p62) at (7*\x,-3*\y) {};
    \node (p63) at (8*\x,-3*\y) {};
    \node (p71) at (9*\x,-3*\y) {};
\end{scope}

\begin{scope}[every node/.style={fill=white,circle,font=\scriptsize,inner sep=1pt},
              every edge/.style={draw}]
    \path [-] (q1) edge (p1);
    \path [-] (q1) edge (p2);
    \path [-] (q1) edge (p3);
    \path [-] (q2) edge (p4);
    \path [-] (q3) edge (p5);
    \path [-] (q3) edge (p6);
    \path [-] (q3) edge (p7);
    \path [-] (p1) edge (p11);
    \path [-] (p2) edge (p22);
    \path [-] (p3) edge (p31);
    \path [-] (p4) edge (p41);
    \path [-] (p4) edge (p42);
    \path [-] (p4) edge (p43);
    \path [-] (p4) edge (p44);
    \path [-] (p4) edge (p45);
    \path [-] (p4) edge (p46);
    \path [-] (p4) edge (p47);
    \path [-] (p5) edge (p51);
    \path [-] (p6) edge (p61);
    \path [-] (p6) edge (p62);
    \path [-] (p6) edge (p63);
    \path [-] (p7) edge (p71);
\end{scope}
\end{tikzpicture}
\end{minipage}
\begin{minipage}[t]{0.43\textwidth}
\raggedleft
\begin{tikzpicture}
\raggedleft
\node (C2) at (-5.4*\x,-2*\y) {$C^0_2 = $};
\draw [thick,decorate] (-3.5*\x,-1*\y) to[out=240,in=120]  (-3.5*\x,-3*\y);
\draw [thick,decorate] (3.5*\x,-1*\y) to[out=300,in=60]  (3.5*\x,-3*\y);

\begin{scope}[every node/.style={circle,draw,,inner sep=1pt}]
    \node (q2) at (0*\x,-1*\y) {}; 
    \node (p4) at (0*\x,-2*\y) {}; 
    \node (p41) at (-3*\x,-3*\y) {};   
    \node (p42) at (-2*\x,-3*\y) {};
    \node (p43) at (-1*\x,-3*\y) {};
    \node (p44) at (0*\x,-3*\y) {};
    \node (p45) at (1*\x,-3*\y) {};   
    \node (p46) at (2*\x,-3*\y) {};
    \node (p47) at (3*\x,-3*\y) {};
\end{scope}

\begin{scope}[every node/.style={fill=white,circle,font=\scriptsize,inner sep=1pt},
              every edge/.style={draw}]
    \path [-] (q2) edge (p4);
    \path [-] (p4) edge (p41);
    \path [-] (p4) edge (p42);
    \path [-] (p4) edge (p43);
    \path [-] (p4) edge (p44);
    \path [-] (p4) edge (p45);
    \path [-] (p4) edge (p46);
    \path [-] (p4) edge (p47);
\end{scope}
\end{tikzpicture}

\medskip
\begin{tikzpicture}
\raggedleft
\node (C1) at (-7.4*\x,-2*\y) {$C^1_2 = $};
\draw [thick,decorate] (-5.5*\x,-1*\y) to[out=240,in=120]  (-5.5*\x,-3*\y);
\draw [thick,decorate] (5.5*\x,-1*\y) to[out=300,in=60]  (5.5*\x,-3*\y);
\node (c1) at (-0*\x,-3*\y) {\Large ,};

\begin{scope}[every node/.style={circle,draw,,inner sep=1pt}]
    \node (q1) at (-3*\x,-1*\y) {};  
    \node (q3) at (3*\x,-1*\y) {}; 
    \node (p1) at (-5*\x,-2*\y) {};  
    \node (p2) at (-3*\x,-2*\y) {};  
    \node (p3) at (-1*\x,-2*\y) {};   
    \node (p11) at (-5*\x,-3*\y) {};    
    \node (p22) at (-3*\x,-3*\y) {};
    \node (p31) at (-1*\x,-3*\y) {};

    \node (p5) at (1*\x,-2*\y) {}; 
    \node (p6) at (3*\x,-2*\y) {}; 
    \node (p7) at (5*\x,-2*\y) {}; 
    \node (p51) at (1*\x,-3*\y) {};
    \node (p61) at (2*\x,-3*\y) {};
    \node (p62) at (3*\x,-3*\y) {};
    \node (p63) at (4*\x,-3*\y) {};
    \node (p71) at (5*\x,-3*\y) {};
\end{scope}

\begin{scope}[every node/.style={fill=white,circle,font=\scriptsize,inner sep=1pt},
              every edge/.style={draw}]
    \path [-] (q1) edge (p1);
    \path [-] (q1) edge (p2);
    \path [-] (q1) edge (p3);

    \path [-] (q3) edge (p5);
    \path [-] (q3) edge (p6);
    \path [-] (q3) edge (p7);
    \path [-] (p1) edge (p11);

    \path [-] (p2) edge (p22);

    \path [-] (p3) edge (p31);
    
    \path [-] (p5) edge (p51);
    \path [-] (p6) edge (p61);
    \path [-] (p6) edge (p62);
    \path [-] (p6) edge (p63);
    \path [-] (p7) edge (p71);
\end{scope}
\end{tikzpicture} 
\end{minipage}
\caption{An ordered tree $T$ of height 3, and a cover $\mathcal{C}_j$ of $\mathcal{T}_j$ for all $0<j<3$. Recall that $\mathcal{T}_j$ is the set of distinct subtrees of $T$ rooted at depth $3-j$, while $\mathcal{C}^k_j$ is the $k$th chain in $\mathcal{C}_j$.}
\label{fig:cover}
\end{figure}

\begin{defi}
Let $\mathcal{C}$ be a cover of $\mathcal{T}$. Let $H$ be a subgraph of $G_\tau$ and $j=\ceil{\pi(H)/2}$.
For a fixed chain $\mathcal{C}^k_j$ in $\mathcal{C}_j$, the \emph{$k$th-width} of $H$, denoted $\alpha^k_\mathcal{C}(H)$, is the smallest integer $i\geq 0$ such that $H$ admits a finite feasible node labeling $\nu:V(H)\rightarrow L(T^k_{i,j})$. 
If such an $i$ does not exist, then we set $\alpha^k_{\mathcal{C}}(H) = \infty$.
\end{defi}

Recall that $T^k_{i,j}$ is the $(i+1)$-th smallest tree in the chain $\mathcal{C}^k_j$.
We are mainly interested in the case when $H$ is a cycle, and write $\alpha^k(H)$ whenever the cover $\mathcal{C}$ is clear from context.
Observe that the definition above requires $\nu(v)\neq \top$ for all $v\in V(H)$.
Hence, an odd cycle has infinite $k$th-width by the Cycle Lemma.
As $(\mathcal{C}^k_j,\sqsubseteq)$ is a chain, $H$ admits a finite feasible node labeling $\nu:V(H)\rightarrow L(T^k_{i,j})$ for all $\alpha^k(H)\leq i < |\mathcal{C}^k_j|$.
The next lemma illustrates the connection between the $k$th-width of an even cycle and its path decomposition.

\begin{lem}\label{lem:bit_req_cycle}
Let $\mathcal{C}$ be a cover of $\mathcal{T}$.
For an even cycle $C$, let $\Pi(C) = \set{v_1, v_2, \dots, v_\ell}$ and $j=\pi(C)/2$. 
Decompose $C$ into arc-disjoint paths $P_1,P_2,\dots, P_\ell$ such that each $P_i$ ends at $v_i$.
Then, $\alpha^k(C) = \max_{i\in [\ell]}\alpha^k(P_i)$ for all $0\leq k < |\mathcal{C}_j|$.
\end{lem}

\begin{proof} 
Fix a chain $\mathcal{C}^k_j$ and let $\alpha^* = \max_{i\in [\ell]}\alpha^k(P_i)$.
Clearly, any node labeling which is feasible in $C$ is also feasible in $P_i$ for all $i\in [\ell]$. 
So, $\alpha^k(C)\geq \alpha^*$.
Next, we prove the reverse inequality.
For each $i\in [\ell]$, there exists a node labeling $\nu_i:V(P_i)\rightarrow L(T^k_{\alpha^*,j})$ which is feasible in $P_i$.
Since $v_i$ has out-degree 0 in $P_i$, 
we may assume that $\nu_i(v_i) = \min L(T^k_{\alpha^*,j})$ for all $i\in [\ell]$ without loss of generality.
Let us define a new node labeling $\nu$ as follows.
For $v\in V(C)$, set $\nu(v) := \nu_i(v)$ where $i\in [\ell]$ is the unique index such that $\delta^-_{P_i}(v)\neq \emptyset$.
Then, $\nu$ is feasible in $C$ because $\nu(v)\neq \top$ for all $v\in V(C)$.
Hence, $\alpha^k(C)\leq \alpha^*$.
\end{proof}

For a base node $v\in B(G_\tau)$, let us consider the cycles in $G_\tau$ which are dominated by $v$.
Among these cycles, we are interested in finding one with the smallest $k$th-width.
So, we extend the notion of $k$th-width to base nodes in the following way.

\begin{defi}
Let $\mathcal{C}$ be a cover of $\mathcal{T}$.
Let $v\in B(G_\tau)$ be a base node and $j=\pi(v)/2$.
For $0\leq k < |\mathcal{C}_j|$, the \emph{$k$th-width} of $v$ is defined as
\[\alpha^k_{\mathcal{C}}(v):=\min\set{\alpha^k_{\mathcal{C}}(C):C \text{ is a cycle dominated by $v$ in $G_\tau$}}.\]
\end{defi}

Again, we write $\alpha^k(v)$ whenever the cover $\mathcal{C}$ is clear from context.
Observe that $T^k_{\alpha^k(v),\pi(v)/2}$ is the smallest tree in the chain $\mathcal{C}^k_{\pi(v)/2}$ for which some cycle dominated by $v$ in $G_\tau$ admits a finite feasible node labeling.

Given a leaf $\xi\in L(T)$ and integers $i,j,k\in \Z_{\geq 0}$, the following subroutine locates a member of the chain $\mathcal{C}^k_j$ in $T$ whose leaves are at least $\xi$ and into which $T^k_{i,j}$ is embeddable.

\begin{myframe}{Raise{$(\xi,i,j,k)$}}
  Given a leaf $\xi\in L(T)$ and integers $i\geq 0$, $j\in [h]$, $k\geq 0$, return the smallest leaf $\xi'\in L(T)$ such that (1) $\xi'\geq \xi$; and (2) $\xi'$ is the smallest leaf in the subtree $T^k_{i',j}$ for some $i'\geq i$. 
  If $\xi'$ does not exist, then return $\top$.
\end{myframe}

This subroutine allows us to relate the $k$th-width $\alpha^k(v)$ of a base node $v$ to its threshold label $\widehat{\mu}(v)$.
For $0\leq k < |\mathcal{C}_{\pi(v)/2}|$, the element returned by \textsc{Raise}$(\mu(v),\alpha^k(v),\frac{\pi(v)}{2},k)$ is at least $\widehat{\mu}(v)$.
Moreover, the smallest such element over all $k$ is precisely $\widehat{\mu}(v)$.

\begin{lem}\label{lem:threshold}
Let $\mu\in \mathcal{L}$ be a node labeling such that $G_\tau$ does not have loose arcs.
Let $v\in B(G_\tau)$ be a base node and $j=\pi(v)/2$.
If $\xi^k\in \bar{L}(T)$ is the element returned by {\sc Raise}$(\mu(v),\alpha^k(v),j,k)$, then
\[\widehat{\mu}(v) = \min_{0\leq k < |\mathcal{C}_j|} \xi^k.\]
\end{lem}

\begin{proof}
First, we prove that $\widehat{\mu}(v)\leq \xi^k$ for all $0\leq k< |\mathcal{C}_j|$.
Fix a $k$ and assume that $\xi^k\neq \top$.
Then, $\xi^k$ is the smallest leaf in a subtree $T^k_{i,j}$ of $T$ for some $i\geq \alpha^k(v)$. 
From the definition of $\alpha^k(v)$, there exists a cycle $C$ dominated by $v$ in $G_\tau$ which admits a finite feasible node labeling $\nu:V(C)\rightarrow L(T^k_{i,j})$.
Since $\pi(v) = 2j$ and $\nu(w)\neq \top$ for all $w\in V(C)$, we may assume that $\nu(v) = \min L(T^k_{i,j})$. 
Let us define a new node labeling $\tilde{\mu}:V\rightarrow \bar{L}(T)$ as follows.
If $w\in V(C)$, set $\tilde{\mu}(w)$ as the concatenation of $\xi^k|_{2j}$ and $\nu(v)$, which is a leaf in $T$. 
Otherwise, set $\tilde{\mu}(w) := \top$.
Then, $\tilde{\mu}$ is feasible in $C$ and $\tilde{\mu}(v) = \xi^k \geq \mu(v)$, where the equality is due to our assumption $\nu(v) = \min L(T^k_{i,j})$.
Hence, $\tilde{\mu}(v)\geq \widehat{\mu}(v)$ from the definition of $\widehat{\mu}(v)$.

\smallskip
It is left to show that $\widehat{\mu}(v) \geq  \xi^k$ for some $0\leq k < |\mathcal{C}_j|$.
We may assume that $\widehat{\mu}(v)\neq \top$.
Let $T^k_{i,j}$ be the subtree of $T$ rooted at $\widehat{\mu}(v)|_{2j}$.
From the definition of $\widehat{\mu}(v)$, there exists a node labeling $\tilde{\mu}\in \mathcal{L}$ and a cycle $C$ dominated by $v$ in $G_\tau$ such that $\tilde{\mu}$ is feasible in $C$ and $\tilde{\mu}(v) = \widehat{\mu}(v) \geq \mu(v)$. 
By the Cycle Lemma, $\tilde{\mu}(v)|_{2j} = \tilde{\mu}(w)|_{2j}$ for all $w\in V(C)$.
Since $\tilde{\mu}$ is feasible in $C$ and $\tilde{\mu}(w)\neq \top$ for all $w\in V(C)$, it follows that $i\geq \alpha^k(C)\geq \alpha^k(v)$.
It remains to show that $\widehat\mu(v)$ is the smallest leaf in the subtree of $T$ rooted at $\widehat\mu(v)|_{2j}$, as this would imply $\widehat\mu(v)$ is at least the output of {\sc Raise}$(\mu(v),i,j,k)$. Suppose otherwise for a contradiction. Then, $\widehat\mu(v) = \mu(v)$. Let $u$ be the unique out-neighbour of $v$ in $C$, and consider the $u$-$v$ subpath $P$ of $C$. Since $\tilde\mu(v)\geq \mu(v)$ and there are no loose arcs in $P$ with respect to $\mu$, an inductive argument from $v$ along $P$ shows that $\tilde\mu(w)\geq \mu(w)$ for all $w\in V(P)$. It follows that $\mu(v)|_{2j} = \tilde\mu(v)|_{2j} = \tilde\mu(u)|_{2j}\geq \mu(u)|_{2j}$, so $vu$ is loose with respect to $\mu$. This is a contradiction.
\end{proof}

The necessary number of chains in the subcover $\mathcal{C}_{\pi(v)/2}$ can be large if $T$ is an arbitrary ordered tree.
Fortunately, the universal trees constructed in the literature admit covers with small subcovers.
We prove that a succinct $(n,h)$-universal tree has a cover with only 1 chain per subcover in Section \ref{sec:succinct}, whereas a succinct Strahler $(n,h)$-universal tree (introduced by Daviaud et al.~\cite{conf/icalp/DaviaudJT20}) has a cover with at most $\log n$ chains per subcover in Section \ref{sec:strahler}.

Let $\rho(T,\mathcal{C})$ denote the running time of {\sc Raise}.
We provide efficient implementations of {\sc Raise} for succinct universal trees and succinct Strahler universal trees in Sections \ref{sec:succinct}--\ref{sec:strahler}.
They have the same running time as \textsc{Tighten}, i.e., $\rho(T,\mathcal{C})=O(\log n \log h)$.

\subsection{Estimating the Width of Base Nodes}
\label{sec:raise}
In light of the previous discussion, we can now focus on computing the $k$th-width of a base node $w\in B(G_\tau)$. 
Fix a $0\leq k < |\mathcal{C}_{\pi(w)/2}|$.
Since we ultimately need a label that lies between the fixed point $\mu^{\mathcal{G}^\uparrow_\tau}(w)$ and the threshold label $\widehat{\mu}(w)$ in order to initialize Algorithm \ref{algo:Bellman+Ford}, it suffices to compute a `good' under-estimation of $\alpha^k(w)$. 
In this subsection, we reduce this problem to computing a minimum bottleneck cycle in an auxiliary digraph $D$ with nonnegative arc costs $c^k \geq 0$.

For a base node $w\in B(G_\tau)$, let $K_w$ denote the SCC containing $w$ in $(G_\tau)_{\pi(w)}$, the subgraph of $G_\tau$ induced by nodes with priority at most $\pi(w)$.
Note that $K_v = K_w$ for all $v\in \Pi(K_w)$.
Let $K'_w$ be the subgraph of $K_w$ after deleting $\delta^-(v)$ for all $v\in \Pi(K_w)\setminus \{w\}$.
Then, we define $J_w$ as the subgraph of $K'_w$ induced by nodes which can reach $w$ in $K'_w$.
These are the nodes which can reach $w$ in $K_w$ without encountering an intermediate node of priority $\pi(w)$.
We remark that $J_w$ may not be an induced subgraph of $K_w$.

\begin{figure}[ht]
\begin{minipage}{0.49\textwidth}
\centering
\begin{tikzpicture}[scale=0.7]

\fill [fill=gray!20] plot [mark=none, smooth cycle] coordinates {(-4.2,-1.2) (-1.5,-1.2) (-0.5,-2.6) (1,-2.6) (2.5,-0.3) (2.5,2.4) (-4,2.2)};
\fill [fill=gray!40] plot [mark=none, smooth cycle] coordinates {(-4.1,-1) (-1.4,-1) (-0.5,-2.4) (1,-2.4) (2.4,0) (2.4,2.2) (-0.5,2.3) (-1.5,1) (-4.2,1)};

\begin{scope}[every node/.style={circle,draw,inner sep = 2pt}]
    \node (B) at (2,2) {1};  
    \node[label={right,xshift=1mm}:{$w_4$}] (D) at (2,0) {4}; 
    \node (G) at (-2,-2) {5};
    \node[label=above:{$w_2$}] (H) at (-2,2) {4}; 
    \node (I) at (-2,0) {2};
    
\end{scope}
\begin{scope}[every node/.style={rectangle,draw}]
    \node (A) at (0,0) {1};  
    \node[label={above,yshift=1mm}:{$w_1$}] (C) at (0,2) {2};  
    \node[label={left,xshift=-1mm}:{$w_3$}] (E) at (-4,0) {4};
    \node (F) at (0,-2) {3}; 
\end{scope}

\begin{scope}[>={Stealth[black]},
              every edge/.style={draw}]
    \path [->] (A) edge (B);
    \path [->] (B) edge (C);
    \path [->] (C) edge (A);
    \path [->] (E) edge (I);
    \path [->] (I) edge (A);
    \path [->] (C) edge (H);
    \path [->] (G) edge (F);
    \path [->] (F) edge [bend right] (D);
    \path [->] (A) edge [bend right] (F);
    \path [->] (F) edge [bend right] (A);
    \path [->] (D) edge (B);
    \path [->] (E) edge [bend right] (G);
    \path [->] (H) edge [bend right] (E);
\end{scope}
\end{tikzpicture}
\end{minipage}
\begin{minipage}{0.49\textwidth}
\centering
\begin{tikzpicture}
\begin{scope}[every node/.style={circle,draw, inner sep=2pt}]
    \node[label=below:{$w_1$}] (A) at (-2,0) {};
    \node[label=above:{$w_2$}] (B) at (1,0.7) {};
    \node[label=below:{$w_3$}] (C) at (0,-0.7) {};
    \node[label=below:{$w_4$}] (D) at (2,-0.7) {};
    
\end{scope}

\begin{scope}[>={Stealth[black]},
              every edge/.style={draw}]
    \path [->] (A) edge [out=45, in=135, loop] (A);
    \path [->] (B) edge [bend right] (C);
    \path [->] (C) edge [bend right] (B);
    \path [->] (C) edge [bend right] (D);
    \path [->] (D) edge [bend right] (B);
    \path [->] (D) edge [out=315, in=45, loop] (D);
\end{scope}
\end{tikzpicture}
\end{minipage}
\caption{An example of a 1-player game $(G_\tau,\pi)$ for Even is given on the left, with its auxiliary digraph $D$ on the right. Nodes in $V_0$ and $V_1$ are drawn as squares and circles respectively. Base nodes are labeled as $w_1$, $w_2$, $w_3$, $w_4$. The light gray region is $K_{w_4}$, while the dark gray region is~$J_{w_4}$.}
\label{fig:aux_graph}
\end{figure}
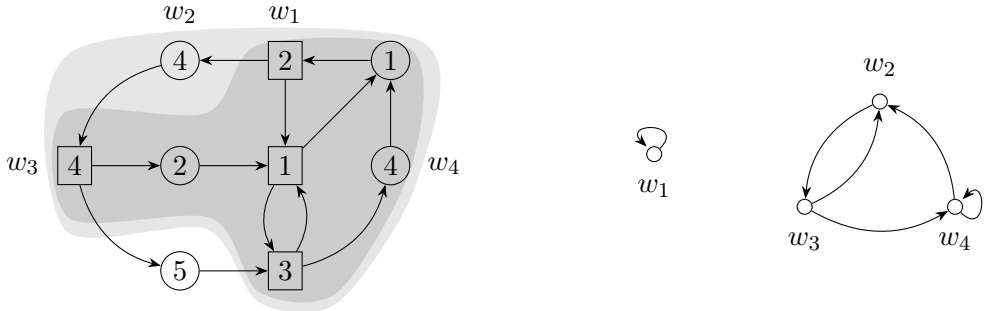

The auxiliary digraph $D$ is constructed as follows.
Its node set is $B(G_\tau)$.
For every ordered pair $(v,w)$ of distinct base nodes, add the arc $vw$ if $v\in \Pi(J_w)$.
For every base node $w$, we also add the self-loop $ww$ if $w$ has an outgoing arc in $J_w$.
Observe that $vw\in E(D)$ if and only if (1) $v\neq w$, $v\in \Pi(K_w)$ and $v$ can reach $w$ by only seeing smaller priorities on the intermediate nodes; or (2) $v=w$ and there exists a cycle $C$ in $G_\tau$ such that $\Pi(C) = \{w\}$.
Hence, every component of $D$ is strongly connected as its vertex set is $\Pi(K_w)$ for some base node $w\in B(G_\tau)$ (see Figure \ref{fig:aux_graph} for an example).
For $w\in B(G_\tau)$, we denote $D_w$ as the component in $D$ which contains $w$.

To finish the description of $D$, all that remains is to assign the arc costs $c^k$.
Note that the graph structure of $D$ is independent of $k$.
We first state a useful lemma which relates the outputs of $\drop(\cdot,\cdot)$ when applied to node labelings from comparable trees.

\begin{lem}\label{lem:finite_comparable}
Given integers $0\leq i_1\leq i_2$ and $j,k\geq 0$, let $f$ be an injective and order-preserving homomorphism from $T^k_{i_1,j}$ to $T^k_{i_2,j}$ with the extension $f(\top) := \top$.
Let $\nu_1:V\to \bar L(T^k_{i_1,j})$ and $\nu_2:V\to \bar L(T^k_{i_2,j})$ be node labelings such that $f(\nu_1(u))\geq \nu_2(u)$ for all $u\in V$.
For any arc $vw\in E_\tau$, we have $f(\drop(\nu_1,vw)) \geq \drop(\nu_2,vw)$.
\end{lem}

\begin{proof}
We only need to check the label at node $v$.
If $vw$ is not loose with respect to $\nu_1$, then $\drop(\nu_1,vw) = \nu_1(v)$.
Hence, $f(\drop(\nu_1,vw))  = f(\nu_1(v)) \geq \nu_2(v) \geq \drop(\nu_2,vw)$.
So, we may assume that $vw$ is loose with respect to $\nu_1$.
Note that this implies $\nu_1(w)< \top$, which in turn implies $\nu_2(w)\leq f(\nu_1(w)) < \top$.
Let $\xi\in L(T^k_{i_2,j})$ be the unique element such that $vw$ is tight with respect to $\nu_2$ after setting $\nu_2(v) = \xi$.
Since $\nu_2(w)|_{\pi(v)}\leq  f(\nu_1(w))|_{\pi(v)} = f(\nu_1(w)|_{\pi(v)})$, we obtain $\drop(\nu_2,vw)\leq \xi \leq f(\drop(\nu_1,vw))$.
\end{proof}

Fix a base node $w\in B(G_\tau)$ and let $j=\pi(w)/2$.
Let $\mathcal{J}^\downarrow_w=\set{\text{Drop}_e:e\in E(J_w)}$ denote the set of Drop operators in $J_w$.
For $0\leq i < |\mathcal{C}^k_j|$, define $\lambda^k_{i,w}:V(J_w)\rightarrow \bar{L}(T^k_{i,j})$ as the greatest simultaneous fixed point of $\mathcal{J}^\downarrow_w$ which satisfies $\lambda^k_{i,w}(w) = \min L(T^k_{i,j})$. %
By Proposition~\ref{prop:Tarski}, $\lambda^k_{i,w}$ exists and can be obtained by applying the operators in $\mathcal{J}^\downarrow_w$ to the node labeling that assigns $\min L(T^k_{i,j})$ to $w$, and $\top$ to other nodes, until convergence.

For every arc $vw\in E(D)$, we specify a range $\underline{c}^k(vw)\leq c^k(vw) \leq \overline{c}^k(vw)$ of permissible arc costs.
The lower bound $\underline{c}^k(vw)$ is the smallest index of a tree in the chain $\mathcal{C}^k_j$ for which $\lambda^k_{i,w}$ is finite at an out-neighbour of $v$, i.e.,
\begin{equation} \label{eq:cost-lower-bound}
  \underline{c}^k(vw) := \min\set{i:\lambda^k_{i,w}(u)\neq \top \text{ for some }u\in N^+_{J_w}(v)}.
\end{equation}
We use the convention that the minimum over the empty set is $\infty$. By Lemma~\ref{lem:finite_comparable}, observe that if $\lambda^k_{i,w}(u)\neq \top$, then $\lambda^k_{i',w}(u)\neq \top$ for all $i'\geq i$.

On the other hand, the upper bound $\overline{c}^k(vw)$ is the smallest $k$th-width of a path from an out-neighbour of $v$ to $w$ in $J_w$, i.e.
\begin{equation} \label{eq:cost-upper-bound}
  \overline{c}^k(vw) := \min\set{\alpha^k(P): P \text{ is a $u$-$w$ path in } J_w \text{ where }u\in N^+_{J_w}(v)}.
\end{equation}
Recall that the $k$th-width of a path $P$ is the smallest index of a tree in the chain $\mathcal{C}^k_j$ for which $P$ has a finite feasible node labeling.

\begin{lem}
For every arc $vw\in E(D)$, we have $\underline{c}^k(vw) \leq \overline{c}^k(vw)$.
\end{lem}

\begin{proof}
We may assume that $\overline{c}^k(vw) < \infty$.
Let $P$ be a $u$-$w$ path in $J_w$ where $u\in N^+_{J_w}(v)$ and $\alpha^k(P) = \overline{c}^k(vw)$.
Let $j=\pi(w)/2$.
From the definition of $\alpha^k(P)$, there exists a finite node labeling $\nu:V(P)\rightarrow L(T^k_{\alpha^k(P),j})$ which is feasible in $P$.
Note that $\nu(s)\geq \lambda^k_{\alpha^k(P),w}(s)$ for all $s\in V(P)$ because $\nu(w) \geq \min L(T^k_{\alpha^k(P),j}) = \lambda^k_{\alpha^k(P),w}(w)$ and there are no loose arcs in $P$ with respect to $\lambda^k_{\alpha^k(P),w}$.
It follows that $\lambda^k_{\alpha^k(P),w}(u) \leq \nu(u) < \top$, which gives $\underline{c}^k(vw) \leq \alpha^k(P)$ as desired.
\end{proof}

For a cycle $C$ in $D$, we define its (bottleneck) $c^k$-\emph{cost} as $c^k(C):=\max_{e\in E(C)}c^k(e)$.
Note that self-loops in $D$ are considered cycles.
The next theorem allows us to obtain the desired initial node labeling $\nu$ for Algorithm \ref{algo:Bellman+Ford} by computing minimum bottleneck cycles in~$D$.

\begin{thm}\label{thm:bottleneck_cycle}
Let $\mathcal{C}$ be a cover of $\mathcal{T}$.
Let $\mu\in \mathcal{L}$ be a node labeling such that $G_\tau$ does not have loose arcs.
Fix a base node $w\in B(G_\tau)$ and let $j=\pi(w)/2$. 
For $0\leq k < |\mathcal{C}_j|$, let $c^k$ be arc costs in $D_w$ where $\underline{c}^k\leq c^k \leq \overline{c}^k$, let $i^k$ be the minimum $c^k$-cost of a cycle containing $w$ in $D_w$, and let $\xi^k$ be the label returned by {\sc Raise}$(\mu(w),i^k,j,k)$.\footnote{We set $\xi^k = \top$ if $i^k = \infty$.}
Then, 
\[\mu^{\mathcal{G}^\uparrow_\tau}(w)\leq \min_{0\leq k < |\mathcal{C}_j|}\xi^k\leq \widehat{\mu}(w).\]
\end{thm}

\begin{proof}
We first prove the lower bound.
Fix a $0\leq k < |\mathcal{C}_j|$ and assume that $\xi^k\neq \top$.
Let $C$ be a minimum $c^k$-cost cycle in $D_w$ containing $w$.
Denote $C = (w_1,w_2,\dots, w_\ell)$ where $w = w_1 = w_\ell$; note that these are base nodes in $G_\tau$.
For every $s \in  [\ell]$, let $\nu_s:V\rightarrow \bar{L}(T)$ be the node labeling defined by $\nu_s(w_s) := \xi^k$ and $\nu_s(v) := \top$ for all $v\neq w_s$.
Consider the greatest simultaneous fixed point $\nu_s^{\mathcal{J}^\downarrow_{w_s}}$, which can be obtained by applying the Drop operators in $\mathcal{J}^\downarrow_{w_s}$ to $\nu_s$ by Proposition~\ref{prop:Tarski}.
Since $\pi(J_{w_s}) = \pi(w_s) = j$, we have $\nu_s^{\mathcal{J}^\downarrow_{w_s}}(w_s) = \nu_s(w_s) = \xi^k$ because $\xi^k$ is the smallest leaf in the subtree of $T$ rooted at $\xi^k|_j$.
This also implies that $\nu^{\mathcal{J}^\downarrow_{w_s}}_s(v)\geq \xi^k$ for all $v\in V$.
Furthermore, $\nu^{\mathcal{J}^\downarrow_{w_s}}_s$ is feasible in $G_\tau\setminus \delta^+(w_s)$, as $\nu_s$ is feasible in $G_\tau\setminus \delta^+(w_s)$ and feasibility is maintained throughout the application of Drop operators.

\begin{clm}\label{clm:out_neighbour}
For each $1<s\leq \ell$, there exists a node $u_s\in N^+_{J_{w_s}}(w_{s-1})$ such that $\nu^{\mathcal{J}^\downarrow_{w_s}}_s(u_s)|_j = \xi^k|_j$.
\end{clm}

\begin{proof}
Fix an $s$ and let $i=\underline{c}^k(w_{s-1}w_s)$.
From the definition of $\underline{c}^k$ in~\eqref{eq:cost-lower-bound}, there exists a node $u_s\in N^+_{J_{w_s}}(w_{s-1})$ such that $\lambda^k_{i,w_s}(u_s)\neq \top$.
Recall that $\lambda^k_{i,w_s}:V(J_{w_s})\rightarrow \bar{L}(T^k_{i,j})$ is the greatest simultaneous fixed point of $\mathcal{J}^\downarrow_{w_s}$ subject to $\lambda^k_{i,w_s}(w_s) = \min L(T^k_{i,j})$.
We will show that $\nu^{\mathcal{J}^\downarrow_{w_s}}_s(u_s)|_j = \xi^k|_j$.
From our construction of $\nu_s$, observe that $\nu^{\mathcal{J}^\downarrow_{w_s}}_s:V\rightarrow \bar{L}(T)$ is the greatest simultaneous fixed point of $\mathcal{J}^\downarrow_{w_s}$ subject to $\nu^{\mathcal{J}^\downarrow_{w_s}}_s(w_s) = \xi^k$.
Since $\xi^k$ is returned by \textsc{Raise}$(\mu(w),i^k,j,k)$, it is the smallest leaf in a copy of $T^k_{i',j}$ in the main tree $T$ for some $i'\geq i^k$.
As $c^k(C) = i^k$ and $w_{s-1}w_s\in E(C)$, we have $i^k\geq c^k(w_{s-1}w_s)\geq \underline{c}^k(w_{s-1}w_s) = i$.
So, $i'\geq i$ and $T^k_{i,j} \sqsubseteq T^k_{i',j}$.
By Lemma \ref{lem:finite_comparable}, $\lambda^k_{i',w_s}(u_s)\neq \top$, which implies $\nu_s^{\mathcal{J}^\downarrow_{w_s}}(u_s)|_j = \xi^k|_j$. 
\end{proof}

Now, consider the node labeling $\nu$ defined by $\nu(v) := \min_{s\in [\ell]}\nu^{\mathcal{J}^\downarrow_{w_s}}_s(v)$ for all $v\in V$.
Note that $\nu(w_s) = \xi^k$ for all $s\in [\ell]$.
It suffices to show that $\nu$ is feasible in $G_\tau$ and $\nu\geq \mu$.
This is because it would imply $\nu\geq \mu^{\mathcal{G}^\uparrow_\tau}$ by the pointwise minimality of $\mu^{\mathcal{G}^\uparrow_\tau}$.
In particular, $\xi^k = \nu(w)\geq \mu^{\mathcal{G}^\uparrow_\tau}(w)$.

\smallskip

We first prove that $\nu$ is feasible in $G_\tau$. 
Since $\nu^{\mathcal{J}^\downarrow_{w_s}}_s$ is feasible in $G_\tau\setminus \delta^+(w_s)$ for all $s \in [\ell]$, $\nu$ is feasible in $G_\tau\setminus \cup_{s=1}^\ell \delta^+(w_s)$.
This is because for each node $u\in V\setminus \{w_1,w_2,\dots,w_\ell\}$, $\nu(u) = \nu_s^{\mathcal{J}^\downarrow_{w_s}}(u)$ for some $s\in [\ell]$ and furthermore, there exists an arc $uv\in \delta^+_{G_\tau}(u)$ which is non-violated with respect to $\nu_s^{\mathcal{J}^\downarrow_{w_s}}$.
Since $\nu_s^{\mathcal{J}^\downarrow_{w_s}}(v)\geq \nu(v)$, $uv$ is also non-violated with respect to $\nu$.
It is left to show that $w_s$ has a non-violated outgoing arc in $G_\tau$ with respect to $\nu$ for all $s\in [\ell]$.
Fix $1<s\leq \ell$.
By Claim \ref{clm:out_neighbour}, there exists a node $u_s\in N^+_{J_{w_s}}(w_{s-1})$ such that $\nu^{\mathcal{J}^\downarrow_{w_s}}_s(u_s)|_j = \xi^k|_j$.
Hence, $\nu(u_s)|_j = \xi^k|_j$.
As $\nu(w_{s-1}) = \xi^k$ and $\pi(w_{s-1}) = j$, the arc $w_{s-1}u_s$ is non-violated with respect to $\nu$.

\smallskip

Next, we show that $\nu^{\mathcal{J}^\downarrow_{w_s}}_s\geq \mu$ for all $s\in [\ell]$, which would imply $\nu\geq \mu$ as desired.
We proceed by induction on $s$.
For the base case $s=\ell$, we know that $\nu^{\mathcal{J}^\downarrow_{w_\ell}}_\ell$ is the greatest simultaneous fixed point of $\mathcal{J}^\downarrow_{w_\ell}$ subject to $\nu^{\mathcal{J}^\downarrow_{w_\ell}}_\ell(w_\ell) = \xi^k$.
Observe that $\mu$ is also a simultaneous fixed point of $\mathcal{J}^\downarrow_{w_\ell}$ because there are no loose arcs in $G_\tau$ with respect to $\mu$. 
Since $\nu^{\mathcal{J}^\downarrow_{w_\ell}}_\ell(w_\ell) = \xi^k \geq \mu(w_\ell)$, where the inequality is due to $w_\ell=w$, we get $\nu^{\mathcal{J}^\downarrow_{w_\ell}}_\ell\geq \mu$.
For the inductive step, suppose that $\nu^{\mathcal{J}^\downarrow_{w_s}}_s\geq \mu$ for some $1<s\leq \ell$.
By Claim \ref{clm:out_neighbour}, there exists a node $u_s\in N^+_{J_{w_s}}(w_{s-1})$ such that $\nu^{\mathcal{J}^\downarrow_{w_s}}_s(u_s)|_j = \xi^k|_j$.
Then, $\mu(w_{s-1})\leq \nu_s^{\mathcal{J}^\downarrow_{w_s}}(w_{s-1}) = \xi^k$.
The inequality follows from the inductive hypothesis, whereas the equality is due to the tightness of $w_{s-1}u_s$ with respect to $\nu^{\mathcal{J}^\downarrow_{w_s}}_s$. 
As $\nu^{\mathcal{J}^\downarrow_{w_{s-1}}}_{s-1}$ is the greatest simultaneous fixed point of $\mathcal{J}^\downarrow_{w_{s-1}}$ subject to $\nu^{\mathcal{J}^\downarrow_{w_{s-1}}}_{s-1}(w_{s-1}) = \xi^k\geq \mu(w_{s-1})$, we obtain $\nu^{\mathcal{J}^\downarrow_{w_{s-1}}}_{s-1}\geq \mu$ because $\mu$ is also a simultaneous fixed point of $\mathcal{J}^\downarrow_{w_{s-1}}$. 

\bigskip

Lastly, we prove the upper bound, i.e., $\xi^k\leq \widehat{\mu}(w)$ for some $0\leq k < |\mathcal{C}_j|$.
We may assume that $\widehat{\mu}(w)\neq \top$.
By Lemma \ref{lem:threshold}, $\widehat{\mu}(w)$ is returned by {\sc Raise}$(\mu(w),\alpha^k(w),j,k)$ for some $0\leq k< |\mathcal{C}_j|$. 
Let $H$ be a cycle in $G_\tau$ such that $w\in \Pi(H)$ and $\alpha^k(H) = \alpha^k(w)$.
Let $\Pi(H) = \set{w_1,w_2,\dots,w_r}$ and decompose $H$ into arc-disjoint paths $P_1, P_2, \dots, P_r$ such that $P_s$ is a $w_{s-1}$-$w_s$ path for all $s\in [r]$, where $w_0 := w_r$.
Since each $P_s$ lies in the subgraph $J_{w_s}$, we have $w_{s-1}w_s\in E(D)$ for all $s\in [r]$, and their union induces a cycle $H'$ containing $w$ in~$D$. 
Then,
\[i^k \leq c^k(H') = \max_{s\in [r]}c^k(w_{s-1}w_s) \leq \max_{s\in[r]}\overline{c}^k(w_{s-1}w_s) \leq \max_{s\in [r]}\alpha^k(P_s) = \alpha^k(H) = \alpha^k(w).\]
The third inequality follows from the definition of $\overline{c}^k$ in ~\eqref{eq:cost-upper-bound}, while the second equality is due to Lemma~\ref{lem:bit_req_cycle}.
It follows that $\xi^k\leq \widehat{\mu}(w)$ because \textsc{Raise} is monotone with respect to its second argument.
\end{proof}

\subsection{The Label-Correcting Algorithm}
The overall algorithm for computing $\mu^{\mathcal{G}^\uparrow_\tau}$ is given in Algorithm~\ref{algo:Cramer}.
The main idea is to initialize the labels on base nodes via the recipe given in Theorem~\ref{thm:bottleneck_cycle}, before running Algorithm~\ref{algo:Bellman+Ford}.
The labels on $V\setminus B(G_\tau)$ are initialized to~$\top$.
The auxiliary graph $D$ serves as a condensed representation of the `best' paths between base nodes.
The arc costs are chosen such that minimum bottleneck cycles in $D$ give a good estimate on the width of base nodes.

\begin{algorithm}[htbp]
  \caption{Label-Correcting: $(G_\tau,\pi)$ 1-player game for Even, $\mathcal{C}$ cover of $\mathcal{T}$ for some universal tree $T$, $\mu:V\rightarrow \bar{L}(T)$ node labeling such that $G_\tau$ does not contain loose arcs.}
  \label{algo:Cramer}
  \begin{algorithmic}[1]
    \Procedure{LabelCorrecting}{($G_\tau,\pi),\mathcal{C},\mu$}
    \State $\nu(v) \gets \top$ for all $v\in V$
    \State Construct auxiliary digraph $D$
    \ForAll{components $H$ in $D$}
    \For{$k=0$ \textbf{to} $|\mathcal{C}_{\pi(H)/2}|-1$}
      \State Assign arc costs $c^k$ to $H$ where $\underline{c}^k\leq c^k\leq \overline{c}^k$
      \ForAll{$w\in V(H)$} 
      \State $i^k \gets $ minimum $c^k$-cost of a cycle containing $w$ in $H$
      \If{$i^k < \infty$}
        \State $\nu(w) \gets \min(\nu(w), \text{\sc Raise}({\mu(w),i^k,\frac{\pi(H)}{2},k}))$
      \EndIf
    \EndFor
    \EndFor   
    \EndFor 
    \State $\nu \gets $ \Call{BellmanFord}{$(G_\tau,\pi),\nu$}
    \State \Return $\nu$
    \EndProcedure
  \end{algorithmic}
\end{algorithm}

In the next paragraph, we elaborate on how the arc costs $c^k$ are computed.

\paragraph{Computing arc costs}

Fix a base node $w\in B(G_\tau)$ and let $j=\pi(w)/2$.
We show how to compute $c^k(vw)$ for all incoming arcs $vw\in \delta^-_{D}(w)$ by running {\sc BellmanFord} on $J_w$.
For every $0\leq i < |\mathcal{C}^k_j|$, let $\nu_i:V(J_w)\rightarrow \bar{L}(T^k_{i,j})$ be the node labeling given by $\nu_i(w) := \min L(T^k_{i,j})$ and $\nu_i(u) := \top$ for all $u\neq w$.
Let $\nu'_i$ be the node labeling obtained by running {\sc BellmanFord} on $J_w$ with $\nu_i$.
Then, the cost of $vw\in \delta^-_D(w)$ is set as
\[c^k(vw) := \min\set{i:\nu'_i(u)\neq \top \text{ for some }u\in N^+_{J_w}(v)}.\]

\begin{prop}
For every arc $vw\in \delta^-_D(w)$, the computed arc cost satifies
\[\underline{c}^k(vw)\leq c^k(vw)\leq \overline{c}^k(vw).\]
\end{prop}

\begin{proof}
For every $0\leq i < |\mathcal{C}^k_j|$, we have $\lambda^k_{i,w}\leq \nu'_i$ because $\lambda^k_{i,w} = \nu_i^{\mathcal{J}^\downarrow_w}$.
Therefore, $\underline{c}^k(vw)\leq c^k(vw)$.
Next, let $P$ be a $u$-$w$ path in $J_w$ for some $u\in N^+_{J_w}(v)$.
For every $i\geq \alpha^k(P)$, there exists a finite node labeling $\bar\nu_i:V(P)\to L(T^k_{i,j})$ which is feasible on $P$.
Without loss of generality, we may assume that $\bar\nu_i(w) = \min L(T^k_{i,j})$.
Since {\sc BellmanFord} ran for $|V(J_w)|-1$ iterations, an inductive argument from $w$ along $P$ shows that $\nu'_i\leq \bar \nu_i$.
In particular, $\nu'_i(u)\neq \top$ for all $i\geq \alpha^k(P)$.
It follows that $c^k(vw) \leq \overline{c}^k(vw)$.
\end{proof}

Since the procedure above involves running {\sc BellmanFord} on $J_w$ $|\mathcal{C}^k_j|$ times, its running time is $O(mn\gamma(T)|\mathcal{C}^k_j|)$.
If $|\mathcal{C}^k_j|$ is too big, then the procedure can be combined with binary search.
For every node $u\in V(J_w)$, we want to determine $\min\{i:\nu'_i(u)\neq \top\}$.
Fix a node $u\in V(J_w)$.
We first call {\sc BellmanFord} on $J_w$ with $i=\lfloor |\mathcal{C}^k_j|/2 \rfloor$, i.e., the middle tree in the chain $\mathcal{C}^k_j$.
By Lemma~\ref{lem:finite_comparable}, if $\nu'_i(u) = \top$, then $\nu'_{i'}(u) = \top$ for all $i'\leq i$.
Otherwise, $\nu'_{i'}(u) \neq \top$ for all $i'\geq i$ .
In both cases, we can discard half of the trees in $\mathcal{C}^k_j$.
Then, the process is repeated until the correct tree is found. 
By applying this to every node $u\in V(J_w)$, we obtain the following running time.

\begin{prop}\label{prop:computing-arc-costs}
For every base node $w\in B(G_\tau)$ and $0\leq k < |\mathcal{C}_j|$ where $j = \pi(w)/2$, the arc costs $\{c^k(vw):vw\in \delta^-_D(w)\}$ can be computed in $O(mn\gamma(T)\cdot\min\{|\mathcal{C}^k_j|,n\log|\mathcal{C}^k_j|\})$ time.
\end{prop}

\medskip
We are ready to prove a generic bound on the running time of Algorithm \ref{algo:Cramer} for an arbitrary universal tree $T$ with an arbitrary cover $\mathcal{C}$. 

\begin{thm}\label{thm:cramer}
  In $O(mn^2\gamma(T) \cdot \max_{j,k}|\mathcal{C}_j|\min\{|\mathcal{C}^k_j|,n\log|\mathcal{C}^k_j|\} + n\rho(T,\mathcal{C})\cdot \max_j|\mathcal{C}_j|)$ time,  Algorithm~\ref{algo:Cramer} returns $\mu^{\mathcal{G}^\uparrow_\tau}$.
\end{thm}

\begin{proof}
Correctness follows immediately from Theorem~\ref{thm:bottleneck_cycle} and Theorem~\ref{thm:bellmanFord}.
It is left to prove the running time. 
Identifying the base nodes takes $O(m\log d)$ time (Section \ref{sec:prelude}), while constructing the auxiliary digraph $D$ takes $O(mn)$ time.
Next, for every component $H$ in $D$ and $0\leq k < |\mathcal{C}_{\pi(H)/2}|$, computing the arc costs $c^k$ takes $O(|V(H)|mn\gamma(T)\cdot\min\{|\mathcal{C}^k_{\pi(H)/2}|,n\log|\mathcal{C}^k_{\pi(H)/2}|\})$ by Proposition~\ref{prop:computing-arc-costs}.
Then, computing the minimum $c^k$-cost cycles in $H$ takes $O(|E(H)|\log|V(H)|)$ time by Theorem~\ref{thm:min-bottleneck-cycles}, and applying {\sc Raise} takes $O(|V(H)|\rho(T,\mathcal{C}))$ time.
Finally, {\sc Bellman--Ford} runs in $O(mn\gamma(T))$ time.
\end{proof}

In the next two subsections, we apply Algorithm \ref{algo:Cramer} to the quasi-polynomial universal trees constructed in the literature \cite{conf/lics/JurdzinskiL17,conf/icalp/DaviaudJT20}.
It runs in time $O(mn^2\log n \log d)$ time for a succinct $(n,d/2)$-universal tree, and $O(mn^2\log^3n\log d)$ for a succinct Strahler $(n,d/2)$-universal tree.
The vertices in these trees are encoded using tuples of binary strings. 
For working with these tuples, we introduce the following notation.

\begin{defi}
Given a tuple $\xi = (\xi_{2h-1},\xi_{2h-3},\dots,\xi_1)$ of binary strings, denote $\zeta(\xi)$ as the number of leading zeroes in $\xi_{2h-1}$.
We also define $\zeta(\top) := -1$.
\end{defi}

For a pair of tuples $\xi,\xi'$ of binary strings, note that if $\xi\geq \xi'$, then $\zeta(\xi)\leq \zeta(\xi')$  by the lexicographic order on tuples.

\begin{defi}
Given a tuple $\xi = (\xi_{2h-1},\xi_{2h-3},\dots,\xi_1)$ of binary strings and an integer $\kappa\geq 0$, let $\xi^\kappa$ be the tuple obtained by deleting $\kappa$ leading zeroes from $\xi_{2h-1}$.
If $\kappa>\zeta(\xi)$, then $\xi^\kappa := \top$.
We also define $\top^\kappa := \top$.
\end{defi}

\begin{exa}
If $\xi = (001,\varepsilon,1)$, then $\zeta(\xi) = 2$, $\xi^1 = (01,\varepsilon,1)$, $\xi^2 = (1,\varepsilon,1)$ and $\xi^3 = \top$.
\end{exa}

\subsection{Application to Succinct Universal Trees}
\label{sec:succinct}

Let $T$ be a succinct $(n,h)$-universal tree.
Recall that every leaf $\xi\in L(T)$ corresponds to an $h$-tuple of binary strings where the number of bits fulfills $\size{\xi}\leq \floor{\log n}$.
First, we show that each $\mathcal{T}_j$ has a cover of size 1.
Equivalently, each $(\mathcal{T}_j,\sqsubseteq)$ is a chain.

\begin{lem}\label{lem:cover}
There exists a cover $\mathcal{C}$ of $\mathcal{T}$ such that $|\mathcal{C}_j| = 1$ for all $0\leq j\leq h$.
\end{lem}

\begin{proof}
Fix $0\leq j \leq h$ and pick two vertices $r_1,r_2\in V(T)$ at depth $h-j$. 
Let $T_1$ and $T_2$ be the subtrees of $T$ rooted at $r_1$ and $r_2$ respectively.
Every leaf $\xi_1\in L(T_1)$ and $\xi_2\in L(T_2)$ corresponds to a $j$-tuple of binary strings where $|\xi_1|\leq \floor{\log n} - \size{r_1}$ and $\size{\xi_2}\leq \floor{\log n} - \size{r_2}$.
Without loss of generality, assume that $\size{r_1}\geq\size{r_2}$.
Then, the identity map from $V(T_1)$ to $V(T_2)$ is an order-preserving and injective homomorphism.
Hence, $T_1\sqsubseteq T_2$.
\end{proof}

Since each subcover $\mathcal{C}_j$ of $\mathcal{C}$ consists of a single chain, we write $\mathcal{C}_j = \mathcal{T}_j$ and omit the superscript $k$.
The subtrees in $\mathcal{T}_j$ are $T_{0,j} \sqsubset T_{1,j} \sqsubset \cdots \sqsubset T_{\floor{\log n},j}$.
Observe that every leaf $\xi\in L(T_{i,j})$ corresponds to a $j$-tuple of binary strings where $|\xi|\leq i$.

Next, we give an efficient implementation of the {\sc Raise} subroutine.

\begin{lem}\label{lem:raise}
For a succinct $(n,d/2)$-universal tree $T$ with cover $\mathcal{C}=(\mathcal{T}_0,\mathcal{T}_1,\dots,\mathcal{T}_{d/2})$, the {\sc Raise}$(\xi,i,j,k)$ subroutine runs in $O(\log n \log d)$ time.
\end{lem}

\begin{proof}
We may assume that $\xi$ is the smallest leaf in the subtree rooted at $\xi|_{2j}$.
Otherwise, we can set it to the smallest leaf of the next subtree rooted at that depth using {\sc Tighten} (we return $\top$ if this next subtree does not exist).
Recall that the {\sc Tighten} subroutine for $T$ also runs in $O(\log n\log d)$ time.
It follows that $\xi_{2j-1}\in \set{0\cdots0,\varepsilon}$ and $\xi_q = \varepsilon$ for all odd $q<2j-1$.
If $\size{\xi_{2j-1}}\geq i$, then we simply return~$\xi$.
Otherwise, let $p\in [d]$ be the smallest even integer such that $\xi|_p$ has a child bigger than $\xi|_{p-1}$ with at most $\floor{\log n}-i$ bits.
If $p$ does not exist, then we return $\top$.
Otherwise, $p>2j$. 
Let $r = \floor{\log n} - i - \size{\xi|_{p-1}}$.
There are two cases. 

\smallskip
\emph{Case 1: $r > 0$.} Return
\[(\xi_{d-1},\dots,\xi_{p-1}1\underbrace{0\cdots0}_{r-1},\varepsilon,\dots,\varepsilon,\underbrace{0\cdots0}_{i},\varepsilon,\dots,\varepsilon)\]
where the string of $i$ zeroes is at index $2j-1$.

\smallskip
\emph{Case 2: $r\leq 0$.}
Denote $\xi_{p-1} = b_1b_2\cdots b_\ell$ where $b_q\in \set{0,1}$ for all $q\in [\ell]$.
Note that $\ell\geq 1$ by our choice of $p$.
Furthermore, there exists a largest $t\in [\ell]$ such that $b_t = 0$ and $r' := r+\ell-(t-1)\geq 0$. 
Then,
\begin{itemize}
  \item if $p = 2j + 2$, return
    \[
    (\xi_{d-1},\dots,\xi_{p+1},b_1\cdots b_{t-1},\underbrace{0\cdots0}_{i+r'},\varepsilon,\dots,\varepsilon)\,;
    \]
  \item if $p > 2j + 2$, return
  \[(\xi_{d-1},\dots,\xi_{p+1},b_1\cdots b_{t-1},\underbrace{0\cdots0}_{r'},\varepsilon,\dots,\varepsilon,\underbrace{0\cdots0}_{i},\varepsilon,\dots,\varepsilon)\]
  where the string of $i$ zeroes is at index $2j-1$. \qedhere
\end{itemize}
\end{proof}

Due to the self-similar structure of succinct universal trees, the running time of Algorithm \ref{algo:Cramer} given in Theorem \ref{thm:cramer} can be improved.
The following lemma yields a faster method for computing arc costs for the auxiliary digraph $D_\tau$.
The key observation is that for any pair of trees in a chain $\mathcal{T}_j$, the smaller tree can be obtained from the larger tree by deleting vertices in decreasing lexicographic order.
For example, a succinct (3,2)-universal tree can be obtained from a succinct (7,2)-universal tree by deleting vertices whose first component does not contain a leading zero (compare Figures \ref{fig:succinct_tree} and~\ref{fig:pruning}).  

\begin{figure}[ht]
\def\x{0.8}
\def\y{2}
\centering
\begin{tikzpicture}
\centering
\begin{scope}[every node/.style={circle,draw,,inner sep=2pt}]
    \node (p) at (0,0) {};
    \node (p1) at (-8*\x,-1*\y) {};  
    \node (p2) at (-6*\x,-1*\y) {};  
    \node (p3) at (-4*\x,-1*\y) {};   
    \node (p11) at (-8*\x,-2*\y) {};
    \node (p21) at (-7*\x,-2*\y) {};
    \node (p22) at (-6*\x,-2*\y) {};
    \node (p23) at (-5*\x,-2*\y) {};
    \node (p31) at (-4*\x,-2*\y) {};   
    \node (p4) at (0,-1*\y) {};
    \node (p5) at (4*\x,-1*\y) {}; 
    \node (p6) at (6*\x,-1*\y) {}; 
    \node (p7) at (8*\x,-1*\y) {}; 
\end{scope}
\begin{scope}[every node/.style={circle,draw, inner sep=2pt}] 
    \node (p41) at (0-3*\x,-2*\y) {};   
    \node (p42) at (0-2*\x,-2*\y) {};
    \node (p43) at (0-1*\x,-2*\y) {};
    \node (p44) at (0,-2*\y) {};
    \node (p45) at (0+1*\x,-2*\y) {};   
    \node (p46) at (0+2*\x,-2*\y) {};
    \node (p47) at (0+3*\x,-2*\y) {};
    \node (p51) at (4*\x,-2*\y) {};
    \node (p61) at (5*\x,-2*\y) {};
    \node (p62) at (6*\x,-2*\y) {};
    \node (p63) at (7*\x,-2*\y) {};
    \node (p71) at (8*\x,-2*\y) {};
\end{scope}

\begin{scope}[every node/.style={fill=white,circle,font=\footnotesize,inner sep=2pt},
              every edge/.style={draw}]
    \path [-] (p) edge node {00} (p1);
    \path [-] (p) edge node {0} (p2);
    \path [-] (p) edge node {01} (p3);
    \path [-] (p) edge node {$\varepsilon$} (p4);
    \path [-] (p) edge node {10} (p5);
    \path [-] (p) edge node {1} (p6);
    \path [-] (p) edge node {11} (p7);
    \path [-] (p1) edge node {$\varepsilon$} (p11);
    \path [-] (p2) edge node {0} (p21);
    \path [-] (p2) edge node {$\varepsilon$} (p22);
    \path [-] (p2) edge node {1} (p23);
    \path [-] (p3) edge node {$\varepsilon$} (p31);
    \path [-] (p4) edge node {00} (p41);
    \path [-] (p4) edge node {0} (p42);
    \path [-] (p4) edge node {01} (p43);
    \path [-] (p4) edge node {$\varepsilon$} (p44);
    \path [-] (p4) edge node {10} (p45);
    \path [-] (p4) edge node {1} (p46);
    \path [-] (p4) edge node {11} (p47);
    \path [-] (p5) edge node {$\varepsilon$} (p51);
    \path [-] (p6) edge node {0} (p61);
    \path [-] (p6) edge node {$\varepsilon$} (p62);
    \path [-] (p6) edge node {1} (p63);
    \path [-] (p7) edge node {$\varepsilon$} (p71);
\end{scope}
\end{tikzpicture}
\caption{The succinct $(7,2)$-universal tree}
\label{fig:pruning}
\end{figure}
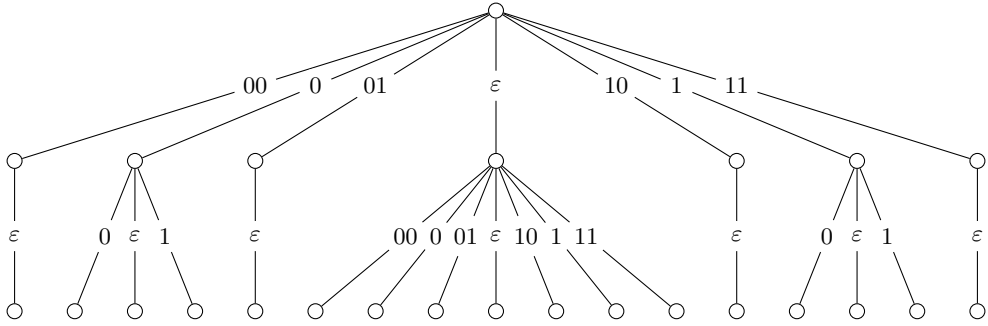

\begin{lem}\label{lem:left_embed}
Given integers $0\leq i_1\leq i_2$ and $j\geq 0$, let $\nu_1:V\rightarrow \bar{L}(T_{i_1,j})$ and $\nu_2:V\rightarrow \bar{L}(T_{i_2,j})$ be node labelings such that $\nu_1(u) = \nu_2(u)^{i_2-i_1}$ for all $u\in V$.
For any arc $vw\in E_\tau$ where $\pi(v)<2j$, we have $\drop(\nu_1,vw) = \drop(\nu_2,vw)^{i_2-i_1}$.
\end{lem}

\begin{proof}
Let $\xi_1 = \drop(\nu_1,vw)$ and $\xi_2 = \drop(\nu_2,vw)$.
First, assume that $vw$ is violated with respect to $\nu_1$.
Then, $\nu_1(v)\neq \top$, which implies that $\nu_2(v)\neq \top$.
In particular, $\zeta(\nu_2(v))\geq i_2 - i_1$.
We claim that $vw$ is also violated with respect to $\nu_2$.
This is clear if $\nu_1(w)\neq \top$.
If $\nu_1(w) = \top$, then $\zeta(\nu_2(w))<i_2-i_1$. 
As $\nu_2(v)|_{2j-1}<\nu_2(w)|_{2j-1}$ and $\pi(v)<2j$, the arc $vw$ is indeed violated with respect to $\nu_2$.
Hence, $\xi_1 = \nu_1(v)$ and $\xi_2 = \nu_2(v)$.

Next, assume that $vw$ is not violated with respect to $\nu_1$.
If $\nu_1(w)\neq \top$, then $vw$ is also not violated with respect to $\nu_2$.
It is easy to verify that $\xi_1 = \xi_2^{i_2-i_1}$.
On the other hand, if $\nu_1(w) = \top$, then $\nu_1(v) = \top$.
So, we have $\zeta(\nu_2(v))<i_2-i_1$ and $\zeta(\nu_2(w))<i_2-i_1$.
Since $\pi(v)<2j$, we also have $\zeta(\xi_2)<i_2-i_1$.
Thus, $\xi_1 = \top = \xi_2^{i_2-i_1}$ as required.
\end{proof}

\paragraph{Computing arc costs}

Fix a component $H$ in the auxiliary graph $D$ and let $j = \pi(H)/2$.
Fix a node $w\in V(H)$.
We first recall the generic method for computing the arc costs of $\delta^-_D(w)$.
For each $0\leq i\leq \floor{\log n}$, let $\nu_i:V(J_w)\rightarrow \bar{L}(T_{i,j})$ be the node labeling defined by $\nu_i(w):= \min L(T_{i,j})$ and $\nu_i(u) := \top$ for all $u\neq w$.
Let $\nu'_i$ be the node labeling obtained by running {\sc BellmanFord} on $J_w$ with input node labeling $\nu_i$.
Then, the cost of every arc $vw\in \delta^-_D(w)$ is set to 
\[c(vw) := \min\set{i:\nu'_i(u)\neq \top \text{ for some }u\in N^+_{J_w}(v)}.\]

The above method involves running {\sc BellmanFord} $\floor{\log n}+1$ times.
We outline a faster method which only runs {\sc BellmanFord} once.
Observe that for any $0\leq i\leq \floor{\log n}$, we have $\nu_i(u) = \nu_{\floor{\log n}}(u)^{\floor{\log n}-i}$ for all $u\in V$.
Since $\delta^-_{J_w}(v) = \emptyset$ for all $v\in \Pi(J_w)\setminus \{w\}$ by construction, and $\nu_i(w) = \min L(T_{i,j})$, Lemma \ref{lem:left_embed} yields
\[\nu'_i(u) = \nu'_{\floor{\log n}}(u)^{\floor{\log n}-i}\] 
for all $u\in V(J_w)\setminus (\Pi(J_w)\setminus\{w\})$.
In other words, from $\nu'_{\floor{\log n}}$ we can infer $\nu'_i(u)$ for all $u\in V(J_w)\setminus (\Pi(J_w)\setminus \{w\})$ and $0\leq i\leq \floor{\log n}$. So, we only need to run {\sc BellmanFord} on $J_w$ once to obtain $\nu'_{\floor{\log n}}$.
Note that $\nu'_{\floor{\log n}}$ is finite because $T_{\floor{\log n},j}$ is an $(n,j)$-universal tree and every node in $J_w$ can reach $w$.
Then, the cost of each arc $vw\in \delta^-_{D}(w)$ is given by 
\[c(vw) = \floor{\log n} - \max_{u\in N^+_{J_w}(v)}\zeta(\nu(u)).\]

The above argument shows that we can compute $c$ for the component $H$ in time $O(|V(H)|mn\gamma(T))$, saving a factor of $\log n$.
We are ready to prove the running time of Algorithm \ref{algo:Cramer} for succinct universal trees.

\begin{thm}
For a succinct $(n,d/2)$-universal tree $T$ with cover $\mathcal{C}=(\mathcal{T}_0,\mathcal{T}_1,\dots, \mathcal{T}_{d/2})$, Algorithm \ref{algo:Cramer} runs in  $O(mn^2\log n \log d)$ time.
\end{thm}

\begin{proof}
By Lemma \ref{lem:cover}, we have $|\mathcal{C}_j|=1$ for all $0\leq j\leq h$.
Computing arc costs for the auxiliary digraph $D_\tau$ takes $O(mn^2\gamma(T))$ time as discussed above.
Hence, the running time of Algorithm \ref{algo:Cramer} becomes $O(n\rho(T,\mathcal{C}) + mn^2\gamma(T))$.
The result then follows from $\gamma(T) = O(\log n \log d) = \rho(T,\mathcal{C})$, where the latter equality is due to Lemma \ref{lem:raise}.
\end{proof}

\subsection{Application to Succinct Strahler Universal Trees}
\label{sec:strahler}

In this subsection, we apply Algorithm \ref{algo:Cramer} to succinct Strahler universal trees.
Let us start by introducing the necessary definitions.
The \emph{Strahler number} of a rooted tree $T$ is the largest height of a perfect binary tree that is a minor of $T$.
For example, a perfect $(\ell,h)$-universal tree has Strahler number 0 if $\ell = 1$, and $h$ otherwise. 

\begin{defi}
A \emph{$g$-Strahler $(\ell,h)$-universal tree} is an ordered tree $T'$ of height $h$ such that $T\sqsubseteq T'$ for every ordered tree $T$ of height $h$, Strahler number at most $g$, and with at most $\ell$ leaves, all at depth exactly $h$. 
\end{defi}

In the definition above, we may assume that $g\leq \min(h,\floor{\log \ell})$.
Daviaud et al.~\cite{conf/icalp/DaviaudJT20} constructed a $g$-Strahler $(\ell,h)$-universal tree with $\ell^{O(1)}(h/g)^g$ leaves.
Note that this is quasipolynomial in $\ell$ and $h$ because $g\leq \min(h,\floor{\log \ell})$.
In this paper, we refer to it as a \emph{succinct $g$-Strahler $(\ell,h)$-universal tree}.
Every leaf $\xi = (\xi_{2h-1},\xi_{2h-3},\dots,\xi_1)$ in this tree corresponds to an $h$-tuple of binary strings which satisfies the following three properties:
\begin{enumerate}
  \item\label{item:nonempty-bit-strings} There are $g$ nonempty bit strings, i.e.~$\size{\set{i:\xi_i\neq \varepsilon}} = g$;
  \item\label{item:total-number-bits} The total number of bits \size{\xi} is at most $g + \floor{\log \ell}$;
  \item For each odd $i\in [2h]$,
  \begin{enumerate}
    \item\label{item:exhaust-nonleading-bits} If there are $f<g$ nonempty bit strings in $\xi|_{i+1}$ and $\size{\xi|_{i+1}}=f+\floor{\log \ell}$, then $\xi_i = 0$. 
    \item\label{item:start-with-zero} If $\xi_j\neq \varepsilon$ for all odd $j\in [i]$, then $\xi_j$ starts with 0 for all odd $j\in [i]$.
  \end{enumerate}
\end{enumerate}
This is the construction of $\mathcal{B}^k_{t,h}$ in~\cite[Definition~19]{conf/icalp/DaviaudJT20}.
In each string $\xi_i$, the first bit is called the \emph{leading bit}, while the remaining bits are the \emph{non-leading bits}.
Properties \ref{item:nonempty-bit-strings} and \ref{item:total-number-bits} imply that $\xi$ contains exactly $g$ leading bits and at most $\floor{\log \ell}$ non-leading bits (see Figure~\ref{fig:strahler} for an example).
Observe that if $g=h$, then the tree is identical to a succinct $(\ell,h)$-universal tree. 
Indeed, one can arrive at the encoding of Jurdzi\'{n}ski and Lazi\'{c} \cite{conf/lics/JurdzinskiL17} by removing the leading zero in every string.

\begin{figure}[ht]
\def\x{0.9}
\def\y{2}
\def\z{0.6}
\centering
\begin{tikzpicture}
\centering
\begin{scope}[every node/.style={circle,draw,,inner sep=2pt}]
    \node (p) at (0,0) {};
    \node (p1) at (-2*\z-7*\x,-1*\y) {};  
    \node (p2) at (-2*\z-6*\x,-1*\y) {};  
    \node (p3) at (-2*\z-5*\x,-1*\y) {};
    \node (p4) at (-2*\z-4*\x,-1*\y) {};
    \node (p5) at (-2*\z-3*\x,-1*\y) {}; 
    \node (p6) at (-2*\z-2*\x,-1*\y) {}; 
    \node (p7) at (-2*\z-1*\x,-1*\y) {}; 
    \node (p8) at (0,-1*\y) {};  
    \node (p9) at (2*\z+1*\x,-1*\y) {};  
    \node (p10) at (2*\z+2*\x,-1*\y) {};
    \node (p11) at (2*\z+3*\x,-1*\y) {};
    \node (p12) at (2*\z+4*\x,-1*\y) {}; 
    \node (p13) at (2*\z+5*\x,-1*\y) {}; 
    \node (p14) at (2*\z+6*\x,-1*\y) {}; 
    \node (p15) at (2*\z+7*\x,-1*\y) {}; 
\end{scope}
\begin{scope}[every node/.style={circle,draw, inner sep=2pt}] 
    \node (q1) at (-2*\z-7*\x,-2*\y) {};
    \node (q2) at (-2*\z-6*\x,-2*\y) {};
    \node (q3) at (-2*\z-5*\x,-2*\y) {};   
    \node (q4) at (-2*\z-4*\x,-2*\y) {};   
    \node (q5) at (-2*\z-3*\x,-2*\y) {};
    \node (q6) at (-2*\z-2*\x,-2*\y) {};
    \node (q7) at (-2*\z-1*\x,-2*\y) {};
    \node (q81) at (-3*\z,-2*\y) {};   
    \node (q82) at (-2*\z,-2*\y) {};
    \node (q83) at (-1*\z,-2*\y) {};
    \node (q84) at (0,-2*\y) {};
    \node (q85) at (1*\z,-2*\y) {};   
    \node (q86) at (2*\z,-2*\y) {};
    \node (q87) at (3*\z,-2*\y) {};  
    \node (q9) at (2*\z+1*\x,-2*\y) {};
    \node (q10) at (2*\z+2*\x,-2*\y) {};
    \node (q11) at (2*\z+3*\x,-2*\y) {};   
    \node (q12) at (2*\z+4*\x,-2*\y) {};   
    \node (q13) at (2*\z+5*\x,-2*\y) {};
    \node (q14) at (2*\z+6*\x,-2*\y) {};
    \node (q15) at (2*\z+7*\x,-2*\y) {};
\end{scope}

\begin{scope}[every node/.style={fill=white,rectangle,font=\scriptsize,inner sep=1pt},
              every edge/.style={draw}]
    \path [-] (p) edge node[pos=0.8] {000} (p1);
    \path [-] (p) edge node[pos=0.8] {00} (p2);
    \path [-] (p) edge node[pos=0.8] {001} (p3);
    \path [-] (p) edge node[pos=0.8] {0} (p4);
    \path [-] (p) edge node[pos=0.8] {010} (p5);
    \path [-] (p) edge node[pos=0.8] {01} (p6);
    \path [-] (p) edge node[pos=0.8] {011} (p7);
    \path [-] (p) edge node {$\varepsilon$} (p8);
    \path [-] (p) edge node[pos=0.8] {100} (p9);
    \path [-] (p) edge node[pos=0.8] {10} (p10);
    \path [-] (p) edge node[pos=0.8] {101} (p11);
    \path [-] (p) edge node[pos=0.8] {1} (p12);
    \path [-] (p) edge node[pos=0.8] {110} (p13);
    \path [-] (p) edge node[pos=0.8] {11} (p14);
    \path [-] (p) edge node[pos=0.8] {111} (p15);
    \path [-] (p1) edge node {$\varepsilon$} (q1);
    \path [-] (p2) edge node {$\varepsilon$} (q2);
    \path [-] (p3) edge node {$\varepsilon$} (q3);
    \path [-] (p4) edge node {$\varepsilon$} (q4);
    \path [-] (p5) edge node {$\varepsilon$} (q5);
    \path [-] (p6) edge node {$\varepsilon$} (q6);
    \path [-] (p7) edge node {$\varepsilon$} (q7);
    \path [-] (p8) edge node[pos=0.8] {000} (q81);
    \path [-] (p8) edge node[pos=0.8] {00} (q82);
    \path [-] (p8) edge node[pos=0.8] {001} (q83);
    \path [-] (p8) edge node[pos=0.8] {0} (q84);
    \path [-] (p8) edge node[pos=0.8] {010} (q85);
    \path [-] (p8) edge node[pos=0.8] {01} (q86);
    \path [-] (p8) edge node[pos=0.8] {011} (q87);
    \path [-] (p9) edge node {$\varepsilon$} (q9);
    \path [-] (p10) edge node {$\varepsilon$} (q10);
    \path [-] (p11) edge node {$\varepsilon$} (q11);
    \path [-] (p12) edge node {$\varepsilon$} (q12);
    \path [-] (p13) edge node {$\varepsilon$} (q13);
    \path [-] (p14) edge node {$\varepsilon$} (q14);
    \path [-] (p15) edge node {$\varepsilon$} (q15);
    
\end{scope}
\end{tikzpicture}
\caption{The succinct 1-Strahler $(7,2)$-universal tree}
\label{fig:strahler}
\end{figure}
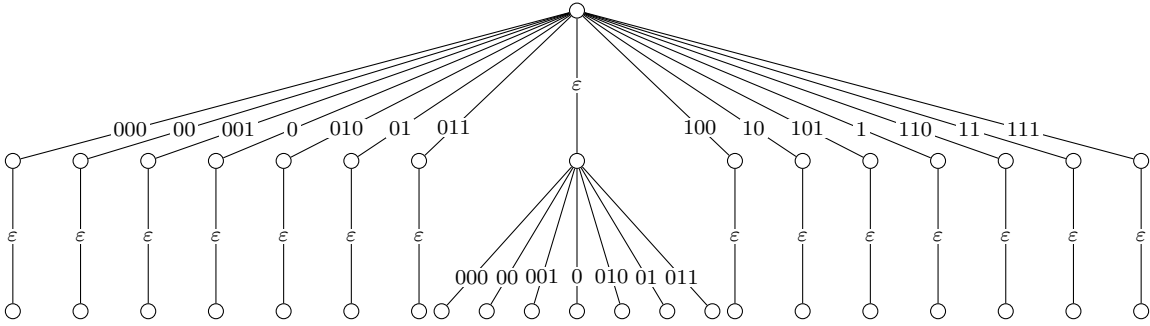

Let $T$ be a succinct $g$-Strahler $(\ell,h)$-universal tree.
Let $v\in V(T)$ be a vertex at depth $h-j$ for some $0\leq j\leq h$.
If $v$ has $k$ nonempty strings and $i$ non-leading bits, then the subtree rooted at $v$ is a succinct $(g-k)$-Strahler $(2^{\floor{\log \ell}-i},j)$-universal tree.
Note that if $i=\floor{\log \ell}$, then the subtree has only one leaf.
This fact is actually independent of $k$.
Indeed, varying $k$ only yields different encodings of the same tree.

The crucial fact for obtaining a small cover is that the subtrees of $T$ with fixed height $j$ and fixed Strahler number $k$ form a chain. 
This leads to the following statement.

\begin{lem}\label{lem:cover_strahler}
There exists a cover $\mathcal{C}$ of $\mathcal{T}$ such that $\size{\mathcal{C}_j}\leq g$ for all $0\leq j \leq h$.
\end{lem}

\begin{proof}
Fix a height $0\leq j\leq h$.
For each $0\leq k\leq g$, let $\mathcal{C}^k_j$ be the set of succinct $k$-Strahler universal trees in $\mathcal{T}_j$.
Then, $\cup_k \mathcal{C}^k_j = \mathcal{T}_j$.
It is left to show that $(\mathcal{C}^k_j,\sqsubseteq)$ is a chain for all $k$.
Fix a $k$ and pick two vertices $r_1,r_2\in V(T)$ at depth $h-j$ such that they each have $g-k$ nonempty bit strings.
Then, they each have $g-k$ leading bits. So, the number of non-leading bits in $r_1$ and $r_2$ are $|r_1| - (g-k)$ and $|r_2| - (g-k)$ respectively.
Let $T_1$ and $T_2$ be the subtrees of $T$ rooted at $r_1$ and $r_2$ respectively.
Observe that $T_1$ is a succinct $k$-Strahler $(2^{\floor{\log n}-\size{r_1}+g-k},j)$-universal tree.
Similarly, $T_2$ is a succinct $k$-Strahler $(2^{\floor{\log n}-\size{r_2}+g-k},j)$-universal tree.
Without loss of generality, assume that $\size{r_1}\geq \size{r_2}$.
Then, the identity map from $V(T_1)$ to $V(T_2)$ is an order-preserving and injective homomorphism.
Hence, $T_1\sqsubseteq T_2$.
\end{proof}

Let $\mathcal{C}$ be the cover given in the proof of Lemma \ref{lem:cover_strahler}, i.e.~$\mathcal{C}_j = (\mathcal{C}^0_j, \mathcal{C}^1_j, \dots, \mathcal{C}^{g}_j)$ for all $0\leq j\leq h$.
\begin{obs}\label{obs:strahler_cover}
We make a few observations about the cover $\mathcal{C}$:
\begin{enumerate}
  \item We have $\mathcal{C}^k_j \neq \emptyset$ if and only if $\max(0,g-(h-j))\leq k\leq \min(j,g)$.
  \item For any $j$ where $\mathcal{C}^0_j\neq \emptyset$, we have $\mathcal{C}^0_j=(T^0_{0,j})$ where $T^0_{0,j}$ has only one leaf.
  \item For any $j,k$ where $\mathcal{C}^k_j\neq \emptyset$, its smallest tree $T^k_{0,j}$ has only one leaf. 
\end{enumerate}
\end{obs}
The next lemma shows that the {\sc Raise} subroutine can be implemented efficiently using this cover when its second argument is positive. This will be sufficient for running Algorithm \ref{algo:Cramer}, as Theorem~\ref{thm:strahler} shows.

\begin{lem}\label{lem:raise_strahler}
For a succinct $g$-Strahler $(n,d/2)$-universal tree $T$ with cover $\mathcal{C}$, if $i>0$, then the {\sc Raise}$(\xi,i,j,k)$ subroutine runs in $O(\log n \log d)$ time.
\end{lem}

\begin{proof}
Note that $k>0$ by Observation~\ref{obs:strahler_cover}.
We may assume that $\xi$ is the smallest leaf in the subtree of $T$ rooted at $\xi|_{2j}$.
Otherwise, we can set it to the smallest leaf of the next subtree rooted at that depth using {\sc Tighten} (we return $\top$ if this next subtree does not exist).
It is known that the {\sc Tighten} subroutine for $T$ runs in $O(\log n\log d)$ time~\cite[Lemma 22]{conf/icalp/DaviaudJT20}.
This gives us $\xi_{2j-1}\in \set{0\cdots 0,\varepsilon}$ and $\xi_q\in \set{0,\varepsilon}$ for all $q<2j-1$.
If there are at least $i$ non-leading bits in $\xi_{2j-1}$ and exactly $k$ non-empty strings among $\xi_{2j-1},\dots,\xi_1$, then we simply return $\xi$.
Otherwise, let $p>2j$ be the smallest even integer such that $\xi|_p$ has a child bigger than $\xi|_{p-1}$ with at most $\floor{\log n}-i$ non-leading bits and exactly $x$ nonempty strings for some $g-k- (p/2 - j-1)\leq x\leq g-k$.
If $p$ does not exist, then we return $\top$.

Our goal is to increase $\xi$ minimally such that the last $j$ components form a tuple from $L(T^k_{i',j})$ for some $i'\geq i$.
Let $z$ be the number of non-empty strings in $\xi|_{p-1}$, and let $y$ be the number of non-leading bits in $\xi|_{p-1}$. 
We set $s = g-k-z$, representing the discrepancy on the number of nonempty bit strings (Property~\ref{item:nonempty-bit-strings}).
Note that $-1\leq s\leq p/2-j$ due to our choice of $p$.
We also set $r=\floor{\log n}-i-y$, representing the discrepancy on the number of non-leading bits (Property~\ref{item:total-number-bits}).
We split the remaining analysis into cases based on the emptiness of $\xi_{p-1}$ and the signs of $r,s$.

\smallskip
\emph{Case 1: $\xi_{p-1} = \varepsilon$.} 
Note that $y\leq \floor{\log n}-i$ and $z<g-k$.
So, $r\geq 0$ and $s>0$.
Return
\[(\xi_{d-1},\dots,\xi_{p+1},1\underbrace{0\cdots0}_{r},\underbrace{0,\dots,0}_{s-1},\varepsilon,\dots,\varepsilon,\underbrace{0\cdots0}_{i+1},\underbrace{0,\dots,0}_{k-1},\varepsilon,\dots,\varepsilon)\]
where the string of $i+1$ zeroes is at index $2j-1$.
Property~\ref{item:start-with-zero} holds because there is at least one empty string among the last $p/2$ components.
As all the other properties are also fulfilled by construction, this is a valid tuple encoding a leaf.

\smallskip
\emph{Case 2: $\xi_{p-1}\neq \varepsilon$ and $s<0$.}
Note that $s=-1$ and $\xi_{p-1}$ has a leading zero due to our choice of $p$.
Let $t$ be the number of non-leading bits in $\xi_{p-1}$.
Then, $y-t\leq \floor{\log n}-i$, which implies that $r+t\geq 0$.
Return
\[(\xi_{d-1},\dots,\xi_{p+1},\varepsilon,\varepsilon,\dots,\varepsilon,\underbrace{0\cdots0}_{i+1+r+t},\underbrace{0,\dots,0}_{k-1},\varepsilon,\dots,\varepsilon)\]
where the string of $i+1+r+t$ zeroes is at index $2j-1$.

\smallskip
\emph{Case 3: $\xi_{p-1} \neq \varepsilon$, $r>0$ and $s\geq 0$.}
Return 
\[(\xi_{d-1},\dots,\xi_{p-1}1\underbrace{0\cdots0}_{r-1},\underbrace{0,\dots,0}_s,\varepsilon,\dots,\varepsilon,\underbrace{0\cdots0}_{i+1},\underbrace{0,\dots,0}_{k-1},\varepsilon,\dots,\varepsilon)\]
where the string of $i+1$ zeroes is at index $2j-1$.
Property~\ref{item:start-with-zero} is satisfied because the last $p/2$ components are nonempty if and only if the last $p/2$ components of $\xi$ are nonempty.

\smallskip
\emph{Case 4: $\xi_{p-1} \neq \varepsilon$, $r\leq0$ and $s\geq 0$.}
Note that $s\leq p/2-j-1$ in this case. 
Denote $\xi_{p-1} = b_1b_2\cdots b_\ell$ where $b_q\in \set{0,1}$ for all $q\in [\ell]$.
Due to our choice of $p$, there exists a largest $t\in [\ell]$ such that $b_t = 0$ and $r' := r+\ell-\max(t-1,1)\geq 0$.
The following subcases remain.
\begin{itemize}
  \item If $s<p/2-j-1$ and $t=1$, return
  \[(\xi_{d-1},\dots,\xi_{p+1},\varepsilon,\underbrace{0\cdots0}_{r'+1},\underbrace{0,\dots,0}_{s},\varepsilon,\dots,\varepsilon,\underbrace{0\cdots0}_{i+1},\underbrace{0,\dots,0}_{k-1},\varepsilon,\dots,\varepsilon)\]
  where the string of $i+1$ zeroes is at index $2j-1$.
  \item If $s=p/2-j-1$ and $t=1$, return
  \[(\xi_{d-1},\dots,\xi_{p+1},1\underbrace{0\cdots0}_{r'},\underbrace{0,\dots,0}_{s},\underbrace{0\cdots0}_{i+1},\underbrace{0,\dots,0}_{k-1},\varepsilon,\dots,\varepsilon)\]
  where the string of $i+1$ zeroes is at index $2j-1$. Property~\ref{item:start-with-zero} holds because there exists an empty string in the last $p/2$ components of $\xi$ by our choice of $p$.
  \item If $s = 0$ and $t>1$, return
  \[(\xi_{d-1},\dots,\xi_{p+1},b_1\cdots b_{t-1},\varepsilon,\dots,\varepsilon,\underbrace{0\cdots0}_{i+1+r'},\underbrace{0,\dots,0}_{k-1},\varepsilon,\dots,\varepsilon)\]
  where the string of $i+1+r'$ zeroes is at index $2j-1$.
  \item If $s>0$ and $t>1$, return
  \[(\xi_{d-1},\dots,\xi_{p+1},b_1\cdots b_{t-1},\underbrace{0\cdots0}_{r'+1},\underbrace{0,\dots,0}_{s-1},\varepsilon,\dots,\varepsilon,\underbrace{0\cdots0}_{i+1},\underbrace{0,\dots,0}_{k-1},\varepsilon,\dots,\varepsilon)\]
  where the string of $i+1$ zeroes is at index $2j-1$. \qedhere
\end{itemize}
\end{proof}

We are ready to prove the running time of Algorithm \ref{algo:Cramer} for succinct Strahler universal trees. 
Recall that for a $g$-Strahler $(n,d/2)$-universal tree, we may assume that $g\leq \log n$.

\begin{thm}\label{thm:strahler}
For a succinct $g$-Strahler $(n,d/2)$-universal tree $T$ with cover $\mathcal{C}$, Algorithm~\ref{algo:Cramer} runs in $O(mn^2g\log^2n\log d)$ time.
\end{thm}
\begin{proof}
By Lemma \ref{lem:cover_strahler}, we have $|\mathcal{C}_j|\leq g$ for all $0\leq j \leq h$.
Furthermore, $|\mathcal{C}^k_j|\leq \log n$ for all $0\leq j \leq h$ and $0\leq k < |\mathcal{C}_j|$.
It is known that \textsc{Tighten}$(\mu,vw)$ takes $O(\log n \log d)$ time~\cite[Lemma 22]{conf/icalp/DaviaudJT20}, whereas Lemma~\ref{lem:raise_strahler} shows that \textsc{Raise}$(\xi,i,j,k)$ takes $O(\log n \log d)$ time when $i>0$.

In Algorithm \ref{algo:Cramer}, if $\text{\sc Raise}({\mu(w),i^k,\pi(H)/2,k})$ is called with $i^k=0$, then there exists a cycle $C$ containing $w$ in $H$ such that $\underline{c}^k(e) \leq c^k(e) = 0$ for all $e\in E(C)$.
Denote $E(C) = \{w_0w_1, w_1w_2, \dots, w_{r-1}w_r\}$ where $w_0 = w_r$.
For each $s\in [r]$, the greatest simultaneous fixed point $\lambda^k_{0,w_s}:V(J_{w_s})\to \bar{L}(T^k_{0,j})$ assigns a finite label to an outneighbour of $w_{s-1}$.
Since $T^k_{0,j}$ has only one leaf by Observation~\ref{obs:strahler_cover}, there exists a $w_{s-1}$-$w_s$ path in $J_{w_s}$ which consists of only even priorities.
Hence $\overline{c}^k(w_{s-1}w_s) = 0$.
The same path also certifies that $\overline{c}^{k'}(w_{s-1}w_s) = \underline{c}^{k'}(w_{s-1}w_s) = 0$ for all $k'\neq k$.
It follows that $c^{k'}(e) = 0$ for all $e\in E(C)$, so the minimum $c^{k'}$-cost of a cycle containing $w$ in $H$ is $0$ for all $k'$.
Thus, the algorithm will always call $\text{\sc Raise}({\mu(w),i^{k'},\pi(H)/2,k'})$ with $i^{k'} = 0$ for all $k'$.
As $T^{k'}_{0,j}$ has only one leaf for all $k'$ by Observation~\ref{obs:strahler_cover}, the minimum of $\text{\sc Raise}({\mu(w),0,\pi(H)/2,k'})$ over all $k'$ is $\mu(w)$ if $\mu(w)$ is the smallest leaf of the subtree rooted at $\mu(w)|_{\pi(H)}$.
Otherwise, it is the smallest leaf of the next subtree rooted at that depth, which can be obtained via {\sc Tighten}.
The overall running time bound then follows from Theorem~\ref{thm:cramer}.
\end{proof}

\section*{Acknowledgment}
  \noindent The authors are grateful for the helpful comments and support by L\'{a}szl\'{o} V\'{e}gh.
  They would like to thank Xavier Allamigeon, Nathana\"{e}l Fijalkow and Marcin Jurdzi\'{n}ski for inspiring discussions.
  They are also thankful to the anonymous reviewers for their valuable comments, which have helped improve the presentation of this paper.

\bibliographystyle{alphaurl}
\bibliography{references}

\appendix

\section{Connection to tropical linear programming}
\label{sec:mpg}

In this section, we describe the well-known connection between parity games and tropical linear programming.
To this end, let us assume that $G$ is bipartite with bipartition $V_0\sqcup V_1$ without loss of generality.
Consider the tropical (min-plus) semiring $(\T,\oplus,\odot)$, where the set $\T=\R\cup \set{\infty}$ is equipped with binary operations $a\oplus b = \min\set{a,b}$ and $a\odot b = a + b$.
Let $n_1=\size{V_1}$ and $y\in \T^{n_1}$ be a vector of variables indexed over the nodes in $V_1$.
A parity game can be formulated as the following system of tropical linear inequalities
\begin{equation*}
  \label{sys_P}
  \tag{P}
    y_u + (-n)^{\pi(u)} \geq \min_{vw\in E} \set{y_w - (-n)^{\pi(v)}} \qquad \forall uv\in \delta^+(V_1).
\end{equation*}
This is precisely the reduction from parity games to mean payoff games \cite{journals/ijgt/EhrenfeuchtM79,thesis/Puri95}.
In particular, for each node $v\in V$, we assign a cost of $-(-n)^{\pi(v)}$ to every arc in $\delta^+(v)$.
A trivial feasible solution to \eqref{sys_P} is given by $\infty\cdot\1$.
However, we are interested in feasible solutions of maximal finite support due to the following statement. 

\begin{thmC}[\cite{AkianGG12}] \label{thm:feasible_mpg}
Let $y^*\in \T^{n_1}$ be a feasible solution to \eqref{sys_P} of maximal finite support. 
Then, Even wins from $u\in V_1$ if and only if $y^*_u< \infty$.
\end{thmC}

\end{document}